\documentclass[12pt, journal,  onecolumn]{IEEEtran}
\usepackage[margin=1in]{geometry}
\usepackage{setspace}
\usepackage{color}

\usepackage{graphicx,tabularx,array,amsmath,amsthm,thmtools,enumitem}
\usepackage{mathtools}

\usepackage{amsfonts}
\usepackage{bbm}
\usepackage{cite}
\usepackage{makecell}
\usepackage{multirow}
 \usepackage{amssymb}
\usepackage{txfonts}
\usepackage{amsthm} 
\usepackage{floatflt}
\usepackage[colorlinks]{hyperref}

\newcommand{\indep}{\rotatebox[origin=c]{90}{$\models$}}

\declaretheoremstyle[headfont=\bfseries, 
    bodyfont=\normalfont]{normalhead}
\declaretheorem[style=normalhead]{Example}

\newtheorem{Theorem}{Theorem}

\newtheorem{Proposition}{Proposition}
\newtheorem{Lemma}{Lemma}
\newtheorem{Corollary}{Corollary}
\newtheorem{Remark}{Remark}
\newtheorem{Definition}{Definition}

\allowdisplaybreaks % makes automatic splits of equations

\begin{document}

\title{A Concentration of Measure Approach to Correlated Graph Matching}
\author{Farhad Shirani$^1$\footnote{This work was supported in part by NSF grant CCF-1815821 and CNS-1619129, and  ND EPSCoR grant FAR0033968. This work was presented in part at IEEE International Symposium on Information Theory (ISIT), July 2018, 51st Asilomar Conference on Signals, Systems, and Computers, November 2017, and 56th Annual Allerton Conference on
Communication, Control, and Computing, October 2018.
\\$^1$ Farhad Shirani is with the Electrical and Computer Engineering Department, North Dakota State University, ND, USA, 58105, Email:
\href{mailto:f.shiranichaharsoogh@ndsu.edu}{f.shiranichaharsoogh@ndsu.edu}\\
$^2$
Siddharth Garg and Elza Erkip are with the Electrical and Computer Engineering Department,
New York University, NY, 11201, Email: {\{\href{mailto:siddharth.garg@nyu.edu}{siddharth.garg},\href{mailto:elza@nyu.edu}{elza}\}\href{mailto:elza@nyu.edu}{@nyu.edu}}.
}, Siddharth Garg$^2$, and Elza Erkip$^2$
 \\\date{} }

\maketitle \thispagestyle{empty}

\begin{abstract}
The graph matching problem emerges naturally in various applications such as {web} privacy, image processing and computational biology. In this paper{,} graph matching is considered under a stochastic model, where a pair of randomly generated graphs with pairwise correlated edges are to be matched {such that given the labeling of the vertices in the first graph, the labels in the second graph are recovered by leveraging the correlation among their edges.} The problem is considered under various settings and graph models. In the first step, the Correlated Erd\"{o}s-R\'enyi (CER) graph model is studied, where all edge pairs whose vertices have similar labels are generated based on identical distributions and independently of other edges. A matching scheme called the \textit{typicality matching scheme} is introduced. The scheme operates by investigating the joint typicality of the adjacency matrices of the two graphs. New results on the typicality of permutations of sequences lead to necessary and sufficient conditions for successful matching based on the parameters of the CER model. In the next step, the results are extended to {graphs with community structure generated based on the Stochastic Block Model (SBM)}. The {SBM} model is a generalization of the {CER} model where each vertex in the graph is associated with a community label, which affects its edge statistics.  The results are further extended to matching of ensembles of more than two correlated graphs. Lastly, the problem of seeded graph matching is investigated where a subset of the labels in the second graph are known prior to matching. In this scenario, in addition to obtaining necessary and sufficient conditions for successful matching, a polynomial time matching algorithm is proposed. 
%It is shown that successful matching is guaranteed when the number of seeds grows logarithmically in the number of graph vertices. The logarithmic coefficient is shown to be inversely proportional to the mutual information between the edge variables in the two graphs.
\end{abstract}

%\begin{IEEEkeywords}
%IEEEtran, journal, \LaTeX, paper, template.
%\end{IEEEkeywords}

\section{Introduction}

{Online social networks store large quantities of personal data from their users. As a result, social network privacy has become an issue of significant concern. 
Social network data is often released to third-parties in an anonymized and obfuscated form for various purposes including targeted advertising, developing new applications, and academic research \cite{gross2005information,shah2014collaborative}. However, it has been pointed out that anonymizing social network data through removing user IDs before publishing the data is far from enough to protect users’ privacy \cite{backstrom2007wherefore, ref2}. To elaborate, it has been shown through real-world implementation of privacy attacks that an attacker can potentially recover the user IDs by aligning the user profiles in the anonymized social network graph with the public profiles of users in other social networks on the web. In other words, the attacker can \textit{`match'} the anonymized social network profiles of users with their public profiles in other social networks.} {\textit{Graph Matching} --- also known as network alignment --- describes the problem of detecting node correspondence across graphs. In addition to social network deanonymization \cite{Beyah,Grossglauser,shirani2018typicality}, the need for matching two or more graphs arises naturally in a variety of other applications of interest such as pattern recognition \cite{foggia2014graph}, cross-lingual knowledge alignment \cite{xu2019cross}, and protein interaction network alignment \cite{kazemi2016proper}.}  {The significant increase in the ability to store, share, and analyze large graphs has led to a growing need to develop
\textit{low complexity} algorithms for graph matching, and derive \textit{theoretical guarantees} for their success, that is, to study how and when is it possible to perform fast and efficient network alignment.} 
%{The need for matching two or more graphs arise} in a variety of applications including social network de-anonymization, pattern recognition, and computational biology  \cite{conte2004thirty, emmert2016fifty}. 

In {the simplest form of graph matching scenarios}, an agent is given a correlated pair of randomly generated graphs: i) an `anonymized' unlabeled graph, and ii) a `de-anonymized' labeled graph as shown in Figure \ref{Fig:1}. The objective is to leverage the correlation among the edges of the graphs to recover the canonical labeling of the vertices in the anonymized graph.
The fundamental limits of graph matching, i.e. characterizing the necessary and sufficient conditions on graph parameters for successful matching, has been considered under various probabilistic models capturing the correlation among the graph edges. In the \textit{Correlated Erd\"os-R\'enyi}  (CER) model the edges in the two graphs are pairwise correlated and are generated independently, based on identical distributions. More precisely, in this model, edges whose vertices are labeled identically are correlated through an arbitrary joint probability distribution and are generated independently of all other edges. In its simplest form --- where the edges of the two graphs are exactly equal --- graph matching is called \textit{graph isomorphism}. Tight necessary and sufficient conditions for successful matching in the graph isomorphism scenario were derived in \cite{ER,wright} and polynomial time algorithms were proposed in ~\cite{iso1,iso2,iso4}. The problem of matching non-identical pairs of CER graphs was studied in ~\cite{corr1,corr2,corr3,corr4, kia_2017, Lyzinski_2016, cullina2018partial} and conditions for successful matching were derived. 

The CER model assumes the existence of statistical correlation among edge pairs connecting matching vertices in the two graphs, where the correlation model is based on an identical distribution among all matching edge pairs.  Consequently, it does not model the community structure among the graph nodes which manifests in many applications \cite{community,fortunato2012community}. As an example, in social networks, users may be divided into communities based on various factors such as age-group, profession, and racial background. The users' community memberships affects the probability that they are connected with each other. A matching algorithm may potentially use the community membership information to enhance its performance. In order to take the users' community memberships into account, an extension to the CER model is considered which is called the {\textit{Stochastic Block Model} (SBM)} model. In this model, the edge probabilities depend on their corresponding vertices' community memberships. There have been several works studying both necessary and sufficient conditions for graph matching and the design of practical matching schemes under the {SBM model \cite{shirani2018matching,nilizadeh2014community,singhal2017significance,lyzinski2014seeded}}. However, characterizing tight necessary and sufficient conditions for successful matching and designing polynomial time algorithms which are reliable under these conditions  remains an open problem both in the CER and {SBM} settings.  

A further extension of the problem, called \textit{`seeded graph matching'} has also been investigated in the literature \cite{efe,kiavash,seed1,seed2,seed3,seed4,Asilomar,Grossglauser,mossel2019seeded,fishkind2019seeded}. Seeded graph matching models applications where the matching agent has access to additional side-information in the form of pre-matched \textit{seeds}. A seed vertex is one whose correct label in both graphs is known prior to the start of the matching process. One pertinent application of seeded graph matching is the de-anonymization of users over multiple social networks. Many web users are members of multiple online social networks such as Facebook, Twitter, Google+, LinkedIn, etc.  Each online network represents a subset of the users' ``real" ego-networks. Graph matching provides algorithms to de-anonymize the users by reconciling these online network graphs, that is, to identify all the accounts belonging to the same individual. In this context, the availability of seeds is justified by the fact that a small fraction of individuals explicitly link their accounts across multiple networks. In this case, these linked accounts can be used as seeds in the matching algorithm. It turns out, that in many cases, these connections may be leveraged to identify a very large fraction of the users in the network ~\cite{seed1,seed2,seed3,seed4,kiavash}. In parallel to the study of fundamental limits of graph matching described above, the design of practical low complexity matching algorithms has also been studied in \cite{lyzinski2017matchability,  zhang2016final, heimann2018regal}, where reliable matching of real-world networks with up to millions of nodes have been performed. 

 \begin{figure}
 \begin{center}
\includegraphics[draft=false, width=0.7\textwidth]{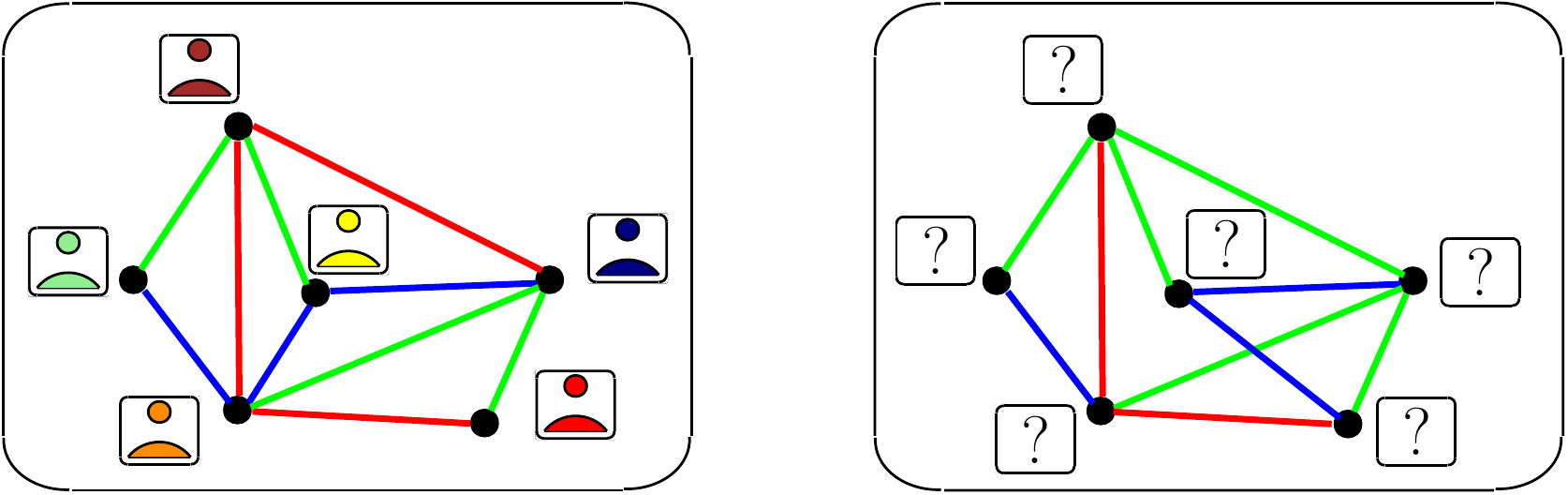}
\caption{An instance of the graph matching problem where the anonymized graph on the right is to be matched to the de-anonymized graph on the left.}
\label{Fig:1}
\end{center}
\end{figure}

In this work, we {construct an information theoretic framework based on} concentration of measure theorems in order to investigate the fundamental limits of graph matching {. We propose} the \textit{`typicality matching'} (TM) strategy which operates based on the concept of typicality of sequences of random variables \cite{csiszarbook}, and is applicable under a wide range of graph models including CER, {SBM} and seeded graph matching. 
The strategy considers the pair of adjacency matrices corresponding to the two graphs. Each $n\times n$ adjacency matrix may be viewed as an $n^2$-length sequence of random variables, where $n$ is the number of vertices in the graph. Consequently, one may naturally extend the notion of typicality of sequences of random variables to that of random adjacency matrices.
The TM strategy finds a labeling for the vertices in the anonymized graph which results in a pair of jointly typical adjacency matrices for the two graphs, where typicality is defined with respect to the underlying joint edge distribution. The success of the matching algorithm is investigated as the graph size grows asymptotically large. The matching algorithm is said to succeed if the fraction of correctly matched vertices approaches one as the number of vertices goes to infinity. As a result, the TM algorithm is successful as long as any labeling which leads to a pair of jointly typical adjacency matrices assigns an incorrect label to a negligible fraction of size $o(n)$ vertices in the anonymized graph\footnote{We write $f(x)=o(g(x))$ if $\lim_{x\to\infty}{\frac{f(x)}{g(x)}=0}$.}. In order to study the conditions for the success of the TM strategy, we derive several new bounds on the probability of joint typicality of permutations of sequences of random variables. The bounds may be of independent interest in other research areas as well. 

The generality of the information theoretic approach allows us to investigate matching under a wide range of statistical models. The  contributions of this work can be summarized as follows:\hspace{-0.1in}
%In addition to deriving new conditions for successful matching under the CER and {SBM} graph models which have been studied in prior works \cite{corr1,corr2,corr3,corr4, kia_2017, Lyzinski_2016, cullina2018partial,nilizadeh2014community,singhal2017significance,lyzinski2014seeded}, we also consider weighted graphs, where the graph edges are allowed to have non-binary attributes. We extend the results to the simultaneous matching of more than two graphs. Additionally, we derive converse results which provide necessary conditions for successful matching. Furthermore, we consider seeded graph matching and derive theoretical guarantees for successful matching as a function of the seed-set size and the parameters of the statistical model. In the case of seeded graph matching, we provide a matching algorithm whose complexity grows polynomially in the number of vertices. We further derive converse results by providing necessary conditions for successful matching as a function of the seed set size. To summarize, the main contributions of this work are as follows:
{\begin{itemize}
\item{We \textcolor{black}{build upon the ideas  in \cite{shirani2018typicality,shirani2018matching}} to develop a general framework based on TM which allows for derivation of necessary and sufficient conditions under which graph matching is possible in a wide range of statistical models. The framework is applicable in matching graphs with weighted edges as well as simultaneous matching of more than two graphs in seeded and seedless matching.}
\item{We apply the TM framework to graph matching under the CER, SBM and seeded graph matching models and to derive theoretical guarantees for successful matching.}
\item{We derive converse results which characterize conditions under which matching is not possible in the CER model as well as  simultaneous matching of more than two graphs.}
\item{We \textcolor{black}{investigate the approach proposed in \cite{Asilomar}, which builds upon the TM framework to propose} a polynomial time time matching algorithm for the seeded graph matching scenario.}
\end{itemize}}

The rest of the paper is organized as follows: Section \ref{sec:not} describes the notation. Section \ref{sec:form} provides the problem formulation. Section \ref{sec:perm} develops the necessary tools for analyzing the performance of the TM algorithm. Section \ref{sec:CER} studies matching under the CER model. Section \ref{sec:CS} considers the {SBM} model. Section \ref{sec:coll} investigates matching collections of more than two graphs. In Section \ref{sec:converse}, necessary conditions and converse results for matching of pairs of graphs are investigated. 
%Section \ref{sec:equi} shows the equivalence of a  well-known  criterion successful matching, which has been considered in the literature, with the one considered in this work. 
Section \ref{sec:SGM} studies seeded graph matching. Section \ref{sec:conc} concludes the paper.

\section{Notation}\label{sec:not}
 We represent random variables by capital letters such as $X, U$ and their realizations by small letters such as $x, u$. Sets are denoted by calligraphic letters such as $\mathcal{X}, \mathcal{U}$. The set of natural numbers, and the real numbers are represented by $\mathbb{N}$, and $\mathbb{R}$ respectively. The random variable $\mathbbm{1}_{\mathcal{E}}$ is the indicator function of the event $\mathcal{E}$.
 The set of numbers $\{n,n+1,\cdots, m\}, n,m\in \mathbb{N}$ is represented by $[n,m]$. Furthermore, for the interval $[1,m]$, we sometimes use the shorthand notation $[m]$ for brevity. 
 For a given $n\in \mathbb{N}$, the $n$-length vector $(x_1,x_2,\hdots, x_n)$ is written as $x^n$. {We write $a\stackrel{.}{=} b\pm \epsilon$ to denote $b-\epsilon \leq a \leq b+\epsilon$, where $a,b,\epsilon\in \mathbb{R}$. We define $|a|^+\triangleq max(0,a)$. \textcolor{black}{The notation $\exp_2(\alpha)$ is used to represent $2^{\alpha}$ to help readability. }}

\section{Problem Formulation}\label{sec:form}
 A graph $g=(\mathcal{V},\mathcal{E})$ is characterized by the vertex set $\mathcal{V}=\{v_1,v_2,\cdots,v_{n}\}$,
 and the edge set $\mathcal{E}$. We consider weighted graphs, where each edge is assigned an \textit{attribute} $x\in [0,l-1]$ and $l\geq 2$. Consequently, the edge set $\mathcal{E}$ is a subset of the set $\{(x,v_i,v_j)| i\neq j, x\in [0,l-1]\}$, where for each pair $(v_i,v_j)$ there is a unique attribute $x$ for which $(x,v_i,v_j)\in \mathcal{E}$. {For instance, an unlabeled graph with binary valued edges is a graph for which $l=2$. In this case, if the pair $v_{n,i}$ and $v_{n,j}$ are not connected, we write $(0,v_{n,i},v_{n,j})\in \mathcal{E}$, otherwise $(1,v_{n,i},v_{n,j})\in \mathcal{E}$. } The edge attribute models the nature of the connection between the corresponding vertices. For instance in social network graphs, where vertices represent the members of the network and edges capture their connections, an edge may take different attributes depending on whether the members are family members, close friends, or acquaintances. {A labeled graph $\tilde{g}=(g,\sigma)$ is a graph equipped with a bijective \textit{labeling function} $\sigma:\mathcal{V}\to [n]$. The labeling represents the identity of the members in the social network. For a labeled graph $\tilde{g}$, the adjacency matrix 
 %$G= [g_{i,j}]_{i,j\in [n]}$ , where $g_{i,j}$ is the unique value for which $(g_{i,j},v_i,v_j)\in \mathcal{E}$. More precisely, for the labeled graph $\tilde{g}$ the adjacency matrix 
 is defined as ${G}_{\sigma}=[{g}_{\sigma,i,j}]_{i,j\in [1,n]}$, where $g_{\sigma,i,j}$ is the unique value such that $(g_{\sigma,i,j},v_k,v_l)\in \mathcal{E}$, where $(v_k,v_l)=(\sigma^{-1}(i),\sigma^{-1}(j))$. The adjacency matrix captures the edge attributes of the graph. The upper triangle (UT) corresponding to $\tilde{g}$ is the structure $U_{\sigma}=[g_{\sigma,i,j}]_{i<j}$. The subscript `$\sigma$' is dropped when there is no ambiguity. The notation is summarized in Table I.}
 \setlength\arrayrulewidth{2pt}

 \begin{table}[]
  \centering
\begin{tabular}{|ll|ll|ll|}
\hline
 $g$:& unlabeled graph & $\mathcal{V}$:  & vertex set &  $\mathcal{E}$: & edge set    \\ 
 \hline
 $\sigma$: & labeling  & $\tilde{g}:$ & labeled graph & $G_{\sigma}$: & adjacency matrix   \\ \hline
$U_{\sigma}$: & upper-triangle & $\ell$: & $\#$ of attributes & $\hat{g}^2$: & relabeled graph   \\ \hline
\end{tabular}
\vspace{0.1in}
\caption{Notation Table: Random Graphs}
\vspace{-0.2in}
\end{table}
 %Labeled graphs are formalized as follows.
 %\begin{Definition}[\bf{Labeled Graphs}]
 %A labeling is a bijective function $\sigma: \mathcal{V}\to [1,n]$.  
%The pair $\tilde{g}=(g, \sigma)$ is called an $(n,l)$-labeled graph.
 %\label{Def:unlabeledER}
 % For a given vertex $v_{n,i}, i\in [1,n]$ the set of neighboring vertices is defined as $\mathcal{E}_{n,i}=\{v_{n,j}|(v_{n,i},v_{n,j})\in \mathcal{E}_n\}$. The degree of $v_{n,i}$ is defined as $d_{v_{n,i}}=|\mathcal{E}_{n,i}|$.
%\end{Definition}

%\begin{Remark}
% Without loss of generality, we assume that for any arbitrary pair of vertices $(v_{n,i},v_{n,j})$, there exists a unique $x\in [0,l-1]$ such that $(x,v_{n,i},v_{n,j})\in \mathcal{E}$.
%\end{Remark}

%\begin{Remark}
%We often consider sequences of graphs $g^{(n)}, n\in \mathbb{N}$, where $g^{(n)}$ has $n$ vertices. In such instances, we write  $g^{(n)}=(\mathcal{V}^{(n)},\mathcal{C}^{(n)},\mathcal{E}^{(n)})$ to characterize the $n$th graph in the sequence. The superscript `$(n)$' is omitted where there is no ambiguity.
%\end{Remark}

%Any pair of labeling functions are related through a permutation as described below.
%\begin{Definition}
%For labelings $\sigma,\sigma'$, the $(\sigma,\sigma')$-permutation $\pi_{(\sigma,\sigma')}$ is defined as $
%\pi_{(\sigma,\sigma')}(i)=j$ for $i,j\in [1,n]$, where ${\sigma'}^{-1}(j)=\sigma^{-1}(i)$.
%\label{def:above}
%\end{Definition}
 We consider graphs whose edges are generated stochastically. Under the CER and {SBM} models, we consider special instances of the following random graph model. 
\begin{Definition}[\bf Random Graph]
A random graph $\tilde{g}$ generated based on $\prod_{i\in[n], j<i}P_{X_{i,j}}$ is an undirected labeled graph, where the edge between $v_i, i\in [n]$ and $v_j, j<i$ is generated according to $P_{X_{\sigma(i),\sigma(j)}}$ independently of the other edges. Alternatively,
\begin{align*}
  P((x,v_i,v_j)\in \mathcal{E})= P_{X_{\sigma(i),\sigma(j)}}(x), x\in [0,l-1], i,j\in[n].  
\end{align*}
\end{Definition}
In the graph matching problem, we are given a pair correlated graphs $(\tilde{g}^1,\tilde{g}^2)$, where only the labeling for the vertices of the first graph is available. The objective is to recover the labeling of the vertices in the second graph by leveraging the correlation among their edges. A pair of correlated random graphs is defined below. 
\begin{Definition}[\bf Correlated Random Graph]
\label{def:cor_rand}
A pair of correlated random graphs $(\tilde{g}^1,\tilde{g}^2)$ generated based on $\prod_{i\in[n], j<i}P_{X^1_{i,j},X^2_{i,j}}$ is a pair of undirected labeled graphs. Let $v^1,w^1$ and $v^2,w^2$ be two pairs of  vertices {with the same label} in $\tilde{g}^1$ and $\tilde{g}^2$, respectively i.e. $\sigma^1(v^1)=\sigma^2(v^2)=s_1$ and $\sigma^1(w^1)=\sigma^2(w^2)=s_2$. Then, the pair of edges between $(v^1,w^1)$ and $(v^2,w^2)$ are generated according to $P_{X^1_{s_1,s_2},X^2_{s_1,s_2}}$. Alternatively,
\begin{align*}
  P((x^1,v^1_i,w^1_j)\in \mathcal{E}^1, (x^2,v^2_i,w^2_j)\in \mathcal{E}^2)= P_{X^1_{s_1,s_2},X^2_{s_1,s_2}}(x^1,x^2), {x^1,x^2}\in [0,l-1], i,j\in[n].  
\end{align*}
\end{Definition}
\begin{Remark}
In Definition \ref{def:cor_rand}, the pair $(\tilde{g}^1,\tilde{g}^2)$ are said to be a correlated pair of Erd\"os-R\'enyi (CPER) graphs if there exists a distribution $P_{X^1,X^2}$ such that  $P_{X^1_{s_1,s_2},X^2_{s_1,s_2}}= P_{X^1,X^2}, \forall s_1,s_2 \in [n]$.
\label{Rem:CPER}
\end{Remark}

A graph matching strategy takes $(\tilde{g}^1,g^2)$ as its input and outputs $(\tilde{g}^1,\hat{g}^2)$, where $g^2$ is the graph $\tilde{g}^2$ with its labels {$\sigma^2$} removed, and   $\hat{g}^2$ is the relabeled graph {after matching}. The matching strategy is said to succeed if the fraction of correctly matched vertices approaches one as the number of vertices is increased asymptotically. This is formalized below. 

\begin{Definition}[\bf Matching Strategy]
For a family of pairs of correlated random graphs $\tilde{g}^1_n=(g^1_n, \sigma^1_n)$  and $\tilde{g}^2_n=(g^2_n, \sigma^2_n), n \in \mathbb{N}$, {generated based on $\prod_{i\in[n], j<i}P_{X^1_{i,j},X^2_{i,j}}, n \in \mathbb{N}$} where $n$ is the number of vertices. 
A matching strategy is a sequence of functions $f_n: (\tilde{g}_n^1,g_n^2)\to (\tilde{g}_n^1,\hat{g}_n^2), n\in \mathbb{N}$, where $\hat{g}_n^2=(g^2_n,\hat{\sigma}^2_n)$ and $\hat{\sigma}^2_n$ is the reconstruction of $\sigma^2$. Let $I_n$ be distributed uniformly over $[n]$. The matching strategy is said to succeed if $P\left(\sigma^2(v^2_{I_n})=\hat{\sigma}^2(v^2_{I_n})\right)\to 1$ as $n\to\infty$.
\label{def:strat}
\end{Definition}

Note that in the above definition, for $f_n$ to succeed, the fraction of vertices whose labels are matched incorrectly must vanish as $n$ approaches infinity. This is a relaxation of  the criteria considered in \cite{corr1,corr2,corr3,corr4, kia_2017, Lyzinski_2016, Asilomar} where all of the vertices are required to be matched correctly simultaneously with vanishing probability of error as $n\to \infty$. As observed in the next sections, this relaxation leads to a significant simplification in the performance analysis of the proposed matching strategies and allows us to  use the concentration of measure theorems and results from information theory to derive theoretical guarantees on the performance of the TM strategy.

%The following defines an achievable region for the graph matching problem.
\begin{Definition}[\bf{Achievable Region}]
\hspace{-0.07in} A family of sets of distributions $\widetilde{P}\!=\!(\mathcal{P}_n)_{n\in \mathbb{N}}$ is in the achievable region if for every sequence of distributions $\prod_{ s_1\in [n],s_2<s_1}P^{(n)}_{X^1_{s_1,s_2},X^2_{s_1,s_2}}\!\!\!\in \mathcal{P}_n$, there exists a successful matching strategy. The maximal achievable family of sets of distributions is denoted by $\mathcal{P}^*$. 
 \label{def:ach}
 \end{Definition} 

In social network deanonymization, among other applications, often the correct label of a fraction of the vertices in the anonymized graph are known beforehand. This is due to a fraction of members having used the same user IDs across graphs, or having linked their accounts externally. In these scenarios, the matching strategy may use these pre-matched vertices as `seeds' to recover the labels for the rest of the vertices. Such matching strategies, which are called seeded matching strategies, are defined rigorously and studied in Section \ref{sec:SGM}.

%\begin{Remark}
%Definitions \ref{def:strat} and \ref{def:ach} can be extended in a natural way to scenarios involving simultaneous matching of multiple graphs. 
%\end{Remark}

\section{Permutations of Typical Sequences}
\label{sec:perm}
In the previous section, we described correlated pairs of random graphs, where the graph edges are generated randomly based on an underlying joint distribution.
Alternatively, the adjacency matrices of the graphs are generated according to a joint distribution. Furthermore, as explained in Definition \ref{def:cor_rand}, we assume that each edge pair connecting two similarly labeled vertices in the two graphs is generated independently of all other edges based on the distribution
$P_{X^1_{i,j},X^2_{i,j}}$, where $i,j$ are the vertex labels.
 Consequently, it is expected, given large enough graph sizes, that the adjacency matrices of the graphs look
`\textit{typical}' with respect to the joint edge distribution.
Roughly speaking, this requires the frequency of joint occurrence of symbols $(x^1,x^2)$ to be close to $\frac{1}{n^2}\sum^n_{i=1}\sum_{j=1}^n P_{X^1_{i,j},X^{2}_{i,j}}(x^1,x^2)$, where $x^1,x^2\in [0,l-1]$. Based on this observation, in the next sections we propose the typicality matching strategy which operates by finding the labeling for the second graph which results in a jointly typical pair of adjacency matrices. 
\begin{floatingfigure}[r]{2.4in}
\begin{center}
\includegraphics[height= 0.7in, width=2.4in]{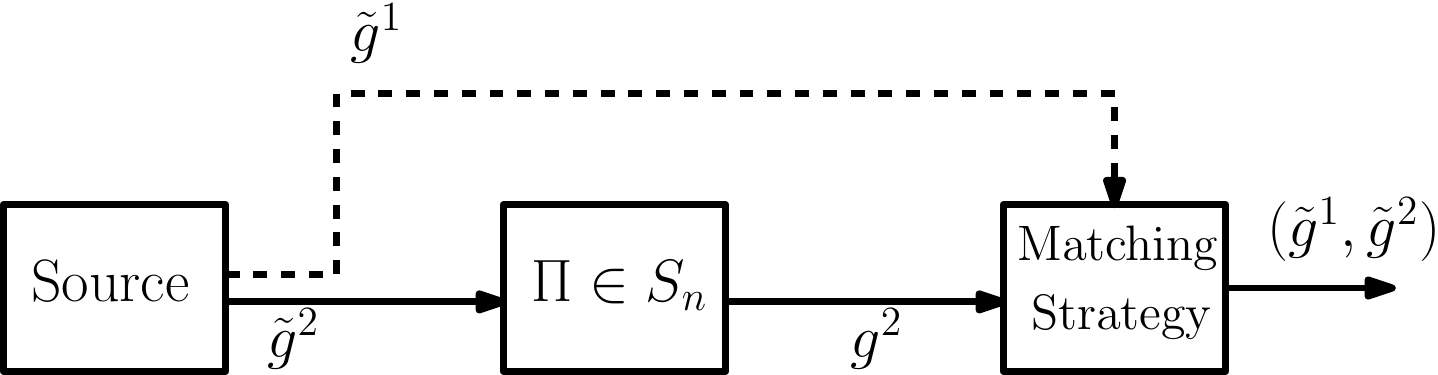}
\caption{The pair of correlated graphs $(\tilde{g}^1,\tilde{g}^2)$ are generated as described in Definition \ref{def:cor_rand}. The labels in $\tilde{g}^2$ undergo a random permutation $\Pi$ chosen uniformly among the set of all possible permutations of  n-length sequences $S_n$. The matching strategy uses $\tilde{g}^1$ as side information to recover $\tilde{g}^2$ from $g^2$.}
\label{fig:source}
\end{center}
\end{floatingfigure}
This is analogous to typicality decoding in the channel coding problem in information theory, where the decoder finds the transmitted sequence by searching for a codeword which is jointly typical with the received sequence. In this analogy which is shown in Figure \ref{fig:source}, the labeled graph $\tilde{g}^2$ is passed through a
`\textit{channel}' which outputs the graph $g^2$ whose labels have undergone a randomly and uniformly chosen permutation, and the matching algorithm acting as a `\textit{decoder}' wants to recover  $\tilde{g}^2$ using $g^2$ and the side-information $\tilde{g}^1$. Changing the labeling of  $g^2$ leads to a permutation of its adjacency matrix. Hence, we need to search over permutations of the adjacency matrix and find the one which leads to a typical pair of adjacency matrices. The error analysis of the typicality matching strategy requires investigating the probability of joint typicality of permutations of pairs of correlated sequences.

In this section, we analyze the joint typicality of permutations of collections of correlated sequences of random variables. While the analysis is used in the subsequent sections to derive the necessary and sufficient conditions for successful matching in various graph matching scenarios, it may also be of independent interest in other research areas as well. 
 
We follow the notation used in \cite{isaacs} in our study of permutation groups summarized below. 
\begin{Definition}[\bf Set Permutation]
\label{def:perm1}
A permutation on the set of numbers $[1,n]$ is a bijection $\pi:[1,n]\to [1,n]$. The set of all permutations on the set of numbers $[1,n]$ is denoted by $S_n$. 
\end{Definition}

\begin{Definition}[\bf Cycle and Fixed Point]
 A permutation $\pi \in \mathcal{S}_n, n\in \mathbb{N}$  is called a cycle if there exists $k\in [1,n]$ and $\alpha_1,\alpha_2,\cdots,\alpha_k\in [1,n]$ such that i) $\pi(\alpha_i)=\alpha_{i+1}, i\in [1,k-1]$, ii) $\pi(\alpha_n)=\alpha_1$, and iii) $\pi(\beta)=\beta$ if $\beta\neq \alpha_i, \forall i\in [1,k]$. The variable $k$ is the length of the cycle. The element $\beta$ is a fixed point of the permutation if $\pi(\beta)=\beta$. We write $\pi=(\alpha_1,\alpha_2,\cdots,\alpha_k)$.
 The cycle $\pi$ is non-trivial if $k\geq 2$. 
\end{Definition}

\begin{Lemma}[\hspace{-.005in}\bf\cite{isaacs}]
 Every permutation $\pi \in \mathcal{S}_n, n\in \mathbb{N}$ has a unique decomposition into disjoint non-trivial cycles.  
\end{Lemma}

\begin{Definition}
For a given $n,m,c\in \mathbb{N}$, and $1\leq i_1\leq i_2\leq \cdots\leq i_c \leq n$ such that $n=\sum_{j=1}^ci_j+m$, an $(m,c,i_1,i_2,\cdots,i_c)$-permutation is a permutation in $\mathcal{S}_n$ which has $m$ fixed points and $c$ disjoint cycles with lengths $i_1,i_2,\cdots,i_c$, respectively.
\label{def:cycle}
\end{Definition}

\begin{Example}
Consider the permutation which maps the vector $(1,2,3,4,5)$ to $(5,1,4,3,2)$. The permutation can be written as a decomposition of disjoint cycles in the following way $\pi=(1,2,5)(3,4)$, where $(1,2,5)$ and $(3,4)$ are cycles with lengths $3$ and $2$, respectively. The permutation $\pi$ is a $(0,2,2,3)$-permutation.
\end{Example}

\begin{Definition}[\bf Sequence Permutation]
\label{def:perm2}
 For a given sequence $y^n\in \mathbb{R}^n$ and permutation $\pi\in \mathcal{S}_n$, the sequence $z^n=\pi(y^n)$ is defined as $z^n=(y_{\pi(i)})_{i\in [1,n]}$.\footnote{Note that in Definitions \ref{def:perm1} and \ref{def:perm2} we have used  $\pi$ to denote both a scalar function which operates on the set $[1,n]$ as well as a function which operates on the vector space $\mathbb{R}^n$.}
\end{Definition}
\begin{Definition}[\bf Derangement]
 A permutation on vectors of length $n$ is a derangement if it has no fixed points. The number of distinct derangements of $n$-length vectors is denoted by $!n$.
\end{Definition}
In our analysis, we make extensive use of the standard permutations defined below.
\begin{Definition}[\bf Standard Permutation]
Let $m,c,i_1,i_2,\cdots,i_c$ be as in Definition \ref{def:cycle}.
The $(m,c,i_1,$ $i_2,\cdots,i_c)$-standard permutation is defined as the $(m,c,i_1,i_2,\cdots,i_c)$-permutation consisting of the cycles $(\sum_{j=1}^{k-1}i_j+1,\sum_{j=1}^{k-1}i_j+2,\cdots,\sum_{j=1}^{k}i_j), k\in [1,c]$. Alternatively, the $(m,c,i_1,$ $i_2,\cdots,i_c)$-standard permutation is defined as:
\begin{align*}
&\pi=(1,2,\cdots,i_1)(i_1+1,i_1+2,\cdots,i_1+i_2)\cdots
\\&\qquad(\sum_{j=1}^{c-1}i_j+1,\sum_{j=1}^{c-1}i_j+2,\cdots,\sum_{j=1}^{c}i_j)(n-m+1)(n-m+2)\cdots (n).
\end{align*}
\label{def:stan_perm}
\end{Definition}
\begin{Example}
The $(2,2,3,2)$-standard permutation is a permutation which has $m=2$ fixed points and $c=2$ cycles. The first cycle has length $i_1=3$ and the second cycle has length $i_2=2$. It is a permutation 
on sequences of length $n=\sum_{j=1}^ci_j+m=3+2+2=7$. The permutation is given by $\pi= (1 2 3)(4 5)(6)(7)$. For an arbitrary sequence $\underline{\alpha}=(\alpha_1,\alpha_2,\cdots,\alpha_7)$, we have:
\begin{align*}
 \pi(\underline{\alpha})=(\alpha_3,\alpha_1,\alpha_2,\alpha_5,\alpha_4,\alpha_6,\alpha_7).
\end{align*}
\end{Example}

\subsection{Typicality of Permutations of Pairs of Correlated Sequences}
{\begin{Definition}[\textbf{Type of Sequences}]
For a sequence $x^n\in \mathcal{X}^n$, the corresponding type vector $\underline{t}= (\underline{t}(x))_{x\in \mathcal{X}}$ is defined as
$
    \underline{t}(x)= \frac{\sum_{i=1}^n\mathbbm{1}(x_i=x)}{n}, x\in \mathcal{X}
$. For the pair $(x^n,y^n)\in \mathcal{X}^n\times \mathcal{Y}^n$, the corresponding joint type $\underline{s}= (\underline{s}(x,y))_{x,y\in \mathcal{X}\times \mathcal{Y}}$ is defined as $    \underline{s}(x,y)= \frac{\sum_{i=1}^n\mathbbm{1}(x_i=x,y_i=y)}{n}, x,y \in \mathcal{X}\times \mathcal{Y}$.
\end{Definition}
\begin{Definition}[\bf Strong Typicality \cite{csiszarbook}]
\label{Def:typ}
Let the pair of random variables $(X,Y)$ be defined on the probability space $(\mathcal{X}\times\mathcal{Y},P_{X,Y})$, where $\mathcal{X}$ and $\mathcal{Y}$ are finite alphabets. The $\epsilon$-typical set of sequences of length $n$ with respect to $P_{X,Y}$ is defined as:
\begin{align*}
&\mathcal{A}_{\epsilon}^n(X,Y)=\Big\{(x^n,y^n): 
\underline{t}(x,y)\stackrel{\cdot}{=}P_{X,Y}(x,y)\pm \epsilon, \forall (x,y)\in \mathcal{X}\times\mathcal{Y}  ~\&~  \underline{t}(x,y)=0 \text{ if } P_{X,Y}(x,y)=0\Big\},
\end{align*}
where $\underline{t}$ is the joint type of $(x^n,y^n)$,  $\epsilon>0$, and $n\in \mathbb{N}$.  
\end{Definition} }
For a correlated pair of independent and identically distributed (i.i.d) sequences $(X^n,Y^n)$ and an arbitrary permutation $\pi\in S_n$, we are interested in bounding the probability $P((X^n,\pi(Y^n))\in \mathcal{A}_{\epsilon}^n(X,Y))$. The following proposition shows that in order to find bounds on the probability of joint typicality of permutations of correlated sequences, it suffices to study standard permutations. 

\begin{Proposition}
\label{prop:1}
 Let $(X^n,Y^n)$ be a pair of i.i.d sequences defined on finite alphabets. We have:
\\ i) For an arbitrary permutation $\pi\in \mathcal{S}_n$, 
 \begin{align*}
 P((\pi(X^n),\pi(Y^n))\in \mathcal{A}_{\epsilon}^n(X,Y))=P((X^n,Y^n)\in \mathcal{A}_{\epsilon}^n(X,Y)).
\end{align*}
ii)  Following the notation in Definition \ref{def:stan_perm}, let $\pi_1$ be an arbitrary $(m,c,i_1,i_2,\cdots,i_c)$-permutation  and let $\pi_2$ be the $(m,c,i_1,i_2,\cdots,i_c)$-standard permutation. Then, 
\begin{align*}
 P((X^n,\pi_1(Y^n))\in \mathcal{A}_{\epsilon}^n(X,Y))=P((X^n,\pi_2(Y^n))\in \mathcal{A}_{\epsilon}^n(X,Y)).
 \end{align*}
 iii) For arbitrary permutations $\pi_x,\pi_y\in \mathcal{S}_n$,  let $\pi$ be the standard permutation having the same number of cycles and cycle lengths as that of $\pi_x^{-1}(\pi_y)$. Then,
\begin{align*}
 P((\pi_x(X^n),\pi_y(Y^n))\in \mathcal{A}_{\epsilon}^n(X,Y))=P((X^n,\pi(Y^n))\in \mathcal{A}_{\epsilon}^n(X,Y)).
\end{align*}
{iv) For an arbitrary permutation $\pi\in \mathcal{S}_n$,
\begin{align*}
 P((X^n,\pi(Y^n))\in \mathcal{A}_{\epsilon}^n(X,Y))=P((X^n,\pi^{-1}(Y^n))\in \mathcal{A}_{\epsilon}^n(X,Y)).
\end{align*}}
\end{Proposition}
%Please refer to \cite{arxiv_matching_ISIT18}.
\begin{proof}
Part i) follows from the fact that permuting both $X^n$ and $Y^n$ by the same permutation does not change their joint type. For part ii), it is known that there exists a permutation $\pi$ such that $\pi(\pi_1)=\pi_2(\pi)$ \cite{isaacs}. Then the statement is proved using part i) as follows:
  \begin{align}
& \nonumber P\left(\left(X^n,\pi_1\left(Y^n\right)\right)\in \mathcal{A}_{\epsilon}^n\left(X,Y\right)\right)
= P\left(\left(\pi\left(X^n\right),\pi\left(\pi_1\left(Y^n\right)\right)\right)\in \mathcal{A}_{\epsilon}^n\left(X,Y\right)\right)
\\&\nonumber= P\left(\left(\pi\left(X^n\right),\pi_2\left(\pi\left(Y^n\right)\right)\right)\in \mathcal{A}_{\epsilon}^n\left(X,Y\right)\right)
\\&\label{eq:a}{=} P\left(\left(\widetilde{X}^n,\pi_2\left(\widetilde{Y}^n\right)\right)\in \mathcal{A}_{\epsilon}^n\left(X,Y\right)\right)
\\&\label{eq:b}{=} P\left(\left(X^n,\pi_2\left(Y^n\right)\right)\in \mathcal{A}_{\epsilon}^n\left(X,Y\right)\right),
 \end{align}
 where in \eqref{eq:a} we have defined $(\widetilde{X}^n,\widetilde{Y}^n)=(\pi(X^n),\pi(Y^n))$. and \eqref{eq:b} holds since $(\widetilde{X}^n,\widetilde{Y}^n)$ has the same distribution as $(X^n,Y^n)$. Part iii) follows directly from Parts i) and ii). {Part iv) follows from Part ii) by noting that the number and lengths of cycles in $\pi^{-1}$ is the same as that of $\pi$.}
\end{proof}

The following theorem provides an upper-bound on the probability of joint typicality of permutations of correlated sequences for an arbitrary permutation with $m\in [n]$ fixed points.  
{
\begin{Theorem}
Let $\epsilon\in [0, \frac{1}{2}\min_{x,y\in \mathcal{X}\times \mathcal{Y}}P_{X,Y}(x,y)]$, and consider $(X^n,Y^n)$ a pair of i.i.d sequences defined on finite alphabets $\mathcal{X}$ and $\mathcal{Y}$, respectively. For any permutation $\pi$ with $m\in [n]$ fixed points, the following holds:
\begin{align}
    &P((X^n,\pi(Y^n))\in \mathcal{A}_{\epsilon}^n(X,Y))\leq 2^{-n (E_{\alpha}-{\zeta_n}-\delta_\epsilon)},
\\& E_{\alpha}= \min_{\underline{t}'_X\in
\mathcal{P}}\frac{1}{2}\left((1-\alpha)D(\underline{t}'_X||P_X)+\alpha D(\underline{t}''_{X}||P_X)+ D(P_{X,Y}|| (1-\alpha) P_XP_{Y''}+\alpha P_{X,Y})\right),
\label{eq:perm_bound_1}
\end{align}
where $\alpha\triangleq\frac{m}{n}$, $\mathcal{P}\triangleq \{ \underline{t}_X\in \mathcal{P}_X|\forall x\in \mathcal{X}: \underline{t}_X(x)\in \frac{1}{1-\alpha}[P_X(x)-\alpha, P_X(x)]\}$, $\mathcal{P}_X$ is the probability simplex on the alphabet $\mathcal{X}$, $D(\cdot||\cdot)$ is the Kullback-Leibler divergence,  $\underline{t}''_X\triangleq \frac{1}{\alpha}(P_X-(1-\alpha) \underline{t}'_X)$, $P_{Y''}(\cdot)\triangleq \sum_{x\in \mathcal{X}} \underline{t}'_X(x) P_{Y|X}(\cdot|x)$,  $\zeta_{n}\triangleq \frac{3}{2}|\mathcal{X}|^2|\mathcal{Y}|\frac{\log{(n+1)}}{n}+ 6|\mathcal{X}||\mathcal{Y}|\frac{\log{(n+1)}}{n}$, and \[\delta_{\epsilon}\triangleq  \epsilon|\mathcal{X}||\mathcal{Y}|
  \big|\max_{x,y \in \mathcal{X}\times \mathcal{Y}:P_{X,Y}(x,y)\neq 0} \log{\frac{P_{X,Y}(x,y)}{
   \alpha P_{X,Y}(x,y)+(1-\alpha)P_X(x)P_Y(y)
  }}\big|+O(\epsilon).\] 
\label{th:1:improved}
\end{Theorem}
\vspace{-0.4in}
The proof is provided in Appendix \ref{Ap:th1:improved}. In the following, we describe an outline of the proof. Let us  define $\mathcal{A}$ as the set of fixed points of the permutation $\pi$. From the theorem statement, the set $\mathcal{A}$ includes $\alpha=\frac{m}{n}$ fraction of the indices $[n]$.  
Let $\underline{T}''_X$ be the type of the vector of $X_i, i\in \mathcal{A}$, and let $\underline{T}'_X$ be the type of the vector $X_i, i\notin \mathcal{A}$. A necessary condition for $(X^n,\pi(Y^n))$ to be jointly $\epsilon$-typical with respect to $P_{X,Y}$ is that $\underline{T}'_X=\underline{t}'_X$ and $\underline{T}''_X=\underline{t}''_X$ such that $(1-\alpha)\underline{t}'_X+\alpha \underline{t}''_X\stackrel{.}{=}P_X\pm \epsilon$. Since $X^n$ is an i.i.d sequence of variables, from standard information theoretic arguments, the probability of the event that $\underline{T}'_X=\underline{t}'_X$ decays exponentially in $n$ with exponent $\alpha D(\underline{t}'_X||P_X)$. Similarly, the probability that $\underline{T}''_X=\underline{t}''_X$ decays exponentially in $n$ with exponent $(1-\alpha)D(\underline{t}''_X||P_X)$. This justifies the term $(1-\alpha) D(\underline{t}'_X||P_X)+\alpha D(\underline{t}''_X||P_X)$ in the exponent $E_{\alpha}$ in Equation \eqref{eq:perm_bound_1}. Next, note that for $i\in \mathcal{A}$,
we have $\pi({Y}_i)=Y_i$. As a result, the joint distribution of each of the pairs  $(X_i,\pi({Y}_i)), i\in \mathcal{A}$ is $P_{X,Y}$. On the other hand, for indices $i\notin \mathcal{A}$, we have $\pi({Y}_i)= Y_{\pi(i)}$, where $\pi(i)\neq i$. So, the pair $(X_i,\pi({Y}_i)), i\notin\mathcal{A}$ is an independent pair of variables, where $X_i$ is generated based on $P_X$, and $\pi(Y_i)$ is generated based on $P_{Y|X}(\cdot| X_{\pi(i)})$. Note that given that $\underline{T}'_X=\underline{t}'_X$, the average distribution of $\pi(Y_i), i\notin \mathcal{A}$ is  $\frac{1}{n-|\mathcal{A}|}\sum_{i\notin \mathcal{A}} P_{Y_{\pi(i)}}=\sum_{x\in \mathcal{X}}\underline{t}'_X(x)P_{Y|X}(\cdot|x)= P_{Y''}$. Consequently, the average distribution of $(X_i,\pi({Y}_i)), i\notin\mathcal{A}$ is $P_XP_{Y''}$. As a result, the average distribution of $(X^n,\pi({Y}^n))$  is $(1-\alpha)P_{X}P_{Y''}+\alpha P_{X,Y}$. Hence, using standard information theoretic arguments, the probability that the pair $(X^n,\pi({Y}^n))$  is jointly $\epsilon$-typical with respect to $P_{X,Y}$ decays exponentially with exponent $D(P_{X,Y}||(1-\alpha)P_{X}P_{Y''}+\alpha P_{X,Y})$,
which appears as the third term in the exponent $E_{\alpha}$ in Equation \eqref{eq:perm_bound_1}.}
{The exponent $E_{\alpha}$ can be further simplified for special classes of permutations. For instance, if $\pi$ does not have any fixed points, the following corollary to Theorem \ref{th:1:improved} holds.}
{
\begin{Corollary}
If the permutation $\pi$ in Theorem \ref{th:1:improved} has no fixed points (i.e. $\alpha=0)$, then:
\begin{align}
    &P((X^n,\pi(Y^n))\in \mathcal{A}_{\epsilon}^n(X,Y))\leq 2^{-n (E_0-\zeta_{n}-\delta_{\epsilon})},
    \end{align}
    where, $E_0= \frac{1}{2}I(X;Y)$, and $
\zeta_{n}$ and $\delta_{\epsilon}$ are defined in Theorem \ref{th:1:improved}.
\label{th:corr}
\end{Corollary}}
{
Proof. The proof follows from Theorem \ref{th:1:improved}. Note that when $\alpha=0$, the set $\mathcal{P}$ in the theorem statement has a single element $P_X$. So, we have $\underline{t}'_X=P_X$, $P_{Y''}=P_Y$,  
$    E_0=
    \frac{1}{2}I(X;Y)
    ,$
and $\delta_{\epsilon}=\epsilon$ $|\mathcal{X}||\mathcal{Y}|$ $
  \big|\max_{x,y \in \mathcal{X}\times \mathcal{Y}:P_{X,Y}(x,y)\neq 0} \log{\frac{P_{X,Y}(x,y)}{
   P_X(x)P_Y(y)
  }}\big|$.}

{
 Corollary \ref{th:corr} shows that for two sequences $(X^n,Y^n)$ generated jointly according to $P_{X,Y}$ and a permutation $\pi$ without any fixed points, the probability that the pair $(X^n,\pi(Y^n))$ are jointly typical with respect to $P_{X,Y}$ decays exponentially in $n$ with exponent $
\frac{1}{2}I(X;Y)$. Note that since there are no fixed points in the permutation, each pair $(X_i,\pi(Y_i)), i\in [n]$ has joint distribution $P_XP_Y$. So, there is no `single-letter' correlation among the elements of $(X^n,\pi(Y^n))$. It is well-known that if the sequences $X^n$ and $Y^n$ are generated independently of each other according to marginal distributions $P_X$ and $P_Y$, respectively, then the probability that they are jointly typical with respect to the joint distribution $P_{X,Y}$ decays exponentially in $n$ with exponent $I(X;Y)$ (e.g. \cite{csiszarbook}).
The coefficient $\frac{1}{2}$ in exponent in Corollary \ref{th:corr} is an artifact of the $n$-letter correlation between $(X^n,\pi(Y^n))$ due to the fact that the original pair of sequences $(X^n,Y^n)$ are correlated. }

{Theorem \ref{th:1:improved} is used in the next sections to derive sufficient conditions under which pairs of correlated graphs can be matched successfully. However, the arguments in the proof of the theorem do not  extend naturally to typicality of collections of more than two permuted sequences. Bounds on the probability of joint typicality of such collections are necessary for evaluating graph matching for collections of graphs studied in Section \ref{sec:coll}. 
To this end, the following theorem and the ensuing corollary provide an alternative bound on the probability of joint typicality of $(X^n,\pi(Y^n))$ for an arbitrary permutation $\pi$ which is then extended to evaluate the typicality of collections of more than two sequences in Theorem \ref{th:cperm}. This is used in Section \ref{sec:coll} to evaluate  matching of more than two graphs. Additionally, we will observe in the proof of Theorem \ref{th:ER} in Section \ref{sec:CER} that  $E'_{\alpha}$ derived below yields tighter bounds on the probability of joint typicality of $(X^n,\pi(Y^n))$ for large $\alpha$ compared to $E_{\alpha}$ derived in Theorem \ref{th:1:improved}.}
{\begin{Theorem}
Let $(X^n,Y^n)$ be a pair of i.i.d sequences defined on finite alphabets $\mathcal{X}$ and $\mathcal{Y}$, respectively. For any permutation $\pi$ with $m\in [n]$ fixed points, the following holds:
 \begin{align}
  \label{eq:perm_bound}
  &P((X^n,\pi(Y^n))\in \mathcal{A}_{\epsilon}^n(X,Y)) 
     \leq 2^{-n(E'_{\alpha}-{\zeta'_n}-{\delta_{\epsilon}})},
\\&
   E'_{\alpha}=\min_{\underline{t}'_{X,Y}\in \mathcal{P}'} \left(\frac{1-\alpha}{3}\right)
   D(\underline{t}'_{X,Y}||P_{X}P_{Y})+
   \alpha D(\underline{t}''_{X,Y}||P_{X,Y}),
\end{align}
where $\alpha\triangleq \frac{m}{n}$, $\mathcal{P}'\triangleq \{ \underline{t}_{X,Y}\in \mathcal{P}_{X,Y}|\forall (x,y)\in \mathcal{X}\times \mathcal{Y}: \underline{t}_{X,Y}(x,y)\in \frac{1}{1-\alpha}[P_{X,Y}(x,y)-\alpha, P_{X,Y}(x,y)]\}$, $\mathcal{P}_{X,Y}$ is the probability simplex on the alphabet $\mathcal{X}\times \mathcal{Y}$, $\underline{t}''_{X,Y}\triangleq \frac{1}{\alpha}(P_{X,Y}-(1-\alpha) \underline{t}'_{X,Y})$,
$\zeta'_n=4|\mathcal{X}||\mathcal{Y}|\log{\frac{n+1}{n}}$, and $\delta_{\epsilon}$ is defined as in Theorem \ref{th:1:improved}. 
\label{th:1}
\end{Theorem}
\begin{proof}
Appendix \ref{app:th1}.
\end{proof}}
{
Computing  $E'_{\alpha}$ requires optimizing over $\underline{t}'_{X,Y}$, which is computationally challenging for large alphabets. The following removes the optimization and provides a lower bound for  $E'_{\alpha}$.\hspace{-.1in} } 
{\begin{Corollary}
Let $(X^n,Y^n)$ be a pair of i.i.d sequences defined on finite alphabets $\mathcal{X}$ and $\mathcal{Y}$, respectively. For any permutation $\pi$ with $m\in [n]$ fixed points, the following holds:
 \begin{align}
  &P((X^n,\pi(Y^n))\in \mathcal{A}_{\epsilon}^n(X,Y)) 
 \leq 2^{-n(\widehat{E}'_{\alpha}-\zeta'_n-\frac{\delta_{\epsilon}}{3})},
  \\& \widehat{E}'_{\alpha}= \frac{1}{3}D(P_{X,Y}
 ||(1-\alpha)P_XP_Y+ \alpha P_{X,Y}),
\end{align}
where $\alpha\triangleq \frac{m}{n}$,
$\zeta'_n\triangleq 4|\mathcal{X}||\mathcal{Y}|\log{\frac{n+1}{n}}$ and $\delta_{\epsilon}$ is defined as in Theorem \ref{th:1:improved}. 
\label{th:cor:1}
\end{Corollary}
\begin{proof}
Appendix \ref{app:cor:th1}.
\end{proof}}
{
It is desirable to find the largest exponent which can be used to bound the exponential decay in the probability of joint typicality of $(X^n,\pi(Y^n))$. Consequently, a question of interest is whether one of the two exponents $E_{\alpha}$ and $E'_{\alpha}$ is strictly larger than the other. Towards such a comparison,
the next lemma shows that the relation $\frac{2}{3}E_{\alpha}\leq  \widehat{E}'_{\alpha}\leq E'_{\alpha} $ holds. On the other hand, it can be shown through analytical evaluations of the bounds under specific distributions $P_{X,Y}$ that in some instances, the relation $E'_{\alpha}< E_{\alpha}$ holds. We will observe in the proof of Theorem \ref{th:ER} in Section \ref{sec:CER} that $E_{\alpha}$ in Theorem \ref{th:1:improved} yields tighter bounds on the probability of joint typicality when $\alpha$ is small, whereas $E'_\alpha$ in Theorem \ref{th:1} is useful when evaluating  permutations with large $\alpha$.
\begin{Lemma}
For the exponents $E_{\alpha}$, $E'_{\alpha}$, and $\widehat{E}'_{\alpha}$ in Theorems \ref{th:1:improved} and \ref{th:1} and Corollary \ref{th:cor:1},  we have:
\begin{align*}
    \frac{2}{3}E_{\alpha}\leq  \widehat{E}'_{\alpha}\leq E'_{\alpha}
\end{align*}
\label{lem:exp}
\end{Lemma}
\vspace{-0.6in}
\begin{proof}
The relation $\widehat{E}'_{\alpha}\leq E'_{\alpha}$ follows by convexity of KL divergence. 
Also, note that
\begin{align*}
    \frac{2}{3}E_{\alpha}&= \frac{2}{3} \min_{\underline{t}'_X\in
\mathcal{P}}\frac{1}{2}\left((1-\alpha)D(\underline{t}'_X||P_X)+\alpha D(\underline{t}''_{X}||P_X)+ D(P_{X,Y}|| (1-\alpha) P_XP_{Y''}+\alpha P_{X,Y})\right)
\\&\leq \frac{1}{3}D(P_{X,Y}|| (1-\alpha) P_XP_{Y}+\alpha P_{X,Y})=\widehat{E}'_{\alpha}, 
\end{align*}
where the inequality follows by taking  $\underline{t}'_X= P_X$. \textcolor{black}{Note that this leads to $\underline{t}''_X= P_X$, so  that $D(\underline{t}'_X||P_X)=D(\underline{t}''_X||P_X)=0$.} 
\end{proof}}
{We have provided bounds on the probability of joint typicality of $X^n$ and $\pi(Y^n)$ as a function of the number of fixed points $m$ of the permutation $\pi(\cdot)$. Such bounds are often used in error analysis and derivation of error bounds in various applications \cite{tuncel2005error,shirani2018typicality,csiszar1998method}, and we will use them in the following sections to evaluate the probability of error in graph matching scenarios. 
%The standard method in such analysis is to use a union bound on the error probability and break the error event into a set of components each pertaining to the  joint typicality of the pair of vectors $(X^n,\pi(Y^n))$ for specific $\pi\in \mathcal{S}_n$. Then, an upper-bound on the probability of error is derived by counting the number of terms $(X^n,\pi(Y^n))$ for which $P((X^n,\pi(Y^n))\in \mathcal{A}_{\epsilon}^n(X,Y))$ is equal to each other and multiplying the total number of terms by that probability. }
%Then, the error terms  whose probability is `close' to each other are grouped together, and their total probability is approximated by multiplying the  total number of such terms with the probability of each of them. For instance, assume that in a specific application (e.g. graph matching), the error event is bounded using the union bound as $P_e\leq \sum_{\pi \in \mathcal{P}}P((X^n,\pi(Y^n))\in \mathcal{A}_{\epsilon}^n(X,Y))$, where $\mathcal{P}$ is a set of permutations. Then, one method to evaluate the right hand side of the inequality is to partition $\mathcal{P}$ into subsets of permutations $\mathcal{P}_i, i \in [n]$, where any permutation in $\mathcal{P}_i$ has precisely $i$ fixed points. 
%{As shown in Theorems \ref{th:1:improved} and \ref{th:1}, the decay exponent of the quantity $P((X^n,\pi(Y^n))\in \mathcal{A}_{\epsilon}^n(X,Y))$ can be bounded as a function of the  number of fixed points of the permutation $\pi$. 
%As a result, 
In order to evaluate the error exponents, the following results on the limiting behavior of the number of distinct permutations with a given number of fixed points are needed. }
\begin{Lemma}
\label{lem:dercount}
Let $n\in \mathbb{N}$. Let $N_{m}$ be the number of distinct permutations with exactly $m\in [0,n]$ fixed points. Then, 
\begin{align}
    \frac{n!}{m!(n-m)}\leq  N_m= {n \choose m} !(n-m)\leq n^{n-m}.
    \label{eq:der1}
\end{align}
Particularly, let $m= \alpha n, 0<\alpha<1$. Then, the following holds:
\begin{align}
    \lim_{n\to \infty} \frac{\log{N_m}}{n\log{n}}= 1-\alpha.
    \label{eq:der2}
\end{align}
\end{Lemma}
\begin{proof}
Appendix \ref{app:lem2}.
\end{proof}

\subsection{Typicality of Permutations of Collections of Correlated Sequences}
We consider joint typicality of permutations of more than two correlated sequences $(X^n_{(1)},X^n_{(2)},\cdots,$ $X^n_{(k)}), n\in \mathbb{N}, k>2$.  The derivations in this section are used in Section \ref{sec:coll} to extend the analysis of the TM strategy to simultaneous matching of collections of more than two graphs.  
\begin{Definition}[\bf Strong Typicality of Collections of Sequences \cite{csiszarbook}]
\label{Def:ctyp}
{Let the random vector $X^k$ be defined on the probability space $(\prod_{j\in [k]}\mathcal{X}_j,P_{X^k})$, where $\mathcal{X}_j, j\in [k]$ are finite alphabets, and $k>2$. The $\epsilon$-typical set of sequences of length $n$ with respect to $P_{X^k}$ is defined as:
\begin{align*}
    &\mathcal{A}_{\epsilon}^n(X^k)=\Big\{(x_{(j)}^n)_{j\in [k]}: \underline{t}(\alpha^k)\stackrel{.}{=}P_{X^k}(\alpha^k)\pm \epsilon, \forall \alpha^k\in \prod_{j\in [k]}\mathcal{X}_j\Big\},
\end{align*}
where $\epsilon>0$, and $\underline{t}(\alpha^k)= \frac{1}{n}\sum_{i=1}^n \mathbbm{1}\big((x_{(j),i})_{j\in [k]}=$ $\alpha^k\big)$ is the type of $(x_{(j)}^n)_{j\in [k]}$. } 
\end{Definition} 
In the previous section, in order to investigate the typicality of permutations of pairs of correlated sequences, we introduced standard permutations which are completely characterized by the number of fixed points, number of cycles, and cycle lengths of the permutation. The concept of standard permutations does not extend naturally when there are more than two sequences (i.e. more than one non-trivial permutation). Consequently, investigating typicality of permutations of collections of sequences requires developing additional analytical tools described next. 

\begin{Definition}[\bf Bell Number \cite{comtet2012advanced}]
 Let $\mathsf{P}=\{\mathcal{P}_1,\mathcal{P}_2,\cdots, \mathcal{P}_{b_k}\}$ be the set of all partitions of $[1,k]$.  The natural number $b_k$ is the $k$'th Bell number. {We make the convention that that $\mathcal{P}_{b_k}= \{[n]\}$.}
\end{Definition}
In the following, we define Bell permutation vectors which are analogous to standard permutations for the case when the problem involves more than one non-trivial permutation.

\begin{Definition}[\bf Partition Correspondence]
\label{Def:corr}
Let $k,n\in \mathbb{N}$ and $(\pi_1,\pi_2,\cdots,\pi_k)$ be arbitrary permutations operating on $n$-length vectors.
The index $i\in [1,n]$ is said to correspond to the partition $\mathcal{P}_j\in \mathsf{P}$ of the set $[1,k]$ if the following holds: 
\begin{align*}
    \forall l,l'\in [1,k]: \pi^{-1}_l(i)=\pi^{-1}_{l'}(i) \iff \exists r: l,l' \in \mathcal{D}_{j,r},
\end{align*}
where $\mathcal{P}_j=\{\mathcal{D}_{j,1},\mathcal{D}_{j,2},\cdots,\mathcal{D}_{j,|\mathcal{P}_j|}\}$.
\end{Definition}

\begin{Example}
Let us consider a triple of permutations of $n$-length sequences, i.e. $k=3$, and the partition $\mathcal{P}= \{\{1,2\},\{3\}\}$. An index $i\in [n]$ corresponds to $\mathcal{P}$ if the first two permutations map the index to the same integer and the third permutation maps the index to a different integer. 
\end{Example} 

\begin{Definition}[\bf Bell Permutation Vector]
Let $(i_1,i_2,\cdots,i_{b_k})$ be an arbitrary sequence, where $\sum_{k\in [b_k]}{i_k}=n, i_k\in [0,n]$, $b_k$ is the $k$th Bell number, and $n,k \in \mathbb{N}$.
 The vector of permutations $(\pi_1,\pi_2,\cdots, \pi_k)$ is called an $(i_1,i_2,\cdots,i_{b_k})$-Bell permutation vector if for every partition $\mathcal{P}_k$ exactly $i_k$ indices correspond to that partition. Equivalently:
\begin{align*}
    &\forall j\in [b_k]: i_k=|\{i\in [n]: \forall l,l'\in [k]: \pi^{-1}_l(i)=\pi^{-1}_{l'}(i) \iff \exists r\in [|\mathcal{P}_j|]: l,l' \in \mathcal{D}_{j,r}\}|,
\end{align*}
where $\mathcal{P}_j=\{\mathcal{D}_{j,1},\mathcal{D}_{j,2},\cdots,\mathcal{D}_{j,|\mathcal{P}_j|}\}$
\label{def:bell}.
\end{Definition}

The definition of Bell permutation vectors is further clarified through the following example.

\begin{Example}
Consider three permutations $(\pi_1,\pi_2,\pi_3)$ of vectors with length seven, i.e. $k=3$ and $n=7$. Then, $b_k=5$ and we have:
\begin{align*}
    &\mathcal{P}_1=\{\{1\},\{2\},\{3\}\}, 
    \quad 
    \mathcal{P}_2=\{\{1,2\},\{3\}\},
    \quad 
    \mathcal{P}_3=\{\{1,3\},\{2\}\},
    \\
    &\qquad\qquad\qquad \mathcal{P}_4=\{\{1\},\{2,3\}\}, 
    \quad 
    \mathcal{P}_5=\{\{1,2,3\}\}.
\end{align*}
Let $\pi_1$ be the identity permutation, $\pi_2= (1 3 5)(2 4)$, and $\pi_3= (1 5)(2 4)(3 7)$. Then:
\begin{align*}
    & \pi_1((1,2,\cdots,7))=(1,2,3,4,5,6,7), \qquad\pi_2((1,2,\cdots,7))=(5,4,1,2,3,6,7),
    \\
     & \pi_3((1,2,\cdots,7))=(5,4,7,2,1,6,3),
\end{align*}
 The vector $(\pi_1,\pi_2,\pi_3)$ is a $(2,1,0,3,1)$-Bell permutation vector, where the indices $(3,5)$ correspond to the $\mathcal{P}_1$ partition (each of the three permutations map indices $(3,5)$ to a different integer), index $7$ corresponds to the $\mathcal{P}_2$ partition (the first two permutations map the index 7 to the same integer which is different from the one for the third permutation), indices $(1,2,4)$  correspond to the $\mathcal{P}_4$ permutation (the second and third permutations map the indices $(1,2,4)$ to the same integer which is different from the output of the first permutation),
and index $6$ corresponds to $\mathcal{P}_5$ (all permutations map the index 6 to the same integer). None of the indices corresponds to $\mathcal{P}_3$ since there is no index which is mapped to the same integer by the first and third permutations and a different integer by the second permutation.
\end{Example}

\begin{Remark}
Bell permutation vectors are not unique. There can be several distinct $(i_1,i_2,\cdots, i_{b_k})$-Bell permutation vectors for given $n,k,i_1,i_2,\cdots,i_{b_k}$.
This is in contrast with standard permutations defined in Definition \ref{def:stan_perm}, which are unique given the parameters $n,k,c,i_1,i_2,\cdots,i_c$. 
\end{Remark}

The following theorem provides bounds on the probability of joint typicality of permutations of collections of correlated sequences:

\begin{Theorem}
\label{th:2}
{Let $(X_{(j)}^n)_{j\in [k]}$ be a collection of correlated sequences of i.i.d random variables defined on finite alphabets $\mathcal{X}_{j}, j\in [k]$. For any $(i_1,i_2,\cdots,i_{b_k})$-Bell permutation vector $(\pi_1,\pi_2,\cdots,\pi_k)$, the following holds:}
 \begin{align}
  \label{eq:cperm_bound}
 &{P((\pi_i(X_{(j)} ^n)_{j\in [k]}\in \mathcal{A}_{\epsilon}^n(X^k))
  \leq 2^{-n(E_{i_1,i_2,\cdots,i_{b_k}}+O(\epsilon)+O(\frac{\log{n}}{m}))},}
  \\&
  { E_{i_1,i_2,\cdots,i_{b_k}}= -\frac{1}{(k(k-1)+1)(b_k-1)}D(P_{X^k}
 ||\sum_{j\in [b_k]}\frac{i_j}{n}
 P_{X_{\mathcal{P}_j}})}
\end{align}
{where $P_{X_{\mathcal{P}_{j}}}=\prod_{r\in [1,|\mathcal{P}_j|]}P_{X_{i_1},X_{i_2},\cdots,X_{i_{|\mathcal{D}_{j,r}|}}}$, $\mathcal{D}_{j,r}=\{i_1,i_2,\cdots, i_{|\mathcal{D}_{j,r}|}\}, j\in [b_k], r\in [1,|\mathcal{P}_j|]$. }
\label{th:cperm}
\end{Theorem}
\begin{proof}
Appendix \ref{app:thcperm}.
\end{proof}
{Note that for the special case of permutations of pairs of sequences of random variables, $k=2$, the second Bell number is $b_k=2$. In this case $(k(k-1)+1)(b_k-1)=3$, and the bound on the probability of joint typicality given in Theorem \ref{th:cperm} is the same as the one in Corollary \ref{th:cor:1}. }

In the following, we generalize Lemma \ref{lem:dercount} to the case where a collection of more than two permuted sequence is considered, and provide upper and lower bounds on the number of distinct Bell permutation vectors for a given vector $(i_1,i_2,\cdots,i_{b_k})$.

\begin{Definition}[\bf k-fold Derangement]
A vector $(\pi_1(\cdot),\pi_2(\cdot),\cdots, \pi_{k}(\cdot))$ of permutations of $n$-length sequences is called an k-fold derangement if  $\pi_1(\cdot)$ is the identity permutation, and $\pi_l(i)\neq \pi_{l'}(i), l,l'\in [k], l\neq l', i\in [n]$. The number of distinct k-fold derangements of $[n]$ is denoted by $d_k(n)$. Particularly $d_{2}(n)=!n$ is the number of derangements of $[n]$.  
\end{Definition}

\begin{Lemma}
\label{lem:foldcount}
Let $n\in \mathbb{N}$ and $k\in [n]$. Then,
\begin{align*}
  ((n-k+1)!)^{k-1}  \leq d_k(n)\leq (!n)^{k-1}.
\end{align*}
\end{Lemma}
\begin{proof}
Appendix \ref{Ap:lem:foldcount}.
\end{proof}

\begin{Lemma}
\label{lem:bell}
Let $(i_1,i_2,\cdots, i_{b_k})$ be a vector of non-negative integers such that $\sum_{j\in [b_k]}i_j=n$. Define $N_{i_1,i_2,\cdots,i_{b_k}}$ as the number of distinct $(i_1,i_2,\cdots, i_{b_k})$-Bell permutation vectors. Then, 
\begin{align}
   &{n \choose i_1,i_2,\cdots, i_{b_k}}\prod_{j\in [b_k]}d_{|\mathcal{P}_{j}|}(i_j)
   \leq N_{i_1,i_2,\cdots,i_{b_k}} \leq {n \choose i_1,i_2,\cdots, i_{b_k}} n^{\sum_{j\in [b_k]} |\mathcal{P}_j|i_j-n}.
    \label{eq:bell1}
\end{align}
Particularly, let $i_k=\alpha_k \cdot n, n\in \mathbb{N}$. 
The following holds:
\begin{align}
    \lim_{n\to\infty} \frac{\log{N_{i_1,i_2,\cdots,i_{b_k}}}}{n\log{n}}= \sum_{j\in [b_k]}{|\mathcal{P}_j|}{\alpha_j}-1. 
    \label{eq:bell2}
\end{align}
\end{Lemma}
\begin{proof}
Appendix \ref{app:lem6}.
\end{proof}

\section{Matching Erd\"{o}s-R\`{e}nyi Graphs}
\label{sec:CER}
In this section, we consider matching of CPER graphs with weighted edges. In Section \ref{sec:form}, we described correlated random graphs. A CPER is a special instance of the correlated random graphs defined in Definition \ref{def:cor_rand}. We propose the typicality matching strategy and provide sufficient conditions on the joint edge statistics under which the strategy succeeds.

\subsection{The Typicality Matching Strategy for CPERs}
 Given a correlated pair of graphs $(\tilde{g}^1,{g}^2)$, where only the labeling for $\tilde{g}^1$ is given, the TM strategy operates as follows. The scheme finds a labeling $\hat{\sigma}^2$, for which the pair of UT's $U^{1}_{{\sigma}^1}$ and $U^{2}_{\hat{\sigma}^2}$ are jointly typical with respect to $P^{(n)}_{X_1,X_2}$  when viewed as vectors of length $\frac{n(n-1)}{2}$. The strategy  fails if no such labeling exists. Alternatively, it finds an element  $\hat{\sigma}^2$ in the set:
\begin{align}
 \widehat{\Sigma}=\{\hat{\sigma}^2| (U^{1}_{{\sigma}^1},U^{2}_{\hat{\sigma}^2})\in \mathcal{A}_{\epsilon}^{\frac{n(n-1)}{2}}(X_1,X_2)\},
 \label{eq:sigma}
\end{align}
where $\epsilon=\omega(\frac{1}{n})$. The algorithm declares $\hat{\sigma}^2$ as the correct labeling. Note that the set $ \widehat{\Sigma}$ may have more than one element. In that case, the strategy chooses one of these elements as the output randomly. We will show that under certain conditions on the joint graph statistics, all of the elements of $ \widehat{\Sigma}$ satisfy the criteria for successful matching given in Definition \ref{def:strat}. In other words, for all of the elements of $ \widehat{\Sigma}$  the probability of incorrect labeling for any given vertex is arbitrarily small for large $n$. {Formally, The TM strategy is a sequence of functions $f_n: (\tilde{g}_n^1,g_n^2)\to (\tilde{g}_n^1,\hat{g}_n^2), n\in \mathbb{N}$, where for any given $n\in \mathbb{N}$, the labeling  $\hat{\sigma}_n^2$ of $\hat{g}_n^2$ is chosen randomly and uniformly from the set $\widehat{\Sigma}$ defined in Equation \eqref{eq:sigma}.}
%The TM strategy is formally defined below:
%\begin{Definition}[\textbf{Typicality Matching Strategy}]
%Consider a family of pairs of CPER graphs $(\tilde{g}^1_n,\tilde{g}^2_n), n \in \mathbb{N}$, {generated based on $\prod_{i\in[n], j<i}P_{X^1,X^2}, n \in \mathbb{N}$} where $n$ is the number of vertices. 
%The TM strategy is a sequence of functions $f_n: (\tilde{g}_n^1,g_n^2)\to (\tilde{g}_n^1,\hat{g}_n^2), n\in \mathbb{N}$, where for any given $n\in \mathbb{N}$, the labeling  $\hat{\sigma}_n^2$ of $\hat{g}_n^2$ is chosen randomly and uniformly from the set $\widehat{\Sigma}$ defined in Equation \eqref{eq:sigma}.
%\label{def:TM}
%\end{Definition}
\begin{Theorem}
\label{th:21}
{For the TM strategy {described above,} a given family of sets of distributions $\widetilde{P}=(\mathcal{P}_n)_{n\in \mathbb{N}}$ is achievable, if for every sequence of distributions $P^{(n)}_{X_1,X_2}\in \mathcal{P}_n, n\in \mathbb{N}$,
\begin{align}
2(1-\alpha)\frac{\log{n}}{n-1}\leq 
max(E_{\alpha^2}, {E'}_{\!\!\alpha^2}),  0\leq \alpha\leq \alpha_n,
\label{eq:th21}
\end{align}
and $\max_{(x_1,x_2): P^{(n)}_{X_1,X_2}(x_1,x_2)\neq 0}|\log{\frac{P^{(n)}_{X_1}(x_1)P^{(n)}_{X_2}(x_2)}{P^{(n)}_{X_1,X_2}(x_1,x_2)}}|^+= o(\log{n})$, 
where $\alpha_n\to 1$ as $n\to \infty$, and $E_{\alpha^2}$ and $E'_{\alpha^2}$ are defined in Theorems \ref{th:1:improved} and \ref{th:1}, respectively.} 
% Particularly, if $X_1$ and $X_2$ are binary random variables, then $\widetilde{P}=(\mathcal{P}_n)_{n\in \mathbb{N}}$ is achievable, if for every sequence of distributions $P_{n,X_1,X_2}\in \mathcal{P}_n, n\in \mathbb{N}$
% \begingroup\makeatletter\def\f@size{9}\check@mathfonts
%\begin{align}
%\Omega(\frac{\log{n}}{n})=(P_{n,X_1,X_2}(0,0)P_{n,X_1,X_2}(1,1)-P_{n,X_1,X_2}(0,1)P_{n,X_1,X_2}(1,0))^2.
%\label{eq:kia}
%\end{align}
%\endgroup
\end{Theorem}
\begin{proof}
Appendix \ref{app:th21}.
\end{proof}
{Theorem \ref{th:21} provides sufficient conditions on the edge statistics of the CPER graphs such that the TM strategy correctly matches `almost' all vertices of the two graphs. That is, the theorem provides sufficient conditions under which $P\left(\sigma^2(v^2_{I_n})=\hat{\sigma}^2(v^2_{I_n})\right)\to 1$ as $n\to\infty$, where $I_n$ is chosen uniformly among all indices $[n]$, as defined in Definition \ref{def:strat}. 
A question of significant interest is how this sufficient condition changes if the success criterion is relaxed so that the strategy is only required to correctly match a fraction $\beta\in [0,1]$ of the vertices. More precisely, we want to know the conditions on $P^{(n)}_{X_1,X_2}, n \in \mathbb{N}$ such that \textcolor{black}{$P\left(\frac{1}{n}\Big|\big\{i:\sigma^2(v^2_{i})=\hat{\sigma}^2(v^2_{i})\big\}\Big|\geq \beta )\right)\to 1$ as  $n \to\infty$.} This is of particular interest in social network deanonymization \cite{Beyah,Grossglauser,shirani2018typicality}, where even a small fraction of matched vertices is a violation of those users' privacy. The following corollary to Theorem \ref{th:21} provides sufficient conditions under which such partial matching is possible using the TM strategy. }
{\begin{Corollary}[\textbf{Partial Matching of CPERs}]
Let $\beta>0$. 
Given a sequence of distributions $P^{(n)}_{X_1,X_2}$, the TM strategy correctly matches at least $\beta$ fraction of the vertices for asymptotically large $n$ (i.e. \textcolor{black}{$P\left(\frac{1}{n}\Big|\big\{i:\sigma^2(v^2_{i})=\hat{\sigma}^2(v^2_{i})\big\}\Big|\geq \beta )\right)\to 1$ as  $n \to\infty$.}), given that the following holds:
\begin{align}
2(1-\alpha)\frac{\log{n}}{n-1}\leq 
max(E_{\alpha^2}, E'_{\alpha^2}),  0\leq \alpha\leq \beta,
\end{align}
and $\max_{(x_1,x_2): P^{(n)}_{X_1,X_2}(x_1,x_2)\neq 0}|\log{\frac{P_{X_1}(x_1)P_{X_2}(x_2)}{P_{X_1,X_2}(x_1,x_2)}}|^+= o(\log{n})$
\end{Corollary}
{
The proof follows from the proof of Theorem \ref{th:21} in Appendix \ref{app:th21} by replacing $\alpha_n$ with $\beta$.}
\begin{Remark}
{We have restricted our analysis to matching undirected graphs where $(x,v_i,v_j)\in \mathcal{E}$ if and only if $(x,v_j,v_i)\in \mathcal{E}$. The results can be extended to directed graphs by evaluating the joint typicality of the complete adjacency matrices of the graphs rather than the upper-triangles of the adjacency matrices.}
\end{Remark}
\subsection{Matching under the Erasure Model}
In the following, we consider matching pairs of CERs under the special case of the erasure model, where  the following distribution on the graph edges is considered:
\begin{align*}
    &P^{(n)}_{X_1,X_2}(0,0)= 1-p_n, \quad P^{(n)}_{X_1,X_2}(0,1)=0, \quad
    P^{(n)}_{X_1,X_2}(1,0)= p_n(1-s), \quad
    P^{(n)}_{X_1,X_2}(1,1)= p_ns,\\
    &
    P^{(n)}_X(0)=1-p_n,\quad p^{(n)}_X(1)=p_n,\quad
    P^{(n)}_Y(0)=1-p_ns\quad
    P^{(n)}_Y(1)=p_ns.
\end{align*}
where $s\in [0,1]$ is fixed in $n$, and $p_n\to 0$ as $n\to \infty$. The model has been studied extensively in the literature (e.g. \cite{kiavash,corr2,efe}). Under the erasure model, there is an edge between each two vertices in $\tilde{g}^1_n$ with probability $p_n$. The edges in $\tilde{g}^2_n$ are sampled from the edges in $\tilde{g}^1_n$ such that each edge in $\tilde{g}^1_n$ is erased in $\tilde{g}^2_n$ with probability $(1-s)$ and it is kept with probability $s$. We are interested in finding the fastest rate at which $p_n\to 0$ such that successful matching is possible. In \cite{kiavash}, it is shown that a sufficient condition for successful matching under the erasure model is that $\frac{s}{2}p_n \geq \frac{\ln{n}}{n}$ as $n\to \infty$, where $\ln$ is the natural logarithm. Alternatively, a sufficient condition for successful matching is $\lim_{n\to \infty}\frac{\frac{\ln{n}}{n}}{p_n}\leq \frac{s}{2}$. The following theorem shows that the TM strategy improves this sufficient condition for successful matching.
}
{\begin{Theorem}
Let $\textcolor{black}{\frac{1}{4}}<s<\frac{1}{2}$. 
There exists a sequence $p_n$ approaching $0$ as $n\to \infty$ for which i) the TM strategy leads to successful matching, and ii) $
    \lim_{n\to \infty}\frac{\frac{\log_e{n}}{n}}{p_n} > \frac{s}{2}.$
\label{th:ER}
\end{Theorem}}
\begin{proof}
Appendix \ref{app:th:ER}.
\end{proof}
\begin{Remark}
\textcolor{black}{Under certain sparsity conditions on graph edges, sufficient conditions for successful matching were derived in \cite{kia_2017}. The erasure Model described above satisfies these sparsity conditions, and \cite{kia_2017} provides guarantees for successful matching when  $
    \lim_{n\to \infty}\frac{\frac{\log_e{n}}{n}}{p_n} \leq  {s}$. }
\end{Remark}

\section{Matching Graphs with Community Structure}
\label{sec:CS}
In this section, we describe the TM scheme for matching graphs generated under the SBM, i.e. graphs with community structure and provide achievable regions for these matching scenarios.  A pair of correlated graphs with community structure are a special instance of the correlated random graphs defined in Definition \ref{def:cor_rand}. In order to describe the notation used in this section, we provide a separate formal definition of random graphs with community structure below. 
\subsection{Problem Setup}
To describe the notation used in the section, consider a graph with $n\in \mathbb{N}$ vertices belonging to $c\in \mathbb{N}$ communities whose edges take $l\geq 2$ possible attributes. It is assumed that the set of communities $\mathcal{C}=\{\mathcal{C}_{1},\mathcal{C}_{2},\cdots,\mathcal{C}_{c}\}$ partitions the vertex set $\mathcal{V}$. {The $i^{th}$ community is written as $\mathcal{C}_{i}= \{v_{j_1}, v_{j_2},\cdots, v_{j_{n_i}}\}$, where $n_i\in [n]$ is the size of the $i^{th}$ community. Consequently, the graph is parametrized by $(n,c,(n_i)_{i\in [c]},l)$. We sometimes refer to such an unlabeled graph as an $(n,c,(n_i)_{i\in [c]},l)$-unlabeled graph with community structure (UCS). The set $\mathcal{E}_{i_1,i_2}=\{(x,v_{j_1},v_{j_2})\in \mathcal{E}|v_{j_1}\in \mathcal{C}_{i_1}, v_{j_2}\in \mathcal{C}_{i_2} \}$ is the set of edges connecting the vertices in communities $\mathcal{C}_{i_1}$ and $\mathcal{C}_{i_2}$. } 
%The following formally defines a graph with community structure. 
%\begin{Definition}[\bf{Graph with Community Structure}]
% An $(n,c,(n_i)_{i\in [c]},l)$-unlabeled graph with community structure (UCS) $g$ is characterized by the triple $(\mathcal{V},\mathcal{C},\mathcal{E})$, where $n,l,c,n_1,n_2,\cdots,n_c\in \mathbb{N}$ and $l\geq 2$. The family of sets $\mathcal{C}=\{\mathcal{C}_{1},\mathcal{C}_{2},\cdots,\mathcal{C}_{c}\}$ partitions $\mathcal{V}$ and is the collection of communities. The $i^{th}$ community is written as $\mathcal{C}_{i}= \{v_{j_1}, v_{j_2},\cdots, v_{j_{n_i}}\}$.
  %For the edge $(x,v_{j_1},v_{j_2})$, the variable `$x$' represents the value assigned to the edge between vertices $v_{j_1}$ and $v_{j_2}$. 
% The set $\mathcal{E}_{i_1,i_2}=\{(x,v_{j_1},v_{j_2})\in \mathcal{E}|v_{j_1}\in \mathcal{C}_{i_1}, v_{j_2}\in \mathcal{C}_{i_2} \}$ is the set of edges connecting the vertices in communities $\mathcal{C}_{i_1}$ and $\mathcal{C}_{i_2}$. 
% \label{Def:unlabeled}
% For a given vertex $v_{n,i}, i\in [1,n]$ the set of neighboring vertices is defined as $\mathcal{E}_{n,i}=\{v_{n,j}|(v_{n,i},v_{n,j})\in \mathcal{E}_n\}$. The degree of $v_{n,i}$ is defined as $d_{v_{n,i}}=|\mathcal{E}_{n,i}|$.
%\end{Definition}
{It can be noted that
The Erd\"os-R\'enyi (ER) graphs studied in Section \ref{sec:CER} are examples of single-community graphs, i.e. $c=1$.}
%For ER graphs, we often omit the variables $c$ and $\mathcal{C}$ in the notation for clarity. Consequently, in the context of Definition \ref{Def:unlabeled}, the ER graph whose edges take $l\geq 2$ values is an $(n,l)$-unlabeled graph and is represented as $g=(\mathcal{V},\mathcal{E})$.
%\end{Remark}

 \begin{table}[]
  \centering
\begin{tabular}{|ll|ll|ll|}
\hline
 n:& $\#$ of vertices & $c$:  & $\#$ of communities &  $\mathcal{C}$: & set of communities    \\ 
 \hline
 $\mathcal{C}_i$: & $i^{th}$ community & $n_i$: & size of $i^{th}$ community & $\mathcal{E}_{i,j}$: & edges between $\mathcal{C}_i$ and $\mathcal{C}_j$  \\ \hline
 $G_{\sigma}$: & adjacency matrix &
$U_{\sigma}$: & upper-triangle & $G_{\sigma,i,j}$: & adj. matrix between $\mathcal{C}_i$ and $\mathcal{C}_j$ 
\\ \hline
\end{tabular}
\vspace{0.1in}
\caption{Notation Table: Graphs with Community Structure}
\end{table}

{We consider random graphs with community structure (RCS) generated stochastically based on the SBM model. In this model, the probability of an edge between a pair of vertices is determined by their community memberships.
More precisely, for a given vertex set $\mathcal{V}$ and set of communities $\mathcal{C}$, let
$P_{X|C_{j_1},C_{j_2}},j_1,j_2\in [c]$ be a set of conditional distributions defined on  $\mathcal{X}\times \mathcal{C}\times\mathcal{C}$, where $\mathcal{X}=[0,l-1]$.
It is assumed that the edge set $\mathcal{E}$ is generated randomly, where the attribute $X$ of the edge between vertices ${v}_{i_1}\in \mathcal{C}_{j_1}$ and $v_{i_2}\in \mathcal{C}_{j_2}$ is generated based on the conditional distribution $P_{X|\mathcal{C}_{j_1},\mathcal{C}_{j_2}}$. So,
\begin{align*}
 P((x,v_{i_1},v_{i_2})\in \mathcal{E})= P_{X|C_{j_1},C_{j_2}}(x|\mathcal{C}_{j_1},\mathcal{C}_{j_2}), \forall x\in [0,l-1],
\end{align*}
where $v_{i_1},v_{i_2}\in \mathcal{C}_{j_1}\times \mathcal{C}_{j_2}$, and edges between different vertices are mutually independent. }
%\begin{Definition}[\bf{Random Graph with Community Structure}]
%\label{Def:RCS}
%Let $P_{X|C_i,C_o}$ be a set of conditional distributions defined on  $\mathcal{X}\times \mathcal{C}\times\mathcal{C}$, where $\mathcal{X}=[0,l-1]$ and $\mathcal{C}$ is defined in Definition \ref{Def:unlabeled}. A random graph with community structure (RCS) $g$ generated according to ${P_{X|C_i,C_o}}$ is a randomly generated  $(n,c,(n_i)_{i\in [c]},l)$-UCS with vertex set $\mathcal{V}$,  community set $\mathcal{C}$, and edge set $\mathcal{E}$, such that
%\begin{align*}
% P((x,v_{j_1},v_{j_2})\in \mathcal{E})= P_{X|C_i,C_o}(x|\mathcal{C}_{j_1},\mathcal{C}_{j_2}), \forall x\in [0,l-1],
%\end{align*}
%where $v_{j_1},v_{j_2}\in \mathcal{C}_{j_1}\times \mathcal{C}_{j_2}$, and edges between different vertices are mutually independent. 
 %In other words, for an arbitrary edge $e=(i,j), i,j\in [1,n]$, we have $P(e\in \mathcal{E})=p$, and the edges generated are mutually independently. 
%\end{Definition}
{It can be noted that for undirected graphs considered in this work, we must have $P_{X|C_i,C_o}(x|\mathcal{C}_{j_1}, \mathcal{C}_{j_2})=P_{X|C_i,C_o}(x| \mathcal{C}_{j_2}, \mathcal{C}_{j_1})$.
%The following provides the notation used to represent the adjacency matrix of labeled graphs with community structure. 
A labeled graph with community structure is a graph with community structure $g$ equipped with a labeling $\sigma$, and is denoted by $\tilde{g}=(g,\sigma)$. 
 %\begin{Definition}[\bf{Adjacency Matrix with Community Structure}]
%  For an $(n,c,(n_i)_{i\in [c]},l)$-UCS $g=(\mathcal{V},\mathcal{C},\mathcal{E})$, a labeling is defined as a bijective function $\sigma: \mathcal{V}\to [1,n]$.  
%The pair $\tilde{g}=(g, \sigma)$ is called an $(n,c,(n_i)_{i\in [c]},l)$-labeled graph with community structure (LCS). 
For the labeled graph $\tilde{g}$ the adjacency matrix is defined as ${G}_{\sigma}=[{G}_{\sigma,i,j}]_{i,j\in [1,n]}$ where $G_{\sigma,i,j}$ is the unique value such that $(G_{\sigma,i,j},v_k,v_l)\in \mathcal{E}_n$, where $(v_k,v_l)=(\sigma^{-1}(i),\sigma^{-1}(j))$. The submatrix $G_{\sigma,\mathcal{C}_i,\mathcal{C}_j}=[G_{\sigma,k,l}]_{k,l: v_k,v_l\in \mathcal{C}_i\times \mathcal{C}_j}$
is the adjacency matrix corresponding to the pair $\mathcal{C}_i$ and $\mathcal{C}_j$. The upper triangle (UT) corresponding to $\tilde{g}$ is the structure $U_{\sigma}=[G_{\sigma,i,j}]_{i<j}$. The upper triangle corresponding to communities $\mathcal{C}_i$ and $\mathcal{C}_j$ in $\tilde{g}$ is denoted by $U_{\sigma,\mathcal{C}_i,\mathcal{C}_j}=[G_{\sigma,k,l}]_{k<l: v_k,v_l\in \mathcal{C}_i\times \mathcal{C}_j}$. The subscript `$\sigma$' is dropped when there is no ambiguity. The notation is summarized in Table II. }
%\end{Definition}

We consider pairs of correlated RCSs. It is assumed that edges between pairs of vertices in the two graphs with the same labeling are correlated and are generated based on a joint probability distribution, whereas edges between pairs of vertices with different labeling are generated independently. A pair of correlated RCSs is formally defined below.

\begin{Definition}[\bf{Correlated Pair of RCSs}]
Let $P_{X,X'|C_{j_1},C_{j_2},C'_{j'_1},C'_{j'_2}}, j_1,j_2,j'_1,j'_2\in [1,c]$ be a set of conditional distributions defined on $\mathcal{X}\times \mathcal{X}'\times\mathcal{C}\times\mathcal{C}\times\mathcal{C}'\times\mathcal{C}'$,
where $\mathcal{X}=\mathcal{X}'=[0,l-1]$ and $(\mathcal{C},\mathcal{C}')$ are a pair of community sets of size $c\in \mathbb{N}$. A correlated pair of random graphs with community structure (CPCS) generated according to $P_{X,X'|C_{j_1},C_{j_2},C'_{j'_1},C'_{j'_2}}$ is a pair $\tilde{\underline{g}}=(\tilde{g},\tilde{g}')$ characterized by: 
i) the pair of RCSs $(g,g')$ generated according to ${P_{X|C_{j_1},C_{j_2}}}$ and ${P_{X'|C'_{j'_1},C'_{j'_2}}}$, respectively, ii) the pair of labelings $(\sigma,\sigma')$, and iii) the probability distribution $P_{X,X'|C_{j_1},C_{j_2},C'_{j'_1},C'_{j'_2}}$, such that: 
\\1) The graphs have the same set of vertices $\mathcal{V}=\mathcal{V}'$.
\\2) For any two edges $e=(x,v_{j_1},v_{j_2}), e'=(x',v'_{j'_1},v'_{j'_2}),  x,x'\in [0,l-1]$, we have
\vspace{-0.05in}
 \begin{align*}
 &Pr\left(e\in \mathcal{E}, e'\in \mathcal{E}'\right)=
\begin{cases}
P_{X,X'}(x,x'),& \text{if } \sigma(v_{j_k})=\sigma'(v'_{j'_k}), k=1,2
\\
Q_{X,X'}(x,x'), & \text{Otherwise}
\end{cases},
\end{align*}
where $l\in \{1,2\}$, $v_{j_1},v_{j_2}\in \mathcal{C}_{j_1}\times \mathcal{C}_{j_2}$, $v'_{j'_1},v'_{j'_2}\in \mathcal{C'}_{j'_1}\times \mathcal{C'}_{j'_2}$, the distribution $P_{X,X'}$ is the joint edge distribution when the edges connect vertices with similar labels and is given by $P_{X,X'|C_{j_1},C_{j_2},C'_{j'_1},C'_{j'_2}}$, the distribution $Q_{X,X'}$ is the conditional edge  distribution when the edges connect labels with different labels and is given  by $P_{X|C_{j_1},C_{j_2}}\times P_{X'|C'_{j'_1},C'_{j'_2}}$. 
% and for an arbitrary edge $e=(v_i,v_j), i,j\in [1,n]$, we have 
%\begin{align*}
% P\left(\mathbbm{1}(e\in \mathcal{E}^1)=\alpha, \mathbbm{1}(e\in \mathcal{E}^2=\beta)\right)=P_{X_1,X_2}(\cdot,\cdot), \forall \cdot,\cdot\in \{0,1\}^2,
%\end{align*}
\label{Def:CPCS}
\end{Definition}
{In this paper, in order to simlify the notation, we assume that the community memberships in both graphs are the same. In other words, we assume that $v_j\in \mathcal{C}_i \Rightarrow v'_{j'}\in \mathcal{C}'_i$ given that $\sigma(v_j)=\sigma'(v'_{j'})$ for any $j,j'\in [n]$ and $i\in [c]$. }
Furthermore, we assume that the size of the communities in the graph sequence grows linearly in the number of vertices. More precisely, let $\Lambda^{(n)}(i)\triangleq|\mathcal{C}^{(n)}_i|$ be the size of the $i^{th}$ community, we assume that\footnote{We write $f(x)=\Theta(g(x))$ if $\lim_{x\to\infty}{\frac{f(x)}{g(x)}}$ is a non-zero constant.} $\Lambda^{(n)}(i)= \Theta(n)$ for all $i\in [c]$. We also assume that the number of communities $c$ is constant in $n$.
%\end{Remark}

{We consider matching strategies under two scenarios: 
\begin{itemize}[leftmargin=*]
\item{\textbf{With Complete Side-information:} In this scenario,  the matching strategy uses prior knowledge of vertices' community memberships. A matching strategy operating with complete side-information is a sequence of functions $f^{CSI}_n: (\underline{g}^{(n)},\mathcal{C}^{(n)},\mathcal{C}^{' (n)})\mapsto  \hat{\sigma}^{' (n)}, n\in \mathbb{N}$, where $\underline{g}^{(n)}=(\tilde{g}_1^{(n)},g_2^{(n)})$ consists of  a pair of graphs with {community structure} with $n$ vertices. 
}
%\item{\textbf{Partial Side-information:} A matching algorithm operating with partial side-information is a sequence of functions $f^{p}_n: (\underline{g}_{P^{(n)}_{X,X'|C_i,C_o,C'_i,C'_o}},\mathcal{C})\mapsto  \hat{\sigma}', n\in \mathbb{N}$.}
\item{\textbf{With Partial Side-information:} A matching strategy operating with partial  side-information does not use prior knowledge of the vertices' community memberships, rather, it uses the statistics $P_{X,X'|\mathcal{C}_i,\mathcal{C}_o,\mathcal{C}_{i'},\mathcal{C}_{o'}}$ and the community sizes $(n_i)_{i\in [c]}$. The matching strategy is a sequence of functions $f^{WSI}_n: \underline{g}^{(n)}\mapsto  {\hat{\sigma}}^{' (n)}, n\in \mathbb{N}$.}
\end{itemize}
The matching strategy is said to be successful if the fraction of correctly matched vertices approaches 1 as $n\to \infty$ as formalized in Definition \ref{def:strat}.}

\subsection{Matching in Presence of Side-information}
First, we describe the matching strategy under the complete side-information scenario.
In this scenario, the community membership of the nodes at both graphs are known prior to matching. Given a CPCS $\tilde{\underline{g}}$ generated according to ${P_{X,X'|C_{j_1},C_{j_2},C'_{j'_1},C'_{j'_2}}}, j_1,j_2,j'_1,j'_2\in [1,c]$, the scheme operates as follows. It finds a labeling ${\hat{\sigma}}'$, for which i) the set of pairs $(G_{\sigma,\mathcal{C}_{j_1},\mathcal{C}_{j_2}},{G'}_{\hat{\sigma}',\mathcal{C}'_{j_1},\mathcal{C}'_{j_2}}), j_1,j_2 \in [c]$  are jointly typical each with respect to $P_{X,X'|C_{j_1},C_{j_2},C'_{j_1},C'_{j_2}}(\cdot,\cdot|C_{j_1},C_{j_2},C'_{j_1},C'_{j_2})$ when viewed as vectors of length $n_i n_j, i\neq j$, and ii) the set of pairs $(U_{\sigma,\mathcal{C}_j,\mathcal{C}_j},U'_{\hat{\sigma}',\mathcal{C}'_j,\mathcal{C}'_j}), j \in [c]$ are jointly typical with respect to $P_{X,X'|C_{j_1},C_{j_2},C'_{j_1},C'_{j_2}}(\cdot,\cdot|\mathcal{C}_j,\mathcal{C}_j,\mathcal{C}'_j,\mathcal{C}'_j)$ when viewed as vectors of length $\frac{n_i(n_i-1)}{2}, j\in [c]$. Specifically, it returns a randomly picked element $\hat{\sigma}'$ from the set:
\begin{align*}
 &\widehat{\Sigma}_{\mathcal{C}.\mathcal{C}'}=\{\hat{\sigma}'|
 (U_{\sigma,\mathcal{C}_j,\mathcal{C}_j}, U'_{\hat{\sigma}',\mathcal{C}'_j,\mathcal{C}'_j})\in \mathcal{A}_{\epsilon}^{\frac{n_j(n_j-1)}{2}}(P_{X,X'|\mathcal{C}_j, \mathcal{C}_j,\mathcal{C}'_j, \mathcal{C}'_j}),\forall j\in [c],
 \\&
 (G_{\sigma,\mathcal{C}_i,\mathcal{C}_j}, G'_{\hat{\sigma}',\mathcal{C}'_i,\mathcal{C}'_j})\in \mathcal{A}_{\epsilon}^{n_i  n_j}(P_{X,X'|\mathcal{C}_i, \mathcal{C}_j,\mathcal{C}'_i, \mathcal{C}'_j}),\forall i,j\in [c], i\neq j
 \},
\end{align*}
where $\epsilon=\omega(\frac{1}{n})$, and declares $\hat{\sigma}'$ as the correct labeling. We show that under this scheme,  the probability of incorrect labeling for any given vertex is arbitrarily small for large $n$.
\begin{Theorem}
\label{th:ach1}
For the TM strategy described above, a given family of sets of distributions $\widetilde{P}=(\mathcal{P}^{(n)})_{n\in \mathbb{N}}$ is achievable, if for any constants $\delta>0$, $\alpha\in [0,1-\delta]$ and every sequence of distributions $P^{(n)}_{X,X'|C_{j_1},C_{j_2},C'_{j'_1},C'_{j'_2}}\in \mathcal{P}_n, j_1,j_2,j'_1,j'_2\in [1,c]$, and community sizes $(n_1,n_2,\cdots, n_c)$ such that $\sum_{j=1}^cn_i=n$, the following holds:
\begin{align}
 \nonumber&{3(1-\alpha)\frac{\log{n}}{n}\leq 
\min_{[\alpha_i]_{i\in [c]}\in \mathcal{A}_{\alpha}}
 \sum_{i,j\in [c], i< j}
\frac{n_in_j}{n^2}\cdot
D(P^{(n)}_{X,X'|\mathcal{C}_i,\mathcal{C}_j}
 ||(1-\beta_{i,j})P^{(n)}_{X|\mathcal{C}_i,\mathcal{C}_j}P^{(n)}_{X'|\mathcal{C}_i,\mathcal{C}_j}+ \beta_{i,j} P^{(n)}_{X,X'|\mathcal{C}_i,\mathcal{C}_j})}
 \\&{
+\sum_{i\in [c]}
 \frac{n_i(n_i-1)}{2n^2}\cdot
 D(P^{(n)}_{X,X'|\mathcal{C}_i,\mathcal{C}_i}
 ||(1-\beta_i)P^{(n)}_{X|\mathcal{C}_i,\mathcal{C}_i}P^{(n)}_{X'|\mathcal{C}_i,\mathcal{C}_i}+ \beta_i P^{(n)}_{X,X'|\mathcal{C}_i,\mathcal{C}_i}),}
\label{eq:th31}
\end{align}
{and $\max_{(x_1,x_2): P^{(n)}_{X_1,X_2|\mathcal{C}_i,\mathcal{C}_j}(x_1,x_2)\neq 0}|\log{\frac{P_{X_1|\mathcal{C}_i}(x_1)P_{X_2|\mathcal{C}_j}(x_2)}{P_{X_1,X_2|\mathcal{C}_i,\mathcal{C}_j}(x_1,x_2)}}|^+= o(\log{n}), i,j\in [c]$}, as $n\to \infty$, where $\mathcal{A}_{\alpha}=  \{([\alpha_i]_{i\in [c]}): \alpha_i\leq \frac{n_i}{n}, \sum_{i\in [c]}\alpha_i=\alpha\}$, and $\beta_{i,j}= \frac{n^2}{n_in_j}\alpha_i\alpha_j, i,j\in [c]$ and $\beta_i=\frac{n\alpha_i(n\alpha_i-1)}{n_i(n_i-1)}, i\in [c]$. The maximal family of sets of distributions which are achievable using the typicality matching strategy with complete side-information is denoted by $\mathcal{P}_{full}$. 
\label{th:3}
% Particularly, if $X_1$ and $X_2$ are binary random variables, then $\widetilde{P}=(\mathcal{P}_n)_{n\in \mathbb{N}}$ is achievable, if for every sequence of distributions $P_{n,X_1,X_2}\in \mathcal{P}_n, n\in \mathbb{N}$
% \begingroup\makeatletter\def\f@size{9}\check@mathfonts
%\begin{align}
%\Omega(\frac{\log{n}}{n})=(P_{n,X_1,X_2}(0,0)P_{n,X_1,X_2}(1,1)-P_{n,X_1,X_2}(0,1)P_{n,X_1,X_2}(1,0))^2.
%\label{eq:kia}
%\end{align}
%\endgroup
\end{Theorem}
\begin{proof}
Appendix \ref{app:th3}.
\end{proof}

\begin{Remark}
Note that the community sizes $(n_1,n_2,\cdots, n_c), n\in \mathbb{N}$ are assumed to grow in $n$ such that $\lim_{n\to \infty}\frac{n_i}{n}>0$. 
\end{Remark}
{It can be noted that Theorem \ref{th:3} includes the achievable region for matching of pairs of Erd\H{o}s-R\`enyi graphs (i.e. single community) derived in Theorem \ref{th:2}.}
\subsection{Matching in Absence of Side-information}
The scheme described in the previous section can be extended to matching graphs without community memberships side-information. In this scenario, it is assumed that the distribution $P_{X,X'|\mathcal{C}_{j_1},\mathcal{C}_{j_2},\mathcal{C}'_{j'_1},\mathcal{C}'_{j'_2}}, j_1,j_2,j'_1,j'_2\in [1,c]$ is known, but the community memberships of the vertices in the graphs are not known. In this case, the scheme sweeps over all possible possible community membership assignments of the vertices in the two graphs. For each community membership assignment, the scheme attempts to match the two graphs using the method proposed in the complete side-information scenario. If it finds a labeling which satisfies the joint typicality conditions, it declares the labeling as the correct labeling. Otherwise, the scheme proceeds to the next community membership assignment. More precisely, for a given community assignment $(\hat{\mathcal{C}},\hat{\mathcal{C}'})$, the scheme forms the following ambiguity set
\begin{align*}
 \widehat{\Sigma}_{\hat{\mathcal{C}},\hat{\mathcal{C}}'}&=\{\hat{\sigma}'|
 (U_{\sigma,
 \hat{\mathcal{C}}_i,\hat{\mathcal{C}}_i}, U'_{\hat{\sigma}',\hat{\mathcal{C}}'_i,\hat{\mathcal{C}}'_i})\in \mathcal{A}_{\epsilon}^{\frac{n_i(n_i-1)}{2}}(P_{X,X'|\hat{\mathcal{C}}_i, \hat{\mathcal{C}}_i,\hat{\mathcal{C}}'_i, \hat{\mathcal{C}}'_i}),\forall i\in [c],
 \\&
 (G_{\sigma,\hat{\mathcal{C}}_i,\hat{\mathcal{C}}_j}, \widetilde{G'}_{\hat{\sigma}',\hat{\mathcal{C}}'_i,\hat{\mathcal{C}}'_j})\in \mathcal{A}_{\epsilon}^{n_i  n_j}(P_{X,X'|\hat{\mathcal{C}}_i, \hat{\mathcal{C}}_j,\hat{\mathcal{C}}'_i, \hat{\mathcal{C}}'_j}),\forall i,j\in [c], i\neq j
 \}.
\end{align*}
Define $
  \widehat{\Sigma}_{0}\triangleq \cup_{(\hat{\mathcal{C}},
  \hat{\mathcal{C}}')\in \mathsf{C}  } \widehat{\Sigma}_{\hat{\mathcal{C}},\hat{\mathcal{C}'}}$,
where $\mathsf{C}$ is the set of all possible community membership assignments. The scheme outputs a randomly chosen element of  $\widehat{\Sigma}_{0}$ as the correct labeling. The following theorem shows that the achievable region is the same as the one described in Theorem \ref{th:ach1}.

\begin{Theorem}
Let $\mathcal{P}_0$ be the maximal family of sets of achievable distributions for the typicality matching strategy without side-information. Then, $\mathcal{P}_0= \mathcal{P}_{full}$.
\label{th:ach2}
\end{Theorem}
The proof follows similar arguments as Theorem \ref{th:ach1}. We provide an outline. It is enough to show that $|\widehat{\Sigma}_{0}|$ has the same exponent as $|\widehat{\Sigma}_{\mathcal{C}.\mathcal{C}'}|$. Note that the size of the set of all community membership assignments $\mathsf{C}$ has an exponent which is $\Theta(n)$ since $
    |\mathsf{C}|\leq 2^{cn}$.
On the other hand,
\begin{align*}
    |\widehat{\Sigma}_{0}|\leq |\mathsf{C}| |\widehat{\Sigma}_{\mathcal{C}.\mathcal{C}'}|\leq 2^{nc} 2^{\Theta(n\log{n})}=2^{\Theta(n\log{n})}.
\end{align*}
The rest of the proof follows by the same arguments as in Theorem \ref{th:ach1}.

\section{Matching Collections of Graphs}
\label{sec:coll}
In the previous sections, we considered matching of pairs of correlated graphs. The results can be further extended to problems involving matching of collections of more than two graphs. In this section, we consider matching collections of more than two correlated graphs, where the first graph is deanonymized and the other graphs are anonymized. For brevity we consider collections of correlated Erd\"{o}s-R\'enyi graphs, i.e. single-community random graphs in Section \ref{sec:CER}. The results can be further exteneded to correlated graphs with community structure in a straightforward manner. 
We formally describe collections of correlated Erd\"{o}s-R\'enyi graphs below.
\begin{Definition}[\bf{Correlated Collection of ER Graphs}]
Let $P_{X^m}$ be a conditional distribution defined on $\prod_{k\in [m]}\mathcal{X}_i$,
where $\mathcal{X}_i=[0,l-1], i\in [m]$ and $m>2$. A correlated collection of ER graphs $\tilde{\underline{g}}=(\tilde{g}^i)_{i\in [m]}$ generated according to ${P_{X^m}}$ is characterized by: 
i) the collection of ER graphs $(g^i)_{i\in [m]}$ each generated according to ${P_{X_i}}$, ii) the collection of labelings $(\sigma_i)_{i\in [m]}$ for the unlabeled graphs $(g^i)_{i\in [m]}$, and iii) the joint probability distribution $P_{X^m}$, such that: 
\\1)The graphs have the same set of vertices $\mathcal{V}=\mathcal{V}_i, i\in [m]$. 
\\2) For any collection of edges $e^i=(x^i,v_{j^i_1},v_{j^i_2}),  x^i\in [0,l-1], i\in [m]$, we have
 \begin{align*}
 &Pr\left(e^i\in \mathcal{E}^i, i\in [m]\right)=
\begin{cases}
P_{X^m}(x^m)
,& \text{if } \sigma^i(v_{j^i_l})=\sigma^k(v_{j^k_l}), \forall i,k\in [m]\\
\prod_{i\in [m]}P_{X_i}(x_i), & \text{Otherwise}
\end{cases},
\end{align*}
where $l\in \{1,2\}$, and $v_{j^i_1},v_{j^i_2}\in \mathcal{V}_1\times \mathcal{V}_2, i\in [m]$.
\label{Def:CCER}
\end{Definition}

Similar to the TM strategy for pairs of correlated graphs described in Section \ref{sec:CER}, we propose a matching strategy based on typicality for collections of correlated graphs. Given a correlated collection of graphs $(g^i)_{i\in [m]}$, where the labeling for $\tilde{g}^1$ is given and the rest of the graphs are anonymized, the TM strategy operates as follows. The scheme finds a collection $\widehat{\Sigma}$ of labelings $\hat{\sigma}^j, j\in [2,m]$, for which the UT's $U^{j}_{{\sigma}^j}, j\in [m]$ are jointly typical with respect to $P_{n,X^m}$  when viewed as vectors of length $\frac{n(n-1)}{2}$. The strategy succeeds if at least one such labeling exists and fails otherwise.
\begin{Theorem}
For the TM strategy, a given family of sets of distributions $\widetilde{P}=(\mathcal{P}_n)_{n\in \mathbb{N}}$ is achievable, if for every sequence of distributions $P_{n,X^m}\in \mathcal{P}_n, n\in \mathbb{N}$ we have
\begin{align}
{\frac{\log{n}}{n}(\sum_{k\in [b_m]}|\mathcal{P}_k|\alpha_k-1)\leq \frac{1}{2(b_m-1)(m(m-1)+1)}D(P_{X^m}
 ||\sum_{k\in [b_m]}\alpha'_k
 P_{X_{\mathcal{P}_k}})+O(\frac{\log{n}}{n}),}
\label{eq:thcoll}
\end{align}
for all $\alpha_1,\alpha_2,\cdots, \alpha_{b_m}:\sum_{k\in [b_m]}\alpha_k=n, \alpha_{b_m} \in [1,1-\alpha_n]$, {and $\max_{x^m:P_{X^m}(x^m)\neq0} |\log{\frac{\prod_{i\in [m]}P_{X_i}(x_i)}{P_{X^m}(x^m)}}|^+=o(\log{n})$},
where $\alpha'_k= \frac{\alpha^2_k}{2}+ \sum_{k',k'':\mathcal{P}_{k',k''}=\mathcal{P}_l}\alpha_{k'}\alpha_{k''}$, $\mathcal{P}_{k',k''}= \{\mathcal{A}'\cap \mathcal{A}'': \mathcal{A}'\in \mathcal{P}_{k'}, \mathcal{A}''\in \mathcal{P}_{k''}\}, k',k''\in [b_m]$, and $\mathcal{P}_{b_m}=[1,n]$. 
\label{th:7}
% Particularly, if $X_1$ and $X_2$ are binary random variables, then $\widetilde{P}=(\mathcal{P}_n)_{n\in \mathbb{N}}$ is achievable, if for every sequence of distributions $P_{n,X_1,X_2}\in \mathcal{P}_n, n\in \mathbb{N}$
% \begingroup\makeatletter\def\f@size{9}\check@mathfonts
%\begin{align}
%\Omega(\frac{\log{n}}{n})=(P_{n,X_1,X_2}(0,0)P_{n,X_1,X_2}(1,1)-P_{n,X_1,X_2}(0,1)P_{n,X_1,X_2}(1,0))^2.
%\label{eq:kia}
%\end{align}
%\endgroup
\end{Theorem}
\begin{proof}
Appendix \ref{app:th7}.
\end{proof}
\begin{Remark}
Note that Equation \eqref{eq:thcoll} recovers the result given in Equation \eqref{eq:th21} for matching of pairs of correlated ER graphs, i.e. $m=2$. 
\end{Remark}
% and for an arbitrary edge $e=(v_i,v_j), i,j\in [1,n]$, we have 
%\begin{align*}
% P\left(\mathbbm{1}(e\in \mathcal{E}^1)=\alpha, \mathbbm{1}(e\in \mathcal{E}^2=\beta)\right)=P_{X_1,X_2}(\cdot,\cdot), \forall \cdot,\cdot\in \{0,1\}^2,
%\end{align*}

\section{Converse Results}
\label{sec:converse}
In this section, we provide conditions on the graph parameters under which graph matching is not possible. Without loss of generality, we assume that $(\bf{\sigma},\bf{\sigma}')$ are a pair of random labelings chosen uniformly among the set of all possible labeling for the two graphs. Roughly speaking, the information revealed by identifying the realization of $\bf{\sigma}'$ is equal to $H(\sigma')=log{(n!)}\approx log{(n^n)}= n\log{n}$. Consequently, using Fano's inequality, we show that the information contained in $(\sigma, g,g')$ regarding $\bf{\sigma}'$, which is quantified as the mutual information $I(\sigma'; \sigma, g,g')$, must be at least $n\log{n}$ bits for successful matching. The mutual information $I(\sigma'; \sigma,g,g')$ is a function of multi-letter probability distributions. We use standard information theoretic techniques to bound $I(\sigma'; \sigma, g,g')$ using information quantities which are functionals of single-letter distributions. The following states the resulting necessary conditions for successful matching. 

\begin{Theorem}
\label{th:converse}
For the graph matching problem under the community structure model with complete side-information, the following provides necessary conditions  for successful matching:
\begin{align*}
{\frac{\log{n}}{n}}&{\leq \sum_{i,j \in [c], i<j}\frac{n_in_j}{n^2} I(X,X'|\mathcal{C}_i,\mathcal{C}_j, \mathcal{C}'_i\mathcal{C}'_j)
+  \sum_{i \in [c]}\frac{n_i(n_i-1)}{2n^2} I(X,X'|\mathcal{C}_i,\mathcal{C}_i, \mathcal{C}'_i,\mathcal{C}'_i)+o(\frac{\log{n}}{n}), }
\end{align*}
where $I(X,X'|\mathcal{C}_i,\mathcal{C}_j, \mathcal{C}'_i\mathcal{C}'_j)$ is defined with respect to $P_{X,X'|\mathcal{C}_i,\mathcal{C}_j, \mathcal{C}'_i\mathcal{C}'_j}$.
\end{Theorem}
\begin{proof}
Appendix \ref{app:conv}.
\end{proof}

% For Erd\H{o}s-R\`enyi graphs, the following corollary is a direct consequence of Theorem \ref{th:converse}.
\begin{Corollary}
For the graph matching problem under the  Erd\H{o}s-R\`enyi model, the following provides necessary conditions  for successful matching:
\begin{align*}
{\frac{2\log{n}}{n}}&{\leq   I(X,X')+o(\frac{\log{n}}{n})}.
\end{align*}
\end{Corollary}

\section{Seeded Graph Matching}
\label{sec:SGM}
\begin{figure}
\centering 
\includegraphics[width=3.7in, draft=false]{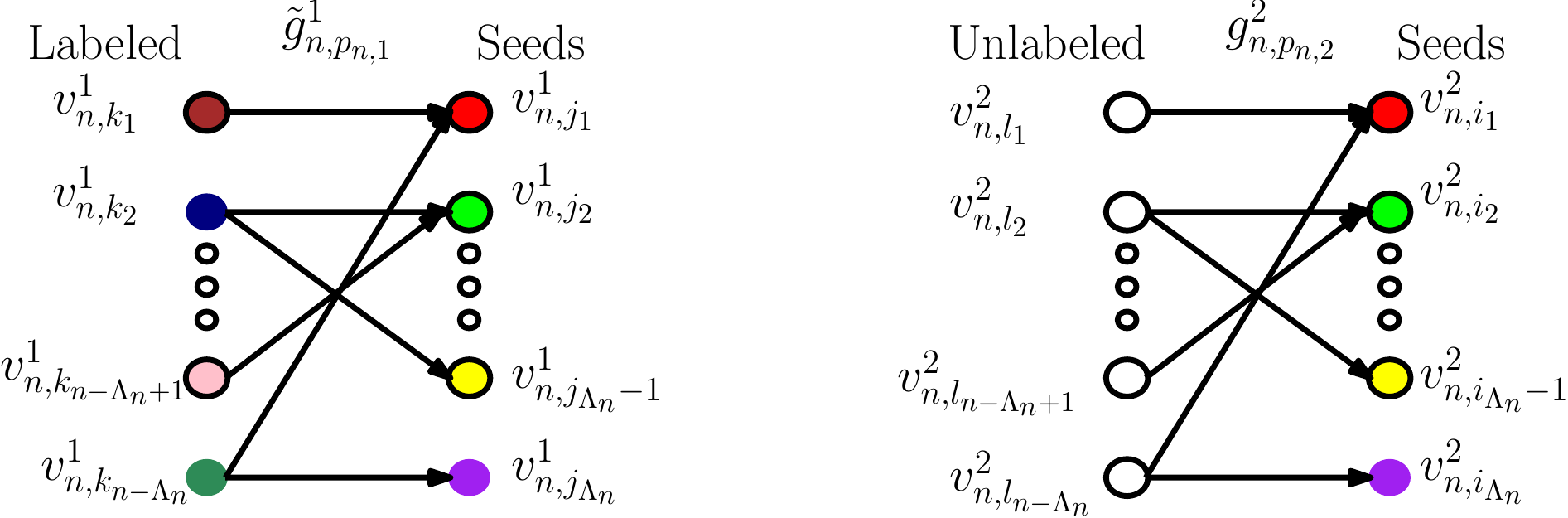}
\caption{The matching algorithm constructs the bipartite graph which captures the connections between the unmatched vertices with the seed vertices.}
\label{fig:bipart}
\end{figure}
So far, we have investigated the fundamental limits of graph matching assuming the availability of unlimited computational resources. In this section, we consider seeded graph matching, and propose a matching algorithm whose complexity grows polynomially in the number of vertices of the graph and leads to successful matching in a wide range of graph matching scenarios. The algorithm leverages ideas from prior work a related problem called \textit{online fingerprinting} which involves matching of correlated bipartite graphs \cite{Allerton}. 

In seeded graph matching, it is assumed that we are given the correct labeling for a subset of the vertices in the anonymized graph prior to the start of the matching process. The subset of pre-matched vertices are called `\textit{seeds}'. The motivation behind the problem formulation is that in many applications of graph matching, the correct labeling of a subset of vertices is known through side-inforamtion. For instance, in social network deanonymization, many users link their social media accounts across social networks publicly. As shown in this section, the seed side-information can be used to significantly reduce the complexity of the matching algorithm. 

The proposed graph matching algorithm operates as follows. First, the algorithm constructs the bipartite graph shown in Figure~\ref{fig:bipart} whose edges consist of the connections between the unmatched vertices with the seeded vertices in each graph. The algorithm proceeds in two steps. First, it constructs the `\textit{fingerprint}' vectors for each of the unmatched vertices in the two bipartite graphs based on their connections to the seed vertices. The fingerprint vector of a vertex is the row in the adjacency matrix of the bipartite graph corresponding to the edges between that  vertex and the seed vertices.  
In the second step, the algorithm finds a jointly typical pair of fingerprint vectors in the deanonymized and deanonymized graph adjacency matrices and matches the corresponding vertices, where typicality is defined based on the joint distribution between the edges of the two graphs. 
Note that the bipartite graphs encompass only a subset of the edges in the original graphs. Hence by restricting the matching process to the bipartite graphs, some of the information which could potentially help in matching is ignored. This leads to more restrictive conditions on successful matching compared to the ones derived in the previous sections. However, the computational complexity of the resulting  matching algorithm is considerably improved. In the following, we focus on matching of seeded CPERs. The results can be easily extended to seeded CPCSs similar to the unseeded graph matching in prior sections. A seeded CPER (SCPER) is formally defined below. 

\begin{Definition}[\bf{Correlated Pair of Seeded ER Graphs}]
An SCPER is a triple $(\tilde{g},\tilde{g}', \mathcal{S})$, where $\tilde{\underline{g}}=(\tilde{g},\tilde{g}')$ is a CPER generated according to $P_{X,X'}$, and $\mathcal{S}\subseteq \mathcal{V}$ is the seed set. 
\label{Def:SCPER}
\end{Definition}

  Let $\mathcal{S}=\{v_{i_1},v_{i_2},\cdots, v_{i_{\Lambda}}\}$ and define the reverse seed set  $\mathcal{S}^{-1}=\{v_{j_1},v_{j_2},\cdots, v_{j_{\Lambda}}\}$, where $\sigma(v_{j_k})=\sigma'(v_{i_k}), k\in [1,\Lambda]$.
 The algorithm is given the correct labeling of all the vertices in the first graph $\sigma:\mathcal{V}\to [1,n]$ and the seed vertices in the second graph $\sigma'|_{\mathcal{S}}: \mathcal{S}\to [1,n]$. The objective is to find the correct labeling of the rest of the vertices in the second graph $\hat{\sigma}_n: \mathcal{V};\to [1,n]$ so that the fraction of mislabeled vertices is negligible as the number of vertices grows asymptotically large, i.e.  $P(\hat{\sigma}'=\sigma')\to 1$ as $n\to \infty$. To this end, the algorithm first constructs a fingerprint for each vertex in each of the graphs. For an arbitrary vertex $v_{i}$ in $g_{P_X}$, its fingerprint is defined as $\underline{F}_{i}=(F_{i}(1), F_{i}(2),\cdots,F_{i}(\Lambda))$.
which indicates its connections to the reverse seed elements:
\begin{align*}
{F}_{i}(l)=
\begin{cases}
 1 \qquad & \text{if} \qquad (v_{i}, v_{j_l})\in \mathcal{E}\\
 0& \text{Otherwise}
\end{cases},
\qquad l\in [1,\Lambda].
\end{align*}

The fingerprint of a vertex $v_{i}$ in the second graph is defined in a similar fashion based on connections to the elements of the seed set $\mathcal{S}$. Take an unmatched vertex $v_{i}\notin \mathcal{S}$. The algorithm matches $v_{i}$ in $g$ to a vertex $v_{j}$ in $g'$ if it is the unique vertex such that the fingerprint pair $(\underline{F}_{i}, \underline{F}'_{j})$ are jointly $\epsilon$-typical with respect to the distribution $P_{X,X'}$, where\footnote{Alternatively, $\lim_{n\to \infty} \frac{\epsilon}{\sqrt{|\mathcal{S}|}}=\infty$.} $\epsilon= \omega(\frac{1}{\sqrt{\Lambda}})$: 
\begin{align*}
 \exists ! i: (\underline{F}_{i}, \underline{F}'_{j})\in \mathcal{A}_{\epsilon}^n(X,X') \Rightarrow \hat{\sigma}(v_{i})=\sigma'(v_{j}),
\end{align*}
where $\mathcal{A}_{\epsilon}^n(X,X')$ is the set of jointly $\epsilon$-typical set sequences of length $n$ with respect to $P_{X,X'}$. If a unique match is not found, then vertex $v_{i}$ is added to the ambiguity set $\mathcal{L}$. Hence, $\mathcal{V}\backslash \mathcal{L}$ is the set of all matched vertices. In the next step, these vertices are added to the seed set and the expanded seed set is used to match the vertices in the ambiguity set.
The algorithm succeeds if all vertices are matched at this step and fails otherwise. We call this strategy the Seeded Typicality Matching Strategy (STM).% The following characterizes the achievable region for the TMS.

\begin{Theorem}
Define the family of sets of pairs of distribution and seed sizes $\widetilde{\mathcal{P}}$ as follows:
\begin{align*}
 \widetilde{\mathcal{P}}=& \Big\{(\mathcal{P}_n,\Lambda_n)_{n\in \mathbb{N}}\Big|
 \forall P_{n,X,X'}\in \mathcal{P}_n: \frac{2\log{n}}{I(X,X')}\leq \Lambda_n ,  I(X;X')=\omega\left(\sqrt{\frac{1}{\Lambda_n}}\right)\Big\}.
\end{align*}
Any family of SCPERs with parameters chosen from $\widetilde{\mathcal{P}}$ is matchable using the STM strategy.
\label{th:seeded}
\end{Theorem}
The proof which is provided in Appendix \ref{app:thseeded} uses the following lemma.
 \begin{Lemma}
 \label{lem:card}
 The following holds:
 \begin{align*}
  P(|\mathcal{L}|> \frac{2n}{\Lambda\epsilon^2})\to 0, \text{ as } n\to \infty,
\end{align*}
\end{Lemma}
\begin{proof}
Appendix \ref{app:card}.
\end{proof}

\section{Conclusion}
\label{sec:conc}
We have considered matching of collections of correlated graphs. We have studied the problem under the Erd\"os-R\`enyi model as well as the more general community structure model. The derivations apply to graphs whose edges may take non-binary attributes. We have introduced a graph matching scheme called the Typicality Matching scheme which relies on tools such as concentration of measure and typicality of sequences of random variables to perform graph matching. We have further provided converse results which lead to necessary conditions on graph parameters for successful matching. 
We have investigated seeded graph matching, where the correct labeling of a subset of graph vertices is known prior to the matching process. We have introduced a matching algorithm for seeded graph matching which successfully matches the graphs in wide range of matching problems with large enough seeds and has a computational complexity which grows  polynomially in the number of graph vertices. 
\bibliographystyle{unsrt}
%  \bibliography{FLReferences}
  \bibliography{Privacy_bibliography_Dec_2017}% Use the BibTeX file ``References.bib''.

\begin{thebibliography}{10}

\bibitem{gross2005information}
Ralph Gross and Alessandro Acquisti.
\newblock Information revelation and privacy in online social networks.
\newblock In {\em Proceedings of the 2005 ACM workshop on Privacy in the
  electronic society}, pages 71--80, 2005.

\bibitem{shah2014collaborative}
Chirag Shah.
\newblock Collaborative information seeking.
\newblock {\em Journal of the Association for Information Science and
  Technology}, 65(2):215--236, 2014.

\bibitem{backstrom2007wherefore}
Lars Backstrom, Cynthia Dwork, and Jon Kleinberg.
\newblock Wherefore art thou r3579x?: anonymized social networks, hidden
  patterns, and structural steganography.
\newblock In {\em Proceedings of the 16th international conference on World
  Wide Web}, pages 181--190. ACM, 2007.

\bibitem{ref2}
Arvind Narayanan and Vitaly Shmatikov.
\newblock De-anonymizing social networks.
\newblock In {\em Proceedings of the 2009 30th IEEE Symposium on Security and
  Privacy}, SP '09, pages 173--187, Washington, DC, USA, 2009. IEEE Computer
  Society.

\bibitem{Beyah}
S.~Ji, W.~Li, N.~Z. Gong, P.~Mittal, and R.~Beyah.
\newblock Seed-based de-anonymizability quantification of social networks.
\newblock {\em IEEE Transactions on Information Forensics and Security},
  11(7):1398--1411, July 2016.

\bibitem{Grossglauser}
E.~Kazemi, L.~Yartseva, and M.~Grossglauser.
\newblock When can two unlabeled networks be aligned under partial overlap?
\newblock In {\em 2015 53rd Annual Allerton Conference on Communication,
  Control, and Computing (Allerton)}, pages 33--42, Sept 2015.

\bibitem{shirani2018typicality}
Farhad Shirani, Siddharth Garg, and Elza Erkip.
\newblock Typicality matching for pairs of correlated graphs.
\newblock In {\em 2018 IEEE International Symposium on Information Theory
  (ISIT)}, pages 221--225. IEEE, 2018.

\bibitem{foggia2014graph}
Pasquale Foggia, Gennaro Percannella, and Mario Vento.
\newblock Graph matching and learning in pattern recognition in the last 10
  years.
\newblock {\em International Journal of Pattern Recognition and Artificial
  Intelligence}, 28(01):1450001, 2014.

\bibitem{xu2019cross}
Kun Xu, Liwei Wang, Mo~Yu, Yansong Feng, Yan Song, Zhiguo Wang, and Dong Yu.
\newblock Cross-lingual knowledge graph alignment via graph matching neural
  network.
\newblock {\em arXiv preprint arXiv:1905.11605}, 2019.

\bibitem{kazemi2016proper}
Ehsan Kazemi, Hamed Hassani, Matthias Grossglauser, and Hassan~Pezeshgi
  Modarres.
\newblock Proper: global protein interaction network alignment through
  percolation matching.
\newblock {\em BMC bioinformatics}, 17(1):527, 2016.

\bibitem{ER}
Paul Erdos and Alfr{\'e}d R{\'e}nyi.
\newblock On the evolution of random graphs.
\newblock {\em Publ. Math. Inst. Hung. Acad. Sci}, 5(1):17--60, 1960.

\bibitem{wright}
Edward~M Wright.
\newblock Graphs on unlabelled nodes with a given number of edges.
\newblock {\em Acta Mathematica}, 126(1):1--9, 1971.

\bibitem{iso1}
L{\'a}szl{\'o} Babai, Paul Erdos, and Stanley~M Selkow.
\newblock Random graph isomorphism.
\newblock {\em SIAM Journal on computing}, 9(3):628--635, 1980.

\bibitem{iso2}
B{\'e}la Bollob{\'a}s.
\newblock Random graphs. 2001.
\newblock {\em Cambridge Stud. Adv. Math}, 2001.

\bibitem{iso4}
Tomek Czajka and Gopal Pandurangan.
\newblock Improved random graph isomorphism.
\newblock {\em Journal of Discrete Algorithms}, 6(1):85--92, 2008.

\bibitem{corr1}
Ehsan Kazemi.
\newblock Network alignment: Theory, algorithms, and applications.
\newblock 2016.

\bibitem{corr2}
Lyudmila Yartseva and Matthias Grossglauser.
\newblock On the performance of percolation graph matching.
\newblock In {\em Proceedings of the first ACM conference on Online social
  networks}, pages 119--130. ACM, 2013.

\bibitem{corr3}
Pedram Pedarsani, Daniel~R Figueiredo, and Matthias Grossglauser.
\newblock A bayesian method for matching two similar graphs without seeds.
\newblock In {\em 2013 51st Annual Allerton Conference on Communication,
  Control, and Computing (Allerton)}, pages 1598--1607. IEEE, 2013.

\bibitem{corr4}
Shouling Ji, Weiqing Li, Mudhakar Srivatsa, and Raheem Beyah.
\newblock Structural data de-anonymization: Quantification, practice, and
  implications.
\newblock In {\em Proceedings of the 2014 ACM SIGSAC Conference on Computer and
  Communications Security}, pages 1040--1053. ACM, 2014.

\bibitem{kia_2017}
Daniel Cullina and Negar Kiyavash.
\newblock Exact alignment recovery for correlated erdos renyi graphs.
\newblock {\em arXiv preprint arXiv:1711.06783}, 2017.

\bibitem{Lyzinski_2016}
Vince Lyzinski.
\newblock Information recovery in shuffled graphs via graph matching.
\newblock {\em arXiv preprint arXiv:1605.02315}, 2016.

\bibitem{cullina2018partial}
Daniel Cullina, Negar Kiyavash, Prateek Mittal, and H~Vincent Poor.
\newblock Partial recovery of erd$\backslash$h $\{$o$\}$
  sr$\backslash$'$\{$e$\}$ nyi graph alignment via $ k $-core alignment.
\newblock {\em arXiv preprint arXiv:1809.03553}, 2018.

\bibitem{community}
Michelle Girvan and Mark~EJ Newman.
\newblock Community structure in social and biological networks.
\newblock {\em Proceedings of the national academy of sciences},
  99(12):7821--7826, 2002.

\bibitem{fortunato2012community}
Santo Fortunato and Claudio Castellano.
\newblock Community structure in graphs.
\newblock {\em Computational Complexity: Theory, Techniques, and Applications},
  pages 490--512, 2012.

\bibitem{shirani2018matching}
Farhad Shirani, Siddharth Garg, and Elza Erkip.
\newblock Matching graphs with community structure: A concentration of measure
  approach.
\newblock In {\em 2018 56th Annual Allerton Conference on Communication,
  Control, and Computing (Allerton)}, pages 1028--1035. IEEE, 2018.

\bibitem{nilizadeh2014community}
Shirin Nilizadeh, Apu Kapadia, and Yong-Yeol Ahn.
\newblock Community-enhanced de-anonymization of online social networks.
\newblock In {\em Proceedings of the 2014 acm sigsac conference on computer and
  communications security}, pages 537--548. ACM, 2014.

\bibitem{singhal2017significance}
Kushagra Singhal, Daniel Cullina, and Negar Kiyavash.
\newblock Significance of side information in the graph matching problem.
\newblock {\em arXiv preprint arXiv:1706.06936}, 2017.

\bibitem{lyzinski2014seeded}
Vince Lyzinski, Daniel~L Sussman, Donniell~E Fishkind, Henry Pao, Joshua~T
  Vogelstein, and Carey~E Priebe.
\newblock Seeded graph matching for large stochastic block model graphs.
\newblock {\em stat}, 1050:12, 2014.

\bibitem{efe}
Efe Onaran, Siddharth Garg, and Elza Erkip.
\newblock Optimal de-anonymization in random graphs with community structure.
\newblock In {\em Signals, Systems and Computers, 2016 50th Asilomar Conference
  on,}, pages 709--713. IEEE, 2016.

\bibitem{kiavash}
Daniel Cullina and Negar Kiyavash.
\newblock Improved achievability and converse bounds for erdos-renyi graph
  matching.
\newblock {\em SIGMETRICS Perform. Eval. Rev.}, 44(1):63--72, June 2016.

\bibitem{seed1}
Ehsan Kazemi, S~Hamed Hassani, and Matthias Grossglauser.
\newblock Growing a graph matching from a handful of seeds.
\newblock {\em Proceedings of the VLDB Endowment}, 8(10):1010--1021, 2015.

\bibitem{seed2}
Carla-Fabiana Chiasserini, Michele Garetto, and Emilio Leonardi.
\newblock Social network de-anonymization under scale-free user relations.
\newblock {\em IEEE/ACM Transactions on Networking}, 24(6):3756--3769, 2016.

\bibitem{seed3}
Vince Lyzinski, Donniell~E Fishkind, and Carey~E Priebe.
\newblock Seeded graph matching for correlated erd{\"o}s-r{\'e}nyi graphs.
\newblock {\em Journal of Machine Learning Research}, 15(1):3513--3540, 2014.

\bibitem{seed4}
Marcelo Fiori, Pablo Sprechmann, Joshua Vogelstein, Pablo Mus{\'e}, and
  Guillermo Sapiro.
\newblock Robust multimodal graph matching: Sparse coding meets graph matching.
\newblock In {\em Advances in Neural Information Processing Systems}, pages
  127--135, 2013.

\bibitem{Asilomar}
Farhad Shirani, Siddharth Garg, and Elza Erkip.
\newblock Seeded graph matching: Efficient algorithms and theoretical
  guarantees.
\newblock In {\em 2017 51st Asilomar Conference on Signals, Systems, and
  Computers}, pages 253--257. IEEE, 2017.

\bibitem{mossel2019seeded}
Elchanan Mossel and Jiaming Xu.
\newblock Seeded graph matching via large neighborhood statistics.
\newblock In {\em Proceedings of the Thirtieth Annual ACM-SIAM Symposium on
  Discrete Algorithms}, pages 1005--1014. SIAM, 2019.

\bibitem{fishkind2019seeded}
Donniell~E Fishkind, Sancar Adali, Heather~G Patsolic, Lingyao Meng, Digvijay
  Singh, Vince Lyzinski, and Carey~E Priebe.
\newblock Seeded graph matching.
\newblock {\em Pattern Recognition}, 87:203--215, 2019.

\bibitem{lyzinski2017matchability}
Vince Lyzinski and Daniel~L Sussman.
\newblock Matchability of heterogeneous networks pairs.
\newblock {\em arXiv preprint arXiv:1705.02294}, 2017.

\bibitem{zhang2016final}
Si~Zhang and Hanghang Tong.
\newblock Final: Fast attributed network alignment.
\newblock In {\em Proceedings of the 22nd ACM SIGKDD International Conference
  on Knowledge Discovery and Data Mining}, pages 1345--1354. ACM, 2016.

\bibitem{heimann2018regal}
Mark Heimann, Haoming Shen, Tara Safavi, and Danai Koutra.
\newblock Regal: Representation learning-based graph alignment.
\newblock In {\em Proceedings of the 27th ACM International Conference on
  Information and Knowledge Management}, pages 117--126. ACM, 2018.

\bibitem{csiszarbook}
I.~Csisz\'{a}r and J.~Korner.
\newblock {\em Information Theory: Coding Theorems for Discrete Memoryless
  Systems}.
\newblock Academic Press Inc. Ltd., 1981.

\bibitem{isaacs}
I~Martin Isaacs.
\newblock {\em Algebra: a graduate course}, volume 100.
\newblock American Mathematical Soc., 1994.

\bibitem{tuncel2005error}
Ertem Tuncel.
\newblock On error exponents in hypothesis testing.
\newblock {\em IEEE Transactions on Information Theory}, 51(8):2945--2950,
  2005.

\bibitem{csiszar1998method}
Imre Csisz{\'a}r.
\newblock The method of types [information theory].
\newblock {\em IEEE Transactions on Information Theory}, 44(6):2505--2523,
  1998.

\bibitem{comtet2012advanced}
Louis Comtet.
\newblock {\em {A}dvanced {C}ombinatorics: {T}he {A}rt of {F}inite and
  {I}nfinite {E}xpansions}.
\newblock Springer Science \& Business Media, 2012.

\bibitem{Allerton}
F.~Shirani, S.~Garg, and E.~Erkip.
\newblock An information theoretic framework for active de-anonymization in
  social networks based on group memberships.
\newblock In {\em 55rd Annual Allerton Conf. on Communication, Control, and
  Computing}, Sept 2017.

\end{thebibliography}
 \newpage
\begin{appendices}
\section{{Proof of Theorem \ref{th:1:improved}}}
\label{Ap:th1:improved}
{ Let $\pi$ be an $(m,c,i_1,i_2,\cdots,i_c)$-permutation for parameters $m,c, i_1,i_2,\cdots, i_c$ as defined in Definition \ref{def:cycle}. Note that $\pi$ and $\pi^{-1}$ have the same number and lengths of cycles \cite{isaacs}. Using Proposition \ref{prop:1}, we assume without loss of generality that $\pi^{-1}$
is a standard permutation such that
\begin{align*}
&\pi^{-1}=(1,2,\cdots,i_1)(i_1+1,i_1+2,\cdots,i_1+i_2)\cdots
\\&\qquad(\sum_{j=1}^{c-1}i_j+1,\sum_{j=1}^{c-1}i_j+2,\cdots,\sum_{j=1}^{c}i_j)(n-m+1)(n-m+2)\cdots (n).
\end{align*}
Define the following partition for the set of indices $[1,n]$:}
\begin{align*}
&{\mathcal{A}_1= \{k| \exists \ell\in [c]: \sum_{j=1}^{\ell-1} i_j +1\leq k< \sum_{j=1}^{\ell} i_j, \text{ and }
k-\sum_{j=1}^{\ell-1} i_j \text{ is odd}\}}
\\&
{\mathcal{A}_2= \{k| \exists \ell\in [c]:  k= \sum_{j=1}^{\ell} i_j, \text{ and }
i_{\ell} \text{ is odd}\}}
\\&
{\mathcal{A}_3= \{k| \exists \ell\in [c]: \sum_{j=1}^{\ell-1} i_j +1\leq k< \sum_{j=1}^{\ell} i_j-1, \text{ and }
k-\sum_{j=1}^{\ell-1} i_j \text{ is even}\}}
\\&
{\mathcal{A}_4= \{k| \exists \ell\in [c]:  k= \sum_{j=1}^{\ell} i_j, \text{ and }
i_{\ell} \text{ is even}\}}
\\&
{\mathcal{A}_5= \{k| \exists \ell\in [c]:  k= \sum_{j=1}^{\ell} i_j-1, \text{ and }
i_{\ell} \text{ is odd}\}}
\\&{\mathcal{A}_6= \{k| k>\sum_{j=1}^c i_j\}.}
\end{align*}
{
The set $\mathcal{A}_1$ is the set of 
odd elements of each cycle which are not its endpoints, i.e the first, third, fifth $\dots$ elements of each cycle if they are not at the end of that cycle.
The set $\mathcal{A}_2$ contains the indices which are endpoints of cycles with odd lengths. The set 
 $\mathcal{A}_3$
is the set of 
even elements of each cycle which are not one of the last two elements of the cycle, i.e the second, fourth, sixth, $\dots$ elements of each cycle if they are not one of the last two elements of that cycle. The set $\mathcal{A}_4$
is the set of 
endpoints of cycles with even lengths. The set
 $\mathcal{A}_5$ is the set of the last even index of the cycles with odd length, i.e. the indices before the last element of cycles with odd lengths. Lastly, $\mathcal{A}_6$ is the set of fixed indices. }
    
  {  Let $Z^n= \pi(Y^n)$. Let $\underline{T}_j, j\in [6]$ be the type of the sequence $(X_i,Z_i), i \in \mathcal{A}_j$, so that $T_{j}(x,y)= \frac{\sum_{i\in \mathcal{A}_j}\mathbbm{1}(X_i=x, Z_i=y)}{|\mathcal{A}_j|}, j\in [6]$ and $x,y\in \mathcal{X}\times  \mathcal{Y}$. Furthermore, let $\underline{T}$ be the type of $(X_i,Z_i), i\in [n]$. Note that $\underline{T}= \alpha_1\underline{T}_1+\alpha_2\underline{T}_2+\alpha_3\underline{T}_3+\alpha_4\underline{T}_4+\alpha_5\underline{T}_5+\alpha_6\underline{T}_6$, where  $\alpha_j\triangleq \frac{|\mathcal{A}_j|}{n}, i\in [6]$.}
    
{     We are interested in the probability of the event $(X^n,Z^n)\in \mathcal{A}_{\epsilon}^n(X,Y)$. From Definition \ref{Def:typ} this event can be rewritten as follows:
\begin{align*}
 &P\left(\left(X^n,Z^n\right)\in \mathcal{A}_{\epsilon}^n(X,Y)\right)
  =P\left(\underline{T}\stackrel{.}{=} P_{X,Y}(\cdot,\cdot)\pm \epsilon\right)
 \\&= P(\alpha_1\underline{T}_1+\alpha_2\underline{T}_2+\alpha_3\underline{T}_3+\alpha_4\underline{T}_4+\alpha_5\underline{T}_5+\alpha_6\underline{T}_6\stackrel{.}{=}P_{X,Y}(\cdot,\cdot)\pm \epsilon),
    \end{align*}}
    {
 where addition of $\underline{T}_j, j\in [6]$ is defined element-wise. We have: 
\begin{align}
\label{eq:Ap:0.25}
 &P((X^n,Z^n)\in \mathcal{A}_{\epsilon}^n(X,Y))
 =\sum_{(\underline{t}_1,\underline{t}_2,\underline{t}_3,\underline{t}_4, \underline{t}_5,\underline{t}_6)\in \mathcal{T}} P(\underline{T}_j=\underline{t}_j, j\in [6]),
\end{align}}
{where $\mathcal{T}= \{(\underline{t}_1,\underline{t}_2,\underline{t}_3,\underline{t}_4,\underline{t}_5, \underline{t}_6):\alpha_1\underline{t}_1+\alpha_2\underline{t}_2+\alpha_3\underline{t}_3+\alpha_4\underline{t}_4+\alpha_5\underline{t}_5+\alpha_6\underline{t}_6\stackrel{.}{=}P_{X,Y}(\cdot,\cdot)\pm \epsilon\}$.
Consequently, we need to investigate $P(\underline{T}_j=\underline{t}_j, j\in [6])$ for tuples $(\underline{t}_1, \underline{t}_2, \underline{t}_3, \underline{t}_4, \underline{t}_5,\underline{t}_6)\in \mathcal{T}$. Note that the vector $(X_i,Z_i), i\in \mathcal{A}_j, j\in [5]$ is independent of the vector $(X_i,Z_i),i \in \mathcal{A}_6$ since the following facts hold: i)
the original sequence $(X^n,Y^n)$ is i.i.d, ii) 
$(X_i,Z_i), i\in \mathcal{A}_j, j\in [5]$ is a permutation of $(X_i,Y_i), i\notin \mathcal{A}_6$, and iii) $(X_i,Z_i),i \in \mathcal{A}_6$ is equal to $(X_i,Y_i), i\in \mathcal{A}_6$.}

{As a result,
\begin{align}
   P_{\underline{T}_1,\underline{T}_2,\underline{T}_3,\underline{T}_4,\underline{T}_5,\underline{T}_6}(\underline{t}_1,\underline{t}_2,\underline{t}_3,\underline{t}_4,\underline{t}_5,\underline{t}_6)=
   P_{\underline{T}_1,\underline{T}_2,\underline{T}_3,\underline{T}_4,\underline{T}_5}(\underline{t}_1,\underline{t}_2,\underline{t}_3,\underline{t}_4,\underline{t}_5)P_{\underline{T}_6}(\underline{t}_6).
   \label{eq:Ap:0.5}
\end{align}}

{For the rest of the proof, we will remove the subscript $\underline{T}_1,\underline{T}_2,\underline{T}_3,\underline{T}_4,\underline{T}_5,\underline{T}_6$ in $P_{\underline{T}_1,\underline{T}_2,\underline{T}_3,\underline{T}_4,\underline{T}_5,\underline{T}_6}$ when there is no ambiguity.  
Let $\underline{T}_{X,j}(x)= \sum_{y\in \mathcal{Y}} \underline{T}_j(x,y), x\in \mathcal{X}, j\in [6]$, and $\underline{t}_{X,j}(x)= \sum_{y\in \mathcal{Y}} \underline{t}_j(x,y), x\in \mathcal{X}, j\in [6]$. The type $\underline{T}_{X,j}$ is the marginal type of $X$ under $\underline{T}_j$, i.e. the type of the vector $X_i, i\in \mathcal{A}_j$, and similarly $\underline{t}_{X,j}$ is the marginal type of $X$ under $\underline{t}_j$. Furthermore, define the following conditional types:
\begin{align*}
    &\underline{V}_{j}(y|x)= \frac{\underline{T}_{j}(x,y)}{\underline{T}_{X,j}(x)}, x,y \in \mathcal{X}\times\mathcal{Y}, j\in [6]\\
    &
    \underline{v}_{j}(y|x)= 
    \frac{\underline{t}_{j}(x,y)}{\underline{t}_{X,j}(x)}, x,y \in \mathcal{X}\times\mathcal{Y}, j\in [6]
\end{align*}}
{Note that $\underline{T}_j, j\in [6]$ is completely determined by $(\underline{T}_{X,j}, \underline{V}_j)$.
We have:
\begin{align}
   \nonumber P(\underline{t}_1,\underline{t}_2,\underline{t}_3,
    \underline{t}_4, \underline{t}_5)
    &= P(\underline{v}_1,\underline{v}_2,\underline{v}_3,\underline{v}_4,\underline{v}_5, \underline{t}_{X,1},\underline{t}_{X,2},\underline{t}_{X,3},\underline{t}_{X,4},\underline{t}_{X,5})
    \\&
    =
    P(\underline{t}_{X,1},\underline{t}_{X,2},\underline{t}_{X,3},\underline{t}_{X,4},\underline{t}_{X,5})
    \times 
    \label{eq:Ap:1}
    P(\underline{v}_1,\underline{v}_2,\underline{v}_3,\underline{v}_4,\underline{v}_5|\underline{t}_{X,1},\underline{t}_{X,2},\underline{t}_{X,3},\underline{t}_{X,4},\underline{t}_{X,5})
    \end{align}
    Note that
    \begin{align}
    &\nonumber
 P(\underline{v}_1,\underline{v}_2,\underline{v}_3,\underline{v}_4,\underline{v}_5|\underline{t}_{X,1},\underline{t}_{X,2},\underline{t}_{X,3},\underline{t}_{X,4},\underline{t}_{X,5})
=
    \sqrt{ P^2(\underline{v}_1,\underline{v}_2,\underline{v}_3,\underline{v}_4,\underline{v}_5|\underline{t}_{X,1},\underline{t}_{X,2},\underline{t}_{X,3},\underline{t}_{X,4},\underline{t}_{X,5})}
    \\&\nonumber
    \leq
    \sqrt{ P(\underline{v}_1|\underline{t}_{X,1},\underline{t}_{X,2},\underline{t}_{X,3},\underline{t}_{X,4},\underline{t}_{X,5})
    }\times
    \sqrt{P(\underline{v}_2,\underline{v}_3,\underline{v}_4,\underline{v}_5|\underline{t}_{X,1},\underline{t}_{X,2},\underline{t}_{X,3},\underline{t}_{X,4},\underline{t}_{X,5})}
    \\&\label{eq:Ap:2}
    = 
       \sqrt{ P(\underline{v}_1|\underline{t}_{X,1},\underline{t}_{X,2},\underline{t}_{X,3},\underline{t}_{X,4})}
       \times 
       \sqrt{P(\underline{v}_2|\underline{t}_{X,2},\underline{t}_{X,5})}
       \times
       \sqrt{P(\underline{v}_3,\underline{v}_4,\underline{v}_5|\underline{t}_{X,1},\underline{t}_{X,3},\underline{t}_{X,4},\underline{t}_{X,5})},
\end{align}}
{where in the last equality, we have used the Markov chains $\underline{V}_1\leftrightarrow \underline{T}_{X,1}, \underline{T}_{X,2}, \underline{T}_{X,3}, \underline{T}_{X,4}\leftrightarrow \underline{T}_{X,5}$, 
$\underline{V}_2 \leftrightarrow \underline{T}_{X,2},\underline{T}_{X,5} \leftrightarrow \underline{T}_{X,1}, \underline{T}_{X,3}, \underline{T}_{X,4}$
and
$\underline{V}_3, \underline{V}_4 , \underline{V}_5 \leftrightarrow \underline{T}_{X,1},\underline{T}_{X,3},\underline{T}_{X,4},\underline{T}_{X,5} \leftrightarrow \underline{T}_{X,2}, \underline{V}_{2}$, which hold by construction of $\mathcal{A}_j, j\in [6]$. From Equations \eqref{eq:Ap:1} and \eqref{eq:Ap:2}, we have:
\begin{align*}
    &P(\underline{t}_1,\underline{t}_2,\underline{t}_3,
    \underline{t}_4, \underline{t}_5)\leq 
     P( \underline{t}_{X,1},\underline{t}_{X,2},\underline{t}_{X,3},\underline{t}_{X,4},\underline{t}_{X,5})
     \times
       \sqrt{ P(\underline{v}_1|\underline{t}_{X,1},\underline{t}_{X,2},\underline{t}_{X,3},\underline{t}_{X,4})}
     \times 
       \sqrt{P(\underline{v}_2|\underline{t}_{X,2},\underline{t}_{X,5})}
              \\&\times
       \sqrt{P(\underline{v}_3,\underline{v}_4,\underline{v}_5|\underline{t}_{X,1},\underline{t}_{X,3},\underline{t}_{X,4},\underline{t}_{X,5})}
       \\&=
    \sqrt{P( \underline{t}_{X,1},\underline{t}_{X,2},\underline{t}_{X,3},\underline{t}_{X,4},\underline{t}_{X,5})}
     \times
       \sqrt{ P(\underline{v}_1,\underline{t}_{X,1}|\underline{t}_{X,2},\underline{t}_{X,3},\underline{t}_{X,4})}
      \times 
       \sqrt{(\underline{v}_2,\underline{t}_{X,2}|\underline{t}_{X,5})}
       \\&\times
       \sqrt{P(\underline{v}_3,\underline{v}_4,\underline{v}_5,\underline{t}_{X,3},\underline{t}_{X,4},\underline{t}_{X,5}|\underline{t}_{X,1})},
\end{align*}}
{where in the last equality, we have used the joint independence of $\underline{T}_{1,X},\underline{T}_{2,X},\underline{T}_{3,X},\underline{T}_{4,X},\underline{T}_{5,X}$ due to the fact that $X^n$ is an i.i.d sequence. So,
\begin{align}
    & P(\underline{t}_1,\underline{t}_2,\underline{t}_3,
    \underline{t}_4, \underline{t}_5)\leq 
    \sqrt{P( \underline{t}_{X,1},\underline{t}_{X,2},\underline{t}_{X,3},\underline{t}_{X,4},\underline{t}_{X,5})}
     \times
       \sqrt{ P(\underline{t}_1|\underline{t}_{X,2},\underline{t}_{X,3},\underline{t}_{X,4})}\times 
       \sqrt{P(\underline{t}_2|\underline{t}_{X,5})}
       \times
       \sqrt{P(\underline{t}_3,\underline{t}_4,\underline{t}_5|\underline{t}_{X,1})},
       \label{eq:Ap:3.5}
\end{align}
In order to evaluate the above, we use the following propositions. 
\begin{Proposition}[\hspace{-.003in}\cite{csiszarbook}]
Assume that the sequence $U^m, m\in\mathbb{N}$ is a sequence of i.i.d variables defined on the probability space $(\mathcal{U}, P_U)$. Let $\underline{T}_U$ be the type of $U^m$, and $\underline{t}_U$ be an arbitrary type vector for m-length sequences defined on $\mathcal{U}$. Then,
\begin{align*}
    P(\underline{T}_U=\underline{t}_U)
    \leq \textcolor{black}{\exp_2\Big(-mD(\underline{t}_U||P_U)\Big)}.
\end{align*}
\label{Prop:Ap:1}
\end{Proposition}}
\begin{Proposition}
\label{Prop:Ap:2}
Let the triple $(U^m,V^m,W^m), m\in \mathbb{N}$ be a triple of random vectors defined on $\mathcal{U}^m\times \mathcal{V}^m\times \mathcal{W}^m$ such that i) $U^m$ has type $\underline{t}_U$, ii) $V^m$ is a vector of independent random variables where each $V_i, i\in [m]$ is  generated based on $P_{V|U}(\cdot|u_i)$, and iii) $W^m$ is a sequence of i.i.d random variables generated according to $P_W$. Let $\underline{T}_{V,W}$ be the joint type of $(V^m,W^m)$, and $\underline{t}_{V,W}$ be an arbitrary type for m-length sequences defined on $\mathcal{V}\times \mathcal{W}$. Then,
\begin{align*}
P(\underline{T}_{V,W}=\underline{t}_{V,W})
\leq \textcolor{black}{\exp_2\Big(-m(D(\underline{t}_{V,W}|| P_WP'_V)- |\mathcal{U}||\mathcal{V}||\mathcal{W}|\frac{\log{(m+1)}}{m})\Big)},
\end{align*}
where $P'_V(\cdot)\triangleq \sum_{u\in \mathcal{U}}\underline{t}_U(u)P_{V|U}(\cdot|u)$.
\end{Proposition}
{Proof of Proposition \ref{Prop:Ap:2}: Let $\underline{T}_{U,V,W}$ be the joint type of $(U^m,V^m,W^m)$, and define  $\underline{T}^u_{V,W}(v,w)\triangleq \frac{1}{\underline{t}_U(u)}\underline{T}_{U,V,W}(u,v,w), u,v,w\in \mathcal{U}
\times \mathcal{V}\times \mathcal{W}$. Note that $\underline{T}^u_{V,W}, u\in \mathcal{U}$ is a collection of types for sequences defined on $\mathcal{V}^m\times \mathcal{W}^m$. Furthermore, note that $\underline{T}_{V,W}= \sum_{u\in \mathcal{U}} \underline{t}_U(u) \underline{T}^u_{V,W}$. 
%Also, define $\underline{t}^{u}_{V,W}\triangleq \frac{1}{m\underline{t}_U(u)}\underline{t}_{U,V,W}(u,v,w), u,v,w\in \mathcal{U}
%\times \mathcal{V}\times \mathcal{W}$ in a similar fashion. 
Then, 
\begin{align}
    \nonumber P(\underline{T}_{V,W}=\underline{t}_{V,W})
    &=\sum_{(\underline{t}^u_{V,W})_{u\in \mathcal{U}}: \underline{t}_{V,W}= \sum_{u\in \mathcal{U}} \underline{t}_U(u) \underline{t}^u_{V,W} }
    P(\underline{T}^u_{V,W}=\underline{t}^u_{V,W}, u\in \mathcal{U})
    \\&\label{Eq:Ap:3}
    =
    \sum_{(\underline{t}^u_{V,W})_{u\in \mathcal{U}}: \underline{t}_{V,W}= \sum_{u\in \mathcal{U}} \underline{t}_U(u) \underline{t}^u_{V,W} }
    \prod_{u\in \mathcal{U}}
    P(\underline{T}^u_{V,W}=\underline{t}^u_{V,W}),
\end{align}}
{where in the last equality, we have used the fact that $(U^m,V^m,W^m)$ is a sequence of independent triples of variables, i.e. $(U_i,V_i,W_i)$ is independent of the vector $(U_j,V_j,W_j)_{j\in [m]: j \neq i}$ for any $i\in [m]$. Note that for a given $u\in \mathcal{U}$, for the subset of indices for which $U_i=u$, the variables $V_i$ are i.i.d with distribution $P_{V|U}(\cdot|u)$. Hence,
from Proposition \ref{Prop:Ap:1} we have:
\begin{align*}
  P(\underline{T}^u_{V,W}=\underline{t}^u_{V,W})\leq \textcolor{black}{\exp_2\Big({-m\underline{t}_U(u)D(\underline{t}^u_{V,W}|| P_{V|U}(\cdot|u)P_W)}\Big)  }
\end{align*}
So, from Equation \eqref{Eq:Ap:3}, we have:}
{
\begin{align*}
    P(\underline{T}_{V,W}=\underline{t}_{V,W})&\leq
    \sum_{(\underline{t}^u_{V,W})_{u\in \mathcal{U}}: \underline{t}_{V,W}= \sum_{u\in \mathcal{U}} \underline{t}_U(u) \underline{t}^u_{V,W} }
    \prod_{u\in \mathcal{U}}
    \textcolor{black}{\exp_2\Big(-m\underline{t}_U(u)D(\underline{t}^u_{V,W}|| P_{V|U}(\cdot|u)P_W)\Big)}
    \\&
    =
     \sum_{(\underline{t}^u_{V,W})_{u\in \mathcal{U}}: \underline{t}_{V,W}= \sum_{u\in \mathcal{U}} \underline{t}_U(u) \underline{t}^u_{V,W} }
  \textcolor{black}{\exp_2\Big(-m \sum_{u\in \mathcal{U}}\underline{t}_U(u)D(\underline{t}^u_{V,W}|| P_{V|U}(\cdot|u)P_W)\Big)}
  \\&  \leq 
      \sum_{(\underline{t}^u_{V,W})_{u\in \mathcal{U}}: \underline{t}_{V,W}= \sum_{u\in \mathcal{U}} \underline{t}_U(u) \underline{t}^u_{V,W} }
    \textcolor{black}{\exp_2\left(-m D\left(\sum_{u\in \mathcal{U}}\underline{t}_U(u)\underline{t}^u_{V,W}|| \sum_{u\in \mathcal{U}}\underline{t}_U(u)P_{V|U}(\cdot|u)P_W\right)\right)},
\end{align*}
where in the last inequality we have used the convexity of KL divergence. Consequently, 
\begin{align*}
     P(\underline{T}_{V,W}=\underline{t}_{V,W})
     &\leq
          \sum_{(\underline{t}^u_{V,W})_{u\in \mathcal{U}}: \underline{t}_{V,W}= \sum_{u\in \mathcal{U}} \underline{t}_U(u) \underline{t}^u_{V,W} }
    \textcolor{black}{\exp_2\Big(-m D(\underline{t}_{V,W}|| P'_VP_W)\Big)}
    \\&
    =\big|\{(\underline{t}^u_{V,W})_{u\in \mathcal{U}}: \underline{t}_{V,W}= \sum_{u\in \mathcal{U}} \underline{t}_U(u) \underline{t}^u_{V,W} \}\big|\textcolor{black}{\exp_2\Big(-m D(\underline{t}_{V,W}|| P'_VP_W)\Big)}
    \\&\leq \textcolor{black}{\exp_2\left(-m\left(D\left(\underline{t}_{V,W}|| P'_VP_W\right)- |\mathcal{U}||\mathcal{V}||\mathcal{W}|\frac{\log{(m+1)}}{m}\right)\right)},
\end{align*}
where in the last inequality, we have used the fact (e.g. \cite{csiszarbook}) that the number of types on m-length sequences defined on the alphabet $\mathcal{U}\times \mathcal{V}\times \mathcal{W}$ is upper-bounded by $(m+1)^{|\mathcal{U}||\mathcal{V}||\mathcal{W}|}$. This completes the proof of the proposition. \qedsymbol}

{From Propositions \ref{Prop:Ap:1} and \ref{Prop:Ap:2} we conclude the following:
\begin{align}
  \label{eq:Ap:4}  &P(\underline{t}_1|\underline{t}_{X,2},\underline{t}_{X,3},\underline{t}_{X,4})\leq
     \textcolor{black}{\exp_2\left(-n\alpha_1\left( D\left(\underline{t}_1|| P_XP_{Y_1}\right)- |\mathcal{X}|^2|\mathcal{Y}|\frac{\log{(n+1)}}{\alpha_1n} \right)\right)}
       \\&P(\underline{t}_2|\underline{t}_{X,5})
       \leq
       \textcolor{black}{\exp_2\left(-n\alpha_2\left( D\left(\underline{t}_2|| P_XP_{Y_2}\right)- |\mathcal{X}|^2|\mathcal{Y}|\frac{\log{(n+1)}}{\alpha_2n} \right)\right)}
       \label{eq:Ap:5}
       \\&
       P(\underline{t}_3,\underline{t}_4,\underline{t}_5|\underline{t}_{X,1})
       \leq \textcolor{black}{\exp_2\left(-n\left(\alpha_3+\alpha_4+\alpha_5\right)\left( D\left(\underline{t}'|| P_XP_{Y'}\right)- |\mathcal{X}|^2|\mathcal{Y}|\frac{\log{(n+1)}}{\left(\alpha_3+\alpha_4+\alpha_5\right)n} \right)\right)},
              \label{eq:Ap:6}
              \\&
P(\underline{t}_{X,j})\leq \textcolor{black}{\exp_2\Big(-n\alpha_j D(\underline{t}_{X,j}||P_X)\Big)}, j\in [5]
         \label{eq:Ap:7}
\\&
P(\underline{t}_6)\leq 
\textcolor{black}{\exp_2\Big(-n\alpha_6 D(\underline{t}_6||P_{X,Y})\Big)},
         \label{eq:Ap:8}
\end{align}}
\noindent {where we have defined  $P_{Y_1}(\cdot)\triangleq \sum_{j=2}^4 \alpha_j\sum_{x\in \mathcal{X}} \underline{t}_{X,j}(x) P_{Y|X}(\cdot|x)
$, $P_{Y_2}(\cdot)\triangleq \sum_{x\in \mathcal{X}} \underline{t}_{X,5}(x) P_{Y|X}(\cdot|x)$, $\underline{t}'\triangleq \alpha_3\underline{t}_3+\alpha_4\underline{t}_4+\alpha_5\underline{t}_5$, and 
$P_{Y'}(\cdot)\triangleq \sum_{x\in \mathcal{X}} \underline{t}_{X,1}(x) P_{Y|X}(\cdot|x)$. In order to explain the above derivations, let us first derive Equation \eqref{eq:Ap:4} using Proposition \ref{Prop:Ap:2}. Take $m= n\alpha_1$,  $U^m= (X_i, i\in \mathcal{A}_2\bigcup \mathcal{A}_3\bigcup \mathcal{A}_4)$,  $V^m=(Z_i, i\in \mathcal{A}_1)$, and $W^m=(X_i, i\in \mathcal{A}_1)$. Note that $|\mathcal{A}_2\bigcup \mathcal{A}_3\bigcup \mathcal{A}_4|= |\mathcal{A}_1|=n\alpha_1=m$ by construction of $\mathcal{A}_j, j\in \{1,2,3,4\}$, so the above vectors are well-defined. Also, note that the conditions in Proposition \ref{Prop:Ap:2} are satisfied as explained in the following: i) $U^m$ has type $\sum_{j=2}^4\alpha_j \underline{t}_{X,j}$, ii) for any index $i\in \mathcal{A}_1$,  if $i$ is not at the beginning of its cycle then $Z_i= \pi(Y_i)=Y_{i-1}$ so $Z_i$ is generated according to $P_{Y|X}(\cdot|x_{i-1})$ where $i-1\in \mathcal{A}_3$ by construction, and if $i$ is at the beginning of its cycle then $Z_i= Y_{i'}$, so $Z_i$ is generated according to $P_{Y|X}(\cdot|x_{i'})$
where $i'\in \mathcal{A}_2\cup \mathcal{A}_4$ is the index of the element at the end of the cycle. So, each element of $V^m$ is generated conditioned on the corresponding element in $U^m$ independent of elements of $W^m$, and iii) $W^m$ is an i.i.d sequence of variables generated according to $P_X$. Equations \eqref{eq:Ap:5}-\eqref{eq:Ap:8} are derived by similar arguments. From Equations \eqref{eq:Ap:3.5}, and \eqref{eq:Ap:4}-\eqref{eq:Ap:8} we conclude:}
{
\begin{align*}
    &\nonumber P(\underline{t}_1,\underline{t}_2,\underline{t}_3,
    \underline{t}_4, \underline{t}_5)\leq
       \sqrt{ \textcolor{black}{\exp_2\Big(-\sum_{j=1}^5n\alpha_jD(\underline{t}_{X,j}||P_X)}\Big)}
       \times
       \sqrt{\textcolor{black}{\exp_2\Big(-n\alpha_1( D(\underline{t}_1|| P_XP_{Y_1})- |\mathcal{X}|^2|\mathcal{Y}|\frac{\log{(n+1)}}{\alpha_1n} )}\Big)}
       \\&
    \times 
       \sqrt{ \textcolor{black}{\exp_2\Big(-n\alpha_2( D(\underline{t}_2|| P_XP_{Y_2})- |\mathcal{X}|^2|\mathcal{Y}|\frac{\log{(n+1)}}{\alpha_2n} )}\Big)}
       \\&\times
       \sqrt{\textcolor{black}{\exp_2\Big(-n(\alpha_3+\alpha_4+\alpha_5)( D(\underline{t}'|| P_XP_{Y'})- |\mathcal{X}|^2|\mathcal{Y}|\frac{\log{(n+1)}}{(\alpha_3+\alpha_4+\alpha_5)n} )\Big)}}
       \\&
       = 
       \textcolor{black}{\exp_2\Big(\frac{-n}{2}(\sum_{j=1}^5 \alpha_jD(\underline{t}_{X,j}||P_X)
       +\alpha_1 D(\underline{t}_1|| P_XP_{Y_1})
       +
       \alpha_2 D(\underline{t}_2|| P_XP_{Y_2})+}
       \\&
       \textcolor{black}{(\alpha_3+\alpha_4+\alpha_5) D(\underline{t}'|| P_XP_{Y'})- 3|\mathcal{X}|^2|\mathcal{Y}|\frac{\log{(n+1)}}{n} )\Big)},
       %\\&
       %\leq 
       %2^{\frac{-\overline{n}}{2}(D(\underline{t}'_X||P_X)+ D(\sum_{j=1}^5\overline{\alpha}_j \underline{t}_j|| P_XP_{Y''})- 3|\mathcal{X}|^2|\mathcal{Y}|\frac{\log{(n+1)}}{n} )
       %},
\end{align*}
From Equation \eqref{eq:Ap:0.5}, we have:
\begin{align*}
    &P(\underline{t}_1,\underline{t}_2,\underline{t}_3,\underline{t}_4,\underline{t}_5,\underline{t}_6)\leq
    \\&
   \textcolor{black}{\exp_2\Big(\frac{-n}{2}(\sum_{j=1}^5 \alpha_jD(\underline{t}_{X,j}||P_X)
       +\alpha_1 D(\underline{t}_1|| P_XP_{Y_1})
       +
       \alpha_2 D(\underline{t}_2|| P_XP_{Y_2})
       +}
       \\&\textcolor{black}{(\alpha_3+\alpha_4+\alpha_5) D(\underline{t}'|| P_XP_{Y'})- 3|\mathcal{X}|^2|\mathcal{Y}|\frac{\log{(n+1)}}{n} )\Big)}
       \textcolor{black}{\exp_2\Big(-n\alpha_6D(\underline{t}_6||P_{X,Y})\Big)}
      \\&
      \stackrel{(a)}{\leq}
          \textcolor{black}{\exp_2\Big(\frac{-n}{2}(\sum_{j=1}^5 \alpha_jD(\underline{t}_{X,j}||P_X)
       +\alpha_1 D(\underline{t}_1|| P_XP_{Y_1})
       +
       \alpha_2 D(\underline{t}_2|| P_XP_{Y_2})+}
       \\&\textcolor{black}{
       (\alpha_3+\alpha_4+\alpha_5) D(\underline{t}'|| P_XP_{Y'})- 3|\mathcal{X}|^2|\mathcal{Y}|\frac{\log{(n+1)}}{n} )
       \Big)}
       \textcolor{black}{\exp_2\Big(-\frac{n}{2}\alpha_6(D(\underline{t}_6||P_{X,Y})+ D(\underline{t}_{X,6}||P_{X}))\Big)}
    \\&\stackrel{(b)}{\leq}
    \textcolor{black}{\exp_2\Big(\frac{-{n}}{2}(\overline{\alpha}D(\underline{t}'_X||P_X)+\alpha_6D(\underline{t}_{6,X}||P_X)+ D(\sum_{j=1}^6{\alpha}_j \underline{t}_j|| \overline{\alpha} P_XP_{Y''}+\alpha_6 P_{X,Y})- 3|\mathcal{X}|^2|\mathcal{Y}|\frac{\log{(n+1)}}{n} )\Big)},
    \end{align*}
where we have defined $\underline{t}'_X\triangleq \sum_{j=1}^5 \overline{\alpha}_j\underline{t}_{X,j}$, $\overline{\alpha}_j= \frac{\alpha_j}{\overline{\alpha}}, j\in [5]$,  $\overline{\alpha}\triangleq \alpha_1+\alpha_2+\alpha_3+\alpha_4+\alpha_5= 1-\alpha_6$,
$P_{Y''}\triangleq \overline{\alpha}_1P_{Y_1}+\overline{\alpha}_2P_{Y_2}+(\overline{\alpha}_3+\overline{\alpha}_4+\overline{\alpha}_5)P_{Y'}$. In (a) we have used the fact that $D(\underline{t}_6||P_{X,Y})\geq \frac{1}{2}(D(\underline{t}_{X,6}||P_X)+D(\underline{t}_6||P_{X,Y}))$ which holds by the chain rule of KL divergence. In (b) we have used the convexity of the KL divergence. We use the following proposition to further simplify the result. }
{
\begin{Proposition}
\label{Prop:div}
Let $\epsilon>0$, and let $P,P_{\epsilon}$ and $Q$ be distributions defined on the same \textcolor{black}{finite} alphabet $\mathcal{U}$ such that i) $max_{u\in \mathcal{U}} |P(u)-P_{\epsilon}(u)|\leq \epsilon$, ii)  $\frac{1}{2}min_{u\in \mathcal{U}:P(u)\neq 0} P(u)> \epsilon$, and iii) $P_{\epsilon}(u)=0$ if $P(u)=0$.
Then 
\begin{align*}
    |D(P_{\epsilon}||Q)-D(P||Q)|\leq \epsilon|\mathcal{U}| \max_{u\in \mathcal{U}:P(u)\neq 0}|\log{\frac{P(u)}{Q(u)}}|+O(\epsilon).
\end{align*}
\label{Prop:Ap:3}
\end{Proposition}
Proof of Proposition \ref{Prop:Ap:3}: By assumption if $P(u)=0$, then $P_{\epsilon}(u)=0$, so that $u$ does not contribute to the value of $D(P_{\epsilon}||Q)$ and $D(P||Q)$. So, we assume $P(u)\neq 0, u\in \mathcal{U}$.
We have
\begin{align*}
  &|D(P_{\epsilon}||Q)-D(P||Q)|=
  \\& 
  \big|-\sum_{u\in \mathcal{U}}P_{\epsilon}(u)\log{\frac{1}{P_{\epsilon}(u)}}+\sum_{u\in \mathcal{U}}P_{\epsilon}(u)\log \frac{1}{Q(U)}+ \sum_{u\in \mathcal{U}}P(u)\log\frac{1}{P(u)} -\sum_{u\in \mathcal{U}}P(u)\log \frac{1}{Q(u)}\big|
   \\&=
   \big|-\sum_{u\in \mathcal{U}} P(u)\log{\frac{P(u)}{P_{\epsilon}(u)}}-
   \sum_{u\in \mathcal{U}} (P_{\epsilon}(u)-P(u))\log{\frac{1}{P_{\epsilon}(u)}}+\sum_{u\in \mathcal{U}} (P_{\epsilon}(u)-P(u))\log{\frac{1}{Q(u)}}\big|
   \\&
   = 
   \big|-\sum_{u\in \mathcal{U}} P(u)\log{\frac{P(u)}{P_{\epsilon}(u)}}+\sum_{u\in \mathcal{U}} (P_{\epsilon}(u)-P(u))\log{\frac{P_{\epsilon}(u)}{Q(u)}}\big|
   \\&
   =    \big|-\sum_{u\in \mathcal{U}} P(u)\log{\frac{1}{1+\frac{P_{\epsilon}(u)-P(u)}{P(u)}}}+\sum_{u\in \mathcal{U}} (P_{\epsilon}(u)-P(u))\log{\frac{P_{\epsilon}(u)}{Q(u)}}\big|
   \\&
=
    \big|-\sum_{u\in \mathcal{U}} P(u)O\left(\frac{P_{\epsilon}(u)-P(u)}{P(u)}\right)+\sum_{u\in \mathcal{U}} (P_{\epsilon}(u)-P(u))\log{\frac{P_{\epsilon}(u)}{Q(u)}}\big|
  \\& =
    \big|-\sum_{u\in \mathcal{U}} O({P_{\epsilon}(u)-P(u)})+\sum_{u\in \mathcal{U}} (P_{\epsilon}(u)-P(u))\log{\frac{P_{\epsilon}(u)}{Q(u)}}\big|
    \\&
    =
    \big|O(\epsilon)+\sum_{u\in \mathcal{U}} (P_{\epsilon}(u)-P(u))\log{\frac{P_{\epsilon}(u)}{Q(u)}}\big|
    \\&=
      O(\epsilon)+\big|\sum_{u\in \mathcal{U}} (P_{\epsilon}(u)-P(u))\log{\frac{P_{\epsilon}(u)-P(u)+P(u)}{Q(u)}}\big|
      \\&
      \leq      O(\epsilon)+\epsilon|\mathcal{U}| \max_{u\in \mathcal{U}}\big\{\big|\log{\frac{P(u)+\epsilon}{Q(u)}},\log{\frac{P(u)-\epsilon}{Q(u)}}\big|\big\}
      \\&
      \leq 
      \epsilon|\mathcal{U}| \max_{u\in \mathcal{U}}\big|\log{\frac{P(u)}{Q(u)}}\big|+O(\epsilon),
\end{align*}
where the last inequality follows from the assumption that $\frac{1}{2}{P(u)}\leq P(u)\pm\epsilon\leq  \frac{3}{2}P(u), u\in \mathcal{U}$ so that $\epsilon|\mathcal{U}| \max_{u\in \mathcal{U}}\{\big|\log{\frac{P(u)\pm\epsilon}{Q(u)}}\big|\}=\epsilon|\mathcal{U}| \max_{u\in \mathcal{U}}\{\big| \log{\frac{P(u)}{Q(u)}}\big|\}+O(\epsilon)$.
This completes the proof of the proposition.\qedsymbol}

{
Note that
$\sum_{j=1}^6{\alpha}_j \underline{t}_j(x,y)\stackrel{\cdot}{=} P_{X,Y}(x,y)\pm \epsilon, x,y\in \mathcal{X}\times \mathcal{Y}$, and if $P_{X,Y}(x,y)=0$, then $\sum_{j=1}^6{\alpha}_j \underline{t}_j(x,y)=0$ by Definition \ref{Def:typ}. In Proposition \ref{Prop:div}, take $\mathcal{U}=\mathcal{X}\times \mathcal{Y}$, $P=P_{X,Y}$, $P_{\epsilon}= \sum_{j=1}^6{\alpha}_j \underline{t}_j(x,y)$ and $Q= \alpha P_{X,Y}+(1-\alpha)P_XP_Y$. Then,
\begin{align*}
  &|D( \sum_{j=1}^6{\alpha}_j \underline{t}_j(x,y)|| \alpha P_{X,Y}+(1-\alpha)P_XP_Y)- D(P_{X,Y}|| \alpha P_{X,Y}+(1-\alpha)P_XP_Y) | \leq 
  \\&
  \epsilon|\mathcal{X}||\mathcal{Y}|
  \big|\max_{x,y \in \mathcal{X}\times \mathcal{Y}:P_{X,Y}(x,y)\neq 0} \log{\frac{P_{X,Y}(x,y)}{
   \alpha P_{X,Y}(x,y)+(1-\alpha)P_X(x)P_Y(y)
  }}\big|+O(\epsilon)
\end{align*}
}

{
Let $\delta_{\epsilon}\triangleq  \epsilon|\mathcal{X}||\mathcal{Y}|
  \big|\max_{x,y \in \mathcal{X}\times \mathcal{Y}:P_{X,Y}(x,y)\neq 0} \log{\frac{P_{X,Y}(x,y)}{
   \alpha P_{X,Y}(x,y)+(1-\alpha)P_X(x)P_Y(y)
  }}\big|+O(\epsilon)$. We have:
\begin{align*}
    &P(\underline{t}_1,\underline{t}_2,\underline{t}_3,\underline{t}_4,\underline{t}_5,\underline{t}_6)\leq
    \textcolor{black}{\exp_2\Big(\frac{-{n}}{2}((1-\alpha)D(\underline{t}'_X||P_X)+\alpha D(\underline{t}_{6,X}||P_X)+}
    \\&
    \textcolor{black}{D(P_{X,Y}|| (1-\alpha) P_XP_{Y''}+\alpha P_{X,Y})- 3|\mathcal{X}|^2|\mathcal{Y}|\frac{\log{(n+1)}}{n} -\delta_{\epsilon})\Big)},
    \end{align*}
    where we have used the fact that by theorem statement, $\alpha_6=\alpha$ and by definition $\overline{\alpha}=1-\alpha$.
From Equation \eqref{eq:Ap:0.25}, we have:
\begin{align*}
 &P((X^n,Z^n)\in \mathcal{A}_{\epsilon}^n(X,Y))
 \leq \sum_{(\underline{t}_1,\underline{t}_2,\underline{t}_3,\underline{t}_4, \underline{t}_5,\underline{t}_6)\in \mathcal{T}}  \textcolor{black}{\exp_2\Big(\frac{-{n}}{2}((1-\alpha)D(\underline{t}'_X||P_X)+\alpha D(\underline{t}_{6,X}||P_X)+}
 \\& \textcolor{black}{D(P_{X,Y}|| (1-\alpha) P_XP_{Y''}+\alpha P_{X,Y})- 3|\mathcal{X}|^2|\mathcal{Y}|\frac{\log{(n+1)}}{n} -\delta_{\epsilon})\Big)}
 \\&
 =|(\underline{t}_1,\underline{t}_2,\underline{t}_3,\underline{t}_4, \underline{t}_5,\underline{t}_6)\in \mathcal{T}|\textcolor{black}{\exp_2\Big(\frac{-{n}}{2}\min_{\underline{t}'_X\in
\mathcal{P}}((1-\alpha)D(\underline{t}'_X||P_X)+\alpha D(\underline{t}''_{X}||P_X)+}
\\&
\textcolor{black}{D(P_{X,Y}|| (1-\alpha) P_XP_{Y''}+\alpha P_{X,Y})- 3|\mathcal{X}|^2|\mathcal{Y}|\frac{\log{(n+1)}}{n} -2\delta_{\epsilon})\Big)}
 \\&
 \leq 
\textcolor{black}{\exp_2\Big(\frac{-{n}}{2}\min_{\underline{t}'_X\in
\mathcal{P}}((1-\alpha)D(\underline{t}'_X||P_X)+\alpha D(\underline{t}''_{X}||P_X)+} \\&
\textcolor{black}{D(P_{X,Y}|| (1-\alpha) P_XP_{Y''}+\alpha P_{X,Y})- 3|\mathcal{X}|^2|\mathcal{Y}|\frac{\log{(n+1)}}{n}- 12|\mathcal{X}||\mathcal{Y}|\frac{\log{(n+1)}}{n} -2\delta_{\epsilon})\Big)},
\end{align*}
where $\underline{t}_{X,6}(x)= \frac{1}{\alpha}(P_X(x)-(1-\alpha)\underline{t}'_X(x)), x\in \mathcal{X}$ by construction, and the additional $\delta_{\epsilon}$ term is the penalty for replacing $D(\underline{t}_{6,X}||P_X)$ by $D(\underline{t}''_{X}||P_X)$ as derived in Proposition \ref{Prop:Ap:3} (note that $\delta_{\epsilon}$ is an upper bound to the penalty term). 
This completes the proof.}

\section{Proof of Theorem \ref{th:1}}
\label{app:th1}
{Define the following partition for the set of indices $[1,n]$:}

\begin{align*}
&{\mathcal{A}_1= \{k| \exists \ell\in [c]: \sum_{j=1}^{\ell-1} i_j +1\leq k< \sum_{j=1}^{\ell} i_j, \text{ and }
k-\sum_{j=1}^{\ell-1} i_j \text{ is odd}\}}
\\&
{\mathcal{A}_2= \{k| \exists \ell\in [c]:  k= \sum_{j=1}^{\ell} i_j, \text{ and }
i_{\ell} \text{ is odd}\}}
\\&
{\mathcal{A}_3= \{k| \exists \ell\in [c]: \sum_{j=1}^{\ell-1} i_j +1\leq k\leq  \sum_{j=1}^{\ell} i_j, \text{ and }
k-\sum_{j=1}^{\ell-1} i_j \text{ is even}\}}
\\&{\mathcal{A}_4= \{k| k>\sum_{i=1}^c i_j\}.}
\end{align*}
{The set $\mathcal{A}_1$ is the set of 
odd elements of each cycle which are not endpoints, i.e the first, third, fifth $\dots$ elements of each cycle if their are not at the end of the cycle.
The set $\mathcal{A}_2$ contains the indices which are endpoints of cycles with odd lengths. The set 
 $\mathcal{A}_3$
is the set of 
even elements of each cycle, i.e the second, fourth, sixth, $\dots$ elements of each cycle. The set $\mathcal{A}_4$ is the set of fixed indices. The vectors $(X_i,Z_i), i\in \mathcal{A}_j, j\in [4]$ are four collections of i.i.d sequences, where $(X_i,Z_i), i\in \mathcal{A}_j, j\in [3]$ are generated with $P_XP_Y$ and $(X_i,Z_i),i \in  \mathcal{A}_4$ are generated with $P_XP_Y$.}

{Let $\underline{T}_j, j\in [4]$ be the type of the sequence $(X_i,Z_i), i\in \mathcal{A}_j, j\in [4] $, so that $\underline{T}_{j,x,y}=$ $ \frac{\sum_{i\in \mathcal{A}_j}\mathbbm{1}(X_i=x, Z_i=y)}{|\mathcal{A}_j|},$ $ j,x,y\in [4]\times\mathcal{X}\times  \mathcal{Y}$. We are interested in the probability of the event $(X^n,Z^n)\in \mathcal{A}_{\epsilon}^n(X,Y)$. From Definition \ref{Def:typ} this event can be rewritten as follows:
\begin{align*}
 &P\left(\left(X^n,Z^n\right)\in \mathcal{A}_{\epsilon}^n(X,Y)\right)
  =P\left(\underline{T}(X^n,Y^n)\stackrel{.}{=} P_{X,Y}(\cdot,\cdot)\pm \epsilon\right)
 \\&= P(\alpha_1\underline{T}_1+\alpha_2\underline{T}_2+\alpha_3\underline{T}_3+\alpha_4\underline{T}_4\stackrel{.}{=}P_{X,Y}(\cdot,\cdot)\pm \epsilon),
\end{align*}
where $\alpha_j= \frac{|\mathcal{A}_j|}{n}, j\in [4]$. We have: 
\begin{align*}
 P((X^n,Z^n)\in \mathcal{A}_{\epsilon}^n(X,Y))
 &=\sum_{(\underline{t}_1,\underline{t}_2,\underline{t}_3,\underline{t}_4)\in \mathcal{T}} P(\underline{T}_j=\underline{t}_j, j\in [4])
\\&=
\sum_{(\underline{t}_1,\underline{t}_2,\underline{t}_3,\underline{t}_4)\in \mathcal{T}} P(\underline{T}_j=\underline{t}_j, j\in [3])P(\underline{T}_4=\underline{t}_4)
\\&\leq  \sum_{(\underline{t}_1,\underline{t}_2,\underline{t}_3,\underline{t}_4)}
\sqrt[3]{P(\underline{T}_1=\underline{t}_1)P(\underline{T}_2=\underline{t}_2)P(\underline{T}_3=\underline{t}_3)}P(\underline{T}_4=\underline{t}_4),
\end{align*}
where $\mathcal{T}= \{(\underline{t}_1,\underline{t}_2,\underline{t}_3,\underline{t}_4):\alpha_1\underline{t}_1+\alpha_2\underline{t}_2+\alpha_3\underline{t}_3+\alpha_4\underline{t}_4\stackrel{.}{=}P_{X,Y}(\cdot,\cdot)\pm \epsilon\}$, and to get the last inequality, we have used the relation $P(\underline{T}_j=\underline{t}_j, j\in [3])= \sqrt[3]{P^3(\underline{T}_j=\underline{t}_j, j\in [3])}\leq \sqrt[3]{\prod_{j=1}P(\underline{T}_j=\underline{t}_j)}$. Using Proposition \ref{Prop:Ap:1} in the proof of Theorem \ref{th:1:improved},
we have, 
\begin{align*}
    &P(\underline{T}_j=\underline{t}_j)\leq \textcolor{black}{\exp_2\Big(-n\alpha_j D(\underline{t}_j||P_XP_Y)\Big)}, j\in [3]\\
    &
    P(\underline{T}_4=\underline{t}_4)\leq \textcolor{black}{\exp_2\Big(-n\alpha_4 D(\underline{t}_4||P_{X,Y})\Big)}
\end{align*}
So,
\begin{align}
 &\nonumber{ P((X^n,Z^n)\in \mathcal{A}_{\epsilon}^n(X,Y))}
 \\&{ \leq \!\!\!\!\!\!\!\! \sum_{(\underline{t}_1,\underline{t}_2,\underline{t}_3,\underline{t}_4)\in \mathcal{T}}\!\!\!\!\!\!\!\! \sqrt[3]{\textcolor{black}{\exp_2\Big(-n(\alpha_1D(\underline{t}_1||P_XP_Y)+\alpha_2D(\underline{t}_2||P_XP_Y)+\alpha_3D(\underline{t}_3||P_XP_Y)\Big)}}\textcolor{black}{\exp_2\Big(-n\alpha_4D(\underline{t}_4||P_{X,Y})\Big)} }
 \label{eq:cor:proof}
 \\&\nonumber{ \leq \sum_{(\underline{t}_1,\underline{t}_2,\underline{t}_3,\underline{t}_4)\in \mathcal{T}} \sqrt[3]{\textcolor{black}{\exp_2\Big(-n(1-\alpha_4)D(\frac{\alpha_1\underline{t}_1+\alpha_2\underline{t}+\alpha_3\underline{t}_3}{1-\alpha_4}||P_XP_Y)\Big)}}\textcolor{black}{\exp_2\Big(-n\alpha_4D(\underline{t}_4||P_{X,Y})\Big)}},
 \end{align}
 where we have used the convexity of KL divergence and the fact that $\alpha_1+\alpha_2+\alpha_3=1-\alpha_4$.
 Define $\underline{t}'= \frac{\alpha_1\underline{t}_1+\alpha_2\underline{t}_2+\alpha_3\underline{t}_3}{1-\alpha_4}$, and let $\underline{t}''(x,y)= \frac{1}{\alpha}(P_{X,Y}(x,y)- (1-\alpha)\underline{t}'_{X,Y}(x,y)), x,y \in \mathcal{X}\times \mathcal{Y}$. Then, using Proposition \ref{Prop:Ap:3}, we have:
 \begin{align*}
 &{ P((X^n,Z^n)\in \mathcal{A}_{\epsilon}^n(X,Y))}
 \\& { \leq  \sum_{(\underline{t}_1,\underline{t}_2,\underline{t}_3,\underline{t}_4)\in \mathcal{T}} \sqrt[3]{\textcolor{black}{\exp_2\Big(-n(1-\alpha_4)D(\underline{t}'||P_XP_Y)\Big)}}\textcolor{black}{\exp_2\Big(-n\alpha_4(D(\underline{t}''||P_{X,Y})-\delta_{\epsilon})\Big)} }
\\&
{\leq
 \sqrt[3]{\textcolor{black}{\exp_2\Big(-n(1-\alpha_4)D(\underline{t}'||P_XP_Y)\Big)}}
 \textcolor{black}{\exp_2\Big(-n\alpha_4(D(\underline{t}''||P_{X,Y})-\delta_{\epsilon}-3|\mathcal{X}||\mathcal{Y}|\frac{\log{n+1}}{n})\Big)}
 },
 \end{align*}
 where the last inequality follows from the fact that the number of joint types of $n$-length pairs of sequences on alphabet $\mathcal{X}\times \mathcal{Y}$ is bounded from above by $(n+1)^{|\mathcal{X}||\mathcal{Y}|}$, so that each of $\underline{t}_j, j \in [3]$ can take at most $(n+1)^{|\mathcal{X}||\mathcal{Y}|}$ values. The term $\delta_{\epsilon}$ is an upper-bound on the variation in the value of KL divergence due to replacing $D(\underline{t}_4||P_{X,Y})$, where $(1-\alpha_4)\underline{t}'+\alpha_4\underline{t}_4\stackrel{.}{=} P_{X,Y}\pm \epsilon$ by  the term $D(\underline{t}''||P_{X,Y})$, where $(1-\alpha_4)\underline{t}'+\alpha_4\underline{t}''= P_{X,Y}$. The upper-bound can be derived similar to the same upper-bound in the proof of Theorem \ref{th:1:improved}.
 This completes the proof.}
 
 \section{Proof of Corollary \ref{th:cor:1}}
 \label{app:cor:th1}
{From Equation \eqref{eq:cor:proof} in the proof of Theorem \ref{th:1}, we have:
 \begin{align*}
 &\nonumber{ P((X^n,Z^n)\in \mathcal{A}_{\epsilon}^n(X,Y))}
 \\&{ \leq \!\!\!\!\!\!\!\! \sum_{(\underline{t}_1,\underline{t}_2,\underline{t}_3,\underline{t}_4)\in \mathcal{T}}\!\!\!\!\!\!\!\! \sqrt[3]{\textcolor{black}{\exp_2\Big(-n(\alpha_1D(\underline{t}_1||P_XP_Y)+\alpha_2D(\underline{t}_2||P_XP_Y)+\alpha_3D(\underline{t}_3||P_XP_Y)+\alpha_4D(\underline{t}_4||P_{X,Y})\Big)}}}
 \\
 &{ 
 \stackrel{(a)}{\leq}
 \sum_{(\underline{t}_1,\underline{t}_2,\underline{t}_3,\underline{t}_4)\in \mathcal{T}} \sqrt[3]{\textcolor{black}{\exp_2\Big(-n(D(\alpha_1\underline{t}_1+\alpha_2\underline{t}_2+\alpha_3\underline{t}_3+\alpha_4\underline{t}_4
 ||(\alpha_1+\alpha_2+\alpha_3)P_XP_Y+ \alpha_4P_{X,Y})\Big)}}}
 \\&
{  \leq
 |\mathcal{T}| \sqrt[3]{\textcolor{black}{\exp_2\Big(-n(D(P_{X,Y}
 ||(1-\alpha)P_XP_Y+ \alpha P_{X,Y})-\delta_{\epsilon}\Big)}}}
 \\&
 { \stackrel{(b)}{\leq} \textcolor{black}{\exp_2\Big(-\frac{n}{3}(D(P_{X,Y}
 ||(1-\alpha)P_XP_Y+ \alpha P_{X,Y})-\delta_{\epsilon}+12|\mathcal{X}||\mathcal{Y}|\log{\frac{n+1}{n}})\Big)},}
\end{align*}
where (a) follows from the convexity of the divergence function and (b) follows by the fact that the number of joint types of $n$-length pairs of sequences on alphabet $\mathcal{X}\times \mathcal{Y}$ is bounded from above by $(n+1)^{|\mathcal{X}||\mathcal{Y}|}$, so that each of $\underline{t}_j, j \in [4]$ can take at most $(n+1)^{|\mathcal{X}||\mathcal{Y}|}$ values.
\qedsymbol}

\section{Proof of Lemma \ref{lem:dercount}}
\label{app:lem2}
First, we prove Equation \eqref{eq:der1}. Note that 
\begin{align*}
    N_m= {n \choose m} !(n-m)\leq  {n \choose m} (n-m)!= \frac{n!}{m!}\leq n^{n-m}.
\end{align*}
This proves the right hand side of the equation. To prove the left hand side, we first argue that the iterative inequality $!n\geq !(n-1)(n-1)$ holds. In other words, the number of derangements of numbers in the interval $[n]$ is at least $n-1$ times the number of derangements of the numbers in the interval $[n-1]$. We prove the statement by constructing $!(n-1)(n-1)$ distinct derangements of the numbers $[n]$. Note that a derangement $\pi(\cdot)$ of $[n]$ is characterized by the vector $(\pi(1),\pi(2),\cdots(n))$. There are a total of $n-1$ choices for $\pi(1)$ (every integer in $[n]$ except for $1$). Once $\pi(1)$ is fixed, the rest of the vector $(\pi(2),\pi(3),\cdots, \pi(n))$ can be constructed using any derangement of the set of numbers $[n]- \{\pi(1)\}$. There are a total of $!(n-1)$ such derangements. So, we have constructed $!(n-1)(n-1)$ distinct derangements of $[n]$. Consequently. $!n\geq !(n-1)(n-1)$. By induction, we have $!n\geq (n-1)!$. So,
\begin{align*}
    N_m= {n \choose m} !(n-m)\geq {n \choose m} (n-m-1)! = \frac{n!}{m!(n-m)}. 
\end{align*}
Next, we prove that Equation \eqref{eq:der2} holds. Note that from the right hand side of Equaation \eqref{eq:der1} we have:
\begin{align*}
    \lim_{n\to \infty}\frac{\log{N_m}}{n\log{n}}\leq \lim_{n\to \infty} \frac{\log{n^{n-m}}}{n\log{n}}= 
    \lim_{n\to \infty}\frac{n-m}{n}=1-\alpha.
\end{align*}
Also, from the left hand side of Equation \eqref{eq:der2}, we have: 
\begin{align*}
     &\lim_{n\to \infty} \frac{\log{N_m}}{n\log{n}}
     \geq \lim_{n\to \infty}\frac{\log{ 
     \frac{n!}{m!(n-m)}}}{n\log{n}}
    %  \\&
    %  \geq
    %   \lim_{n\to \infty}\frac{\log{ \frac{n(n-1)\cdots(m+1)}{(n-m)}}}{n\log{n}}
      = 
      \lim_{n\to \infty}\frac{\log{ {\frac{n!}{m!}}}}{n\log{n}}-
     \frac{{{\log{ (n-m)}}}}{n\log{n}}.
\end{align*}
The second term in the last inequality converges to 0 as $n\to \infty$. Hence, 
\begin{align*}
    &\lim_{n\to \infty} \frac{\log{N_m}}{n\log{n}}
    \geq 
     \lim_{n\to \infty}\frac{\log{ {\frac{n!}{m!}}}}{n\log{n}}
     \\& \stackrel{(a)}{\geq} 
      \lim_{n\to \infty}\frac{\log{ {\frac{n!}{m^m}}}}{n\log{n}}
      {\geq} \lim_{n\to\infty}
     \frac{\log{{n!}}}{n\log{n}}
     -\frac{\log{{m^m}}}{n\log{n}}
     \stackrel{(b)}{\geq} \lim_{n\to\infty}
     \frac{n\log{n}-n+O(\log{n})}{n\log{n}}
     -\frac{\log{{m^m}}}{n\log{n}}
     \\ &
     =\lim_{n\to\infty}
     \frac{n\log{n}}{n\log{n}}
     -\frac{\alpha n\log{{\alpha n}}}{n\log{n}}
     = 1-\alpha,
\end{align*}
where in (a) we have used the fact that $m!\leq m^m$, and (b) follows from Stirling's approximation. This completes the proof.
\qedsymbol

\section{Proof of Theorem \ref{th:cperm}}
\label{app:thcperm}
The proof builds upon the arguments provided in the proof of Theorem \ref{th:1}. Let $Y^{n}=\pi_l(X_{(l)} ^n)_{l\in [k]}$.
First, we construct a partition $\mathsf{D}=\{\mathcal{C}_{j,t}:j\in [b_k], t\in [k(k-1)]\}$ such that each sequence of vectors $(Y_{(l), \mathcal{C}_{j,t}})_{l\in [k]}$ is an collection of independent vectors of i.i.d variables, where $Y_{(l), \mathcal{C}_{j,t}}=(Y_{(l),c})_{c\in  \mathcal{C}_{j,t}}$. Loosely speaking, this partitioning of the indices `breaks' the multi-letter correlation among the sequences induced due to the permutation and allows the application of standard information theoretic tools to bound the probability of joint typicality. The partition is constructed in two steps.  
We first construct a \textit{coarse} partition $\mathsf{C}=\{\mathcal{C}_1,\mathcal{C}_2,\cdots, \mathcal{C}_{b_k}\}$ of the indices $[1,n]$ for which the sequence of vectors $(Y_{(l),\mathcal{C}_j}), l\in [k]$ is identically distributed but not necessarily independent. The set $\mathcal{C}_j, j\in [b_k]$ is defined as the set of indices corresponding to partition $\mathcal{P}_j$, where correspondence is defined in Definition \ref{Def:corr}. {We assume that $\mathcal{P}_{b_k}$ is the single-set partition. So that $\mathcal{C}_{b_k}$ is the set of joint fixed points of the permutations.  It is straightforward to see that $(Y_{(l),\mathcal{C}_{b_k}}), l\in [k]$ are i.i.d with joint distribution $P_{X}^k$. Next, we focus on $(Y_{(l),\mathcal{C}_{j}}), l\in [k], j\in [b_k-1]$.}
Clearly, $\mathsf{C}=\{\mathcal{C}_1,\mathcal{C}_2,\cdots, \mathcal{C}_{b_k}\}$ partitions $[1,n]$ since each index corresponds to exactly one partition $\mathcal{P}_j$. To verify that the elements of the sequence $(Y_{(l),\mathcal{C}_j}), l\in [k]$ are identically distributed let us consider a fixed $j\in [b_k-1]$ and an arbitrary index $c\in \mathcal{C}_j$. Then the vector $(Y_{(1),c},Y_{(2),c},\cdots,Y_{(k),c})$ is distributed according to $P_{X_{\mathcal{P}_j}}$. To see this, note that:
\begin{align*}
    P_{Y_{(1),c},Y_{(2),c},\cdots,Y_{(k),c}}
    &= P_{X_{(1),(\pi_1^{-1}(c))},X_{(2),(\pi_2^{-1}(c))},\cdots,X_{(k),(\pi_k^{-1}(c))}}
\end{align*}
From the assumption that the index $c$ corresponds to the partition $\mathcal{P}_j$, we have that $\pi_l^{-1}(c)= \pi_{l'}^{-1}(c)$ if and only if $l,l'\in \mathcal{A}_{j,r}$ for some integer $r\in [|\mathcal{P}_j|]$. Since by the theorem statement $(X^n_{(l)})_{l\in [k]}$ is an i.i.d sequence of vectors, the variables $X_{(l),\pi^{-1}_l(c)}$ and $X_{(l'),\pi^{-1}_{l'}(c)}$ are independent of each other if $\pi_l^{-1}(c)\neq \pi_{l'}^{-1}(c)$.  Consequently, 
\begin{align*}
    P_{Y_{(1),c},Y_{(2),c},\cdots,Y_{(k),c}}
    &= \prod_{r\in [|\mathcal{P}_j|]}P_{X_{t_1},X_{t_2},\cdots,X_{t_{|\mathcal{A}_{j,r}|}}}=P_{X_{\mathcal{P}_j}}.
\end{align*}
This proves that the sequences $(Y_{(l),\mathcal{C}_j}), l\in [k]$ are identically distributed with distribution $P_{X_{\mathcal{P}_j}}$.
In the next step, we decompose the partition $\mathsf{C}$ to arrive at a finer partition $\mathsf{D}=\{\mathcal{C}_{j,t}:j\in [b_k], t\in [k(k-1)]\}$ of $[1,n]$ such that $(Y_{(l),\mathcal{C}_{j,t}})_{l\in [k]}$ is an i.i.d sequence of vectors. Let $\mathcal{C}_j=\{c_1,c_2,\cdots,c_{|\mathcal{C}_j|}\}, j\in [b_k-1]$.
%By definition, we have $d_{j,r}=d_{j',r}$ if and only if $j,j'$ belong to the same partition element in $\mathcal{P}_k$.
% Let $\mathcal{P}_k= \{\mathcal{A}_{k,1},\mathcal{A}_{k,2},\cdots,\mathcal{A}_{k,|\mathcal{P}_{k}|}\}$. 
% Define $e_{k,l}$ as the minimum index in  $\mathcal{A}_{k,l}, l\in [|\mathcal{P}_k|]$:
% \begin{align*}
%     e_{k,l}= \min \{i: i\in \mathcal{A}_{k,l}\}, l\in [|\mathcal{P}_k|]. 
% \end{align*}
% Furthermore, define $d_{j,r}$ as the inverse image of $c_{j,r}$ with respect to the permutation $\pi_j$:
% \begin{align*}
%     d_{j,r}\triangleq \pi^{-1}_{j}(c_r), j\in [m], r\in [|\mathcal{C}_k|].
% \end{align*}
%The objective is to further partition $\mathcal{C}_k$ into $\{\mathcal{C}_{k,l}:l\in [m(m-1)]\}$ such that the sequence of vectors $(Y_{(j),\mathcal{C}_{k,l}})_{j\in [m]}$ is an i.i.d. sequence. 
The previous step shows that the sequence consists of identically distributed vectors. In order to guarantee independence, we need to ensure that for any $c,c'\in \mathcal{C}_{j,t}$, we have $\pi^{-1}_{l}(c)
\neq \pi^{-1}_{l'}(c'), \forall l,l' \in [k]$. Then, independence of $(Y_{(l),c})_{l\in [k]}$ and $(Y_{(l),c'})_{l\in [k]}$ is guaranteed due to the independence of the sequence of vectors $(X^n_{(l)})_{l\in [k]}$. To this end we assign the indices in $\mathcal{C}_j$ to the sets $\mathcal{C}_{j,t}, t\in [k(k-1)+1]$ as follows:
\begin{align}
    &c_1\in \mathcal{C}_{j,1},\\
    &c_i\in {\mathcal{C}_{j,t}}: t= 
    \min \{t'| \nexists c'\in \mathcal{C}_{j,t'}, l,l' \in [k]: \pi^{-1}_l(c_i)=\pi^{-1}_{l'}(c')\}, i>1.
    \label{eq:assign}
\end{align}
Note that the set $\mathcal{C}_{j,t}$ defined in Equation \eqref{eq:assign} always exists based on the following argument. {Let $\mathcal{E}=\{c'\in \mathcal{C}_j-\{c_i\}|\exists l,l'\in [k]: \pi^{-1}_{l}(c_i)= \pi_{l'}^{-1}(c')\}$, and let $\mathcal{E}'=\{c'\in [n]- \{c_i\}|\exists l,l'\in [k]: \pi^{-1}_{l}(c_i)= \pi_{l'}^{-1}(c')\}$. Note that $\mathcal{E}\subset \mathcal{E}'$. We claim that $|\mathcal{E}|\leq |\mathcal{E}'|\leq k(k-1)$.  To see this,  note that there are at most $k$ distinct values $\pi_{l}^{-1}(c_i), l\in k$, and for each of these values, there are at most $(k-1)$ distinct values $l'\neq l$ for which  $\pi_{l}^{-1}(c_i)=\pi_{l'}^{-1}(c)$ since each permutation is a one to one map.  
 Since there are a total of $k(k-1)+1$ sets $\mathcal{C}_{j,t}$, by the Pigeonhole Principle, there exists at least one set for which there is no element $c'$ such that $\pi_l(c)=\pi_{l'}(c')$ for any value of $l,l'$.} Consequently, $(Y_{(l),\mathcal{C}_{j,t}})_{l\in [k]}$ is an i.i.d. sequence with distribution $P_{X_{\mathcal{P}_{j}}}$.

Let $\underline{T}_{j,t}, j\in [b_k-1], t\in [k(k-1)+1]$ be the type of the sequence of vectors $(Y_{(l),\mathcal{C}_{j,t}})_{l\in [k]}$, so that $\underline{T}_{j,t}(x^k)= \frac{\sum_{c\in \mathcal{C}_{j,t}}\mathbbm{1}((Y_{(1),c},Y_{(2),c},\cdots, Y_{(k),c})=x^k)}{|\mathcal{C}_{j,t}|}, x^k\in \mathcal{X}^k$. {Furthermore, let $\underline{T}_{b_k}$ be the type of the sequence of vectors $(Y_{(l),\mathcal{C}_{b_k}})_{l\in [k]}$.}
We are interested in the probability of the event $(Y^n_{(l)})_{l\in [k]}\in \mathcal{A}_{\epsilon}^n(X^k)$. From Definition \ref{Def:ctyp} this event can be rewritten as follows:
\begin{align*}
 &P\left(\left(Y^n_{(l)})_{l\in [k]}\right)\in \mathcal{A}_{\epsilon}^n(X^k)\right)
  =P\left({T}((Y^n_{(l)})_{l\in [k]},x^m)\stackrel{.}{=} P_{X^k}(x^k)\pm \epsilon, \forall x^k\right)
 \\&= {P(\sum_{j,t}\alpha_{j,t}\underline{T}_{j,t}(x^k)+{\alpha_{b_k}}\underline{T}_{b_k}(x^k)\stackrel{.}{=}P_{X^k}(x^k)\pm \epsilon, \forall x^k)},
\end{align*}
where $\alpha_{j,t}= \frac{|\mathcal{C}_{j,t}|}{n}, j\in [b_k-1], t\in [k(k-1)+1]$, and $\alpha_{b_k}= \frac{|\mathcal{C}_{b_k}|}{n}$. We have: 
\begin{align*}
 &P\left(\left(Y^n_{(l)})_{l\in [k]}\right)\in \mathcal{A}_{\epsilon}^n(X^k)\right)
 =\sum_{(\underline{s}^{b_k-1,k(k-1)+1},\underline{s}_{b_k})\in \mathcal{T}} P(\underline{T}_{j,t}=\underline{s}_{j,t},j\in [b_k-1], t\in [k(k-1)+1], \underline{T}_{b_k}= \underline{s}_{b_k}),
\end{align*}
where $\mathcal{T}= \{(\underline{s}^{b_k,k(k-1)+1}, \underline{s}_{b_k}):P(\sum_{j,t}\alpha_{j,t}{s}_{j,t}(x^k)+{\alpha_{b_k}}\underline{s}_{b_k}(x^k)\stackrel{.}{=}P_{X^k}(x^k)\pm \epsilon, \forall x^k), \forall x^k\}$. Using the property that for any set of events, the probability of the intersection is less than or equal to the geometric average of the individual probabilities, we have:
\begin{align*}
 &P((Y^n_{(l)})_{l\in [k]}\in \mathcal{A}_{\epsilon}^n(X^k))
\leq P(\underline{T}_{b_k}=\underline{s}_{b_k})
\sum_{(\underline{s}^{b_k-1,k(k-1)+1},\underline{s}_{b_k})\in \mathcal{T}} \sqrt[(k(k-1)+1)(b_k-1)]{\Pi_{i\in [j,t]}P(\underline{T}_{j,t}=\underline{s}_{j,t})}.
\end{align*}
Since the elements  $(Y_{(l),\mathcal{C}_{j,t}}), j\in [b_k], t\in [k(k-1)]$ are i.i.d by construction, it follows from standard information theoretic arguments \cite{csiszarbook} that:

\begin{align*}
    & P(\underline{T}_{j,t}=\underline{s}_{j,t}) \leq \textcolor{black}{\exp_2\Big(-|\mathcal{C}_{j,t}|(D(\underline{s}_{j,t}||P_{X_{\mathcal{P}_j}})-\prod_{l\in [k]}|\mathcal{X}_l|\epsilon)\Big)}, j\in [b_k-1], t\in [k(k-1)+1]\\
    &P(\underline{T}_{b_k}=\underline{s}_{b_k}) \leq \textcolor{black}{\exp_2\Big(-|\mathcal{C}_{b_k}|(D(\underline{s}_{b_k}||P_{X_{\mathcal{P}_j}})-\prod_{l\in [k]}|\mathcal{X}_l|\epsilon)\Big)}.
\end{align*}
We have, 
\begin{align*}
 &P((Y^n_{(l)})_{l\in [k]}\in \mathcal{A}_{\epsilon}^n(X^k))
 \leq  \textcolor{black}{\exp_2\Big(-|\mathcal{C}_{b_k}|(D(\underline{s}_{b_k}||P_{X_{\mathcal{P}_j}})-\prod_{l\in [k]}|\mathcal{X}_l|\epsilon)\Big)}\times \\&\!\!\!\!\!\!\!\!
 \sum_{(\underline{s}^{b_k-1,k(k-1)+1},\underline{s}_{b_k})\in \mathcal{T}}\!\!\!\!\!\!\!\! \sqrt[(k(k-1)+1)(b_k-1)]{\Pi_{i\in [j,t]} \textcolor{black}{\exp_2\Big(-|\mathcal{C}_{j,t}|(D(\underline{s}_i||P_{X_{\mathcal{P}_j}})-\prod_{l\in [k]}|\mathcal{X}_l|\epsilon)\Big)}}
 \\
 &\stackrel{(a)}{\leq}
\sum_{(\underline{s}^{b_k-1,k(k-1)+1},\underline{s}_{b_k})\in \mathcal{T}}
 \sqrt[(k(k-1)+1)(b_k-1)]
 {\textcolor{black}{\exp_2\Big(-n(D(\sum_{j,t}\alpha_{j,t}\underline{s}_{j,t}+ \alpha_{b_k}\underline{s}_{b_k}
 ||\sum_{k}P_{X_{\mathcal{P}_j}})-\prod_{l\in [k]}|\mathcal{X}_l|\epsilon)\Big)}}
 \\&\stackrel{(b)}{\leq}  {\textcolor{black}{\exp_2\Big(-\frac{n}{(k(k-1)-1)(b_k+1)}(D(P_{X^k}
 ||\sum_{j\in [b_k]}\frac{|\mathcal{C}_j|}{n}
 P_{X_{\mathcal{P}_j}})+O(\epsilon)+O(\frac{\log{n}}{n}))\Big)}.}
\end{align*}
where the (a) follows from the convexity of the divergence function and (b) follows by the fact that the number of joint types grows polynomially in $n$.

\section{Proof of Lemma \ref{lem:foldcount}}
\label{Ap:lem:foldcount}
The upper-bound follows by the fact that for $k$-fold derangement $(\pi_1(\cdot),\pi_2(\cdot),\cdots,\pi_k(\cdot))$, the first permutation is $\pi_1(\cdot)$ is the identity permutation, and the rest of derangements with respect to $\pi_1(\cdot)$, so by the counting principle there are at most $(!n)^{k-1}$ choices for $(\pi_1(\cdot),\pi_2(\cdot),\cdots,\pi_k(\cdot))$. Next we prove the lower bound. Note that $\pi_1(\cdot)$ is the identity permutation. By the same arguments as in the proof of Lemma \ref{lem:dercount}, there are at least $(n-1)!$ choices of distinct $\pi_2(\cdot)$, and for any fixed $\pi_2(\cdot)$ there are at least $(n-2)!$ distinct $\pi_3(\cdot)$. Generally, for fixed $\pi_2(\cdot),\pi_3(\cdot),\cdots, \pi_j(\cdot)$, there are at least $(n-j+1)!$ choices of distinct $\pi_{j+1}(\cdot)$. By the counting principle, there are at least $\prod_{j\in [k]}(n-j+1)!\geq ((n-k+1)!)^k$ distinct $(\pi_1(\cdot),\pi_2{\cdot},\cdots,\pi_{k}(\cdot))$. This completes the proof.
\qedsymbol

\section{Proof of Lemma \ref{lem:bell}}
\label{app:lem6}
First, we prove the upper-bound in Equation \eqref{eq:bell1}. As an initial step, we count the number of distinct allocations of partition correspondence to indices $i\in [1,n]$. Since we are considering $(i_1,i_2,\cdots, i_{b_k})$-Bell permutation vectors, there are a total of $i_j$ indices corresponding to $\mathcal{P}_j$ for $j\in [b_k]$. So, there are ${n \choose i_1,i_2,\cdots, i_{b_k}}$ allocations of partition correspondence to different indices. Now assume that the $i^{th}$ index corresponds to the $j$th partition. Then, we argue that there are at most $n^{|\mathcal{P}_j|}$ possible values for the vector $(\pi_{l}(i): l\in [k])$. The reason is that by definition, for any two $\pi_l(i)$ and $\pi_{l'}(i)$, their value are equal if and only if $l,l' \in \mathcal{D}_{j,k}$ for some integer $r\in [|\mathcal{P}_j|]$. So, the elements of $(\pi_{l}(i): l\in [m])$ take $|\mathcal{P}_j|$ distinct values among the set $[1,n]$. Consequently $(\pi_{l}(i): l\in [k])$ takes at most $n^{|\mathcal{P}_j|}$ distinct values. By the counting principle, the sequence of vectors $(\pi_{l}(i): l\in [k]), i\in [n]$ takes at most $n^{\sum_{j\in [b_k]} |\mathcal{P}_j|i_j-n}$  distinct values given a specific partition correspondence, since $\pi_1(\cdot)$ is assumed to be the identity permutation. Since there are a total of ${n \choose i_1,i_2,\cdots, i_{b_k}}$ partition correspondences, we have:   
\begin{align*}
    N_{i_1,i_2,\cdots,i_{b_k}} \leq {n \choose i_1,i_2,\cdots, i_{b_k}} n^{\sum_{j\in [b_k]} |\mathcal{P}_j|i_j-n}.
\end{align*}
Next, we prove the lower-bound in Equation \eqref{eq:bell1}. The proof follows by constructing enough distinct $(i_1,i_2,\cdots, i_{b_k})$-Bell permutation vectors. First, we choose a partition correspondence for the indices $i\in [n]$ similar to the proof for the lower-bound. There are ${n \choose i_1,i_2,\cdots, i_{b_k}}$ distinct ways of allocating the partition correspondence. We argue that for every fixed partition correspondence, there are at least $\prod_{j\in [b_k]]}d_{|\mathcal{P}_{j}
|}(i_j)$ permutations which are $(i_1,i_2,\cdots, i_{b_k})$-Bell permutation vectors. To see this,
without loss of generality, assume that the first $i_1$ indices $[1,i_1]$ correspond to $\mathcal{P}_1$, the next $i_2$ indices $[i_1+1,i_1+i_2]$ correspond to $\mathcal{P}_2$, and in general the indices $[\sum_{t=1}^{l-1}i_t+1, \sum_{t=1}^{l}i_t]$ correspond to $\mathcal{P}_j$. Let $(\pi'_{1,j},\pi'_{2,j},\cdots,\pi'_{|\mathcal{P}_j|,j})$ be vectors of $|\mathcal{P}_j|$-fold derangements of $[\sum_{t=1}^{j-1}i_t+1, \sum_{t=1}^{j}i_t]$, where $j\in [b_k]$. Then, the following is an $(i_1,i_2,\cdots,i_{b_k})$-Bell permutation vector. 
\begin{align*}
    \pi_l([\sum_{t=1}^{j-1}i_t+1, \sum_{t=1}^{j}i_t])= \pi'_{l,j}([\sum_{t=1}^{j-1}i_t+1, \sum_{t=1}^{j}i_t]), \quad \text{if } l\in \mathcal{D}_{s,j}, s\in [|\mathcal{P}_j|], j\in [b_k].
\end{align*}
There are a total of $d_{|\mathcal{P}_j|(i_j)}$ choices of $(\pi'_{1,j},\pi'_{2,j},\cdots,\pi'_{|\mathcal{P}_j|,j})$. So, by the counting principle, there are a total of $\prod_{j\in [b_k]}d_{|\mathcal{P}_j|}(i_j)$ choices of $(\pi_{1}(\cdot),\pi_2(\cdot),\cdots, \pi_{k}(\cdot))$ for a fixed partition correspondence. As argued previously, there are a total of  ${n \choose i_1,i_2,\cdots, i_{b_k}}$ distinct choices for partition correspondence. Consequently we have shown that, 
\begin{align*}
     {n \choose i_1,i_2,\cdots, i_{b_k}}\prod_{j\in [b_k]}d_{|\mathcal{P}_{j}|}(i_j)\leq N_{i_1,i_2,\cdots,i_{b_k}}.
\end{align*}
This completes the proof of Equation \eqref{eq:bell1}. We proceed with to prove Equation \eqref{eq:bell2}. Note that from the right hand side of Equation \eqref{eq:bell1}, we have:
\begin{align*}
 &\lim_{n\to\infty} \frac{\log_e{N_{i_1,i_2,\cdots,i_{b_k}}}}{n\log_e{n}}
 \leq
 \lim_{n\to\infty} \frac{\log_e{{n \choose i_1,i_2,\cdots, i_{b_k}} n^{(\sum_{j\in [b_k]} |\mathcal{P}_j|i_j-n)}}}{n\log_e{n}}
 = 
 \lim_{n\to\infty} \frac{\log_e{ n^{(\sum_{j\in [b_k]} |\mathcal{P}_j|i_j-n)}}}{n\log_e{n}}
 +\lim_{n\to\infty} \frac{\log_e{{n \choose i_1,i_2,\cdots, i_{b_k}}}}{n\log_e{n}}
 \\&
= \lim_{n\to\infty} \frac{{ ({\sum_{j\in [b_k]} |\mathcal{P}_j|i_j-n)}}}{n}
 +\lim_{n\to\infty} \frac{\log_e{2^n}}{n\log_e{n}}
 = {\sum_{j\in [b_k]} |\mathcal{P}_j|\alpha_j}-1.
\end{align*}
On the other hand, from the left hand side of Equation \eqref{eq:bell1}, we have:
\begin{align*}
    &\lim_{n\to\infty} \frac{\log{N_{i_1,i_2,\cdots,i_{b_k}}}}{n\log{n}}
    \geq \lim_{n\to\infty} \frac{\log{{n \choose i_1,i_2,\cdots, i_{b_k}}\prod_{j\in [b_k]}d_{|\mathcal{P}_{j}|}(i_j)}}{n\log{n}}
    \\&{\geq}
    {\lim_{n\to\infty} \frac{\log{\prod_{j\in [b_k]}d_{|\mathcal{P}_{j}|}(i_j)}}{n\log{n}}
    \stackrel{(a)}{\geq}
    \lim_{n\to\infty} \frac{\log{\prod_{j\in [b_k]}((i_j-|\mathcal{P}_j|+1)!^{|\mathcal{P}_j|-1})}}{n\log{n}}}
    \\& = 
    \lim_{n\to\infty} \frac{\sum_{j\in [b_k]}{(|\mathcal{P}_j|-1)}\log{(i_j-|\mathcal{P}_j|+1)!}}{n\log{n}}
    \\&\stackrel{(b)}{=}
    \lim_{n\to\infty} \frac{\sum_{j\in [b_k]}{(|\mathcal{P}_j|-1)}({(i_j-|\mathcal{P}_j|+1)\log{(i_j-|\mathcal{P}_j|+1)}-(i_j-|\mathcal{P}_j|+1)+O(\log{(i_j-|\mathcal{P}_j|+1)}))}}{n\log{n}}
    \\&= \sum_{j\in [b_k]}|\mathcal{P}_j|\alpha_j-1,
\end{align*}
where (a) follows from Lemma \ref{lem:foldcount}, and in (b) we have used Stirling's approximation.
\qedsymbol

\section{Proof of Theorem \ref{th:21}}
\label{app:th21}
 \begin{table}[]
  \centering
  \setlength\extrarowheight{5pt}
\begin{tabular}{|ll|ll|}
\hline
 $U^1_{\sigma^1}$: & upper-triangle of $g^1$ & $\widehat{\Sigma}$:  & Possible Output Labelings of TM \\ 
 \hline
  $U^2_{\sigma^2}$: & upper-triangle of $g^2$   &$\mathcal{E}$: & Labelings with at least $\alpha_n$ fraction incorrect    \\ \hline
 $\sigma^1,\sigma^2,\sigma^;$: & labelings & &\\ \hline
\end{tabular}
\vspace{0.1in}
\caption{Notation Table: Appendix \ref{app:th21}}
\label{tab:H}
\vspace{-0.2in}
\end{table}

First, note that for the correct labeling the two UTs are jointly typical with probability approaching one as $n\to \infty$:
\begin{align*}
P((U^{1}_{{\sigma}^1},U^{2}_{{\sigma}^2})\in \mathcal{A}_{\epsilon}^{\frac{n(n-1)}{2}}(X_1,X_2))\to 1 \quad \text{as}\quad n\to \infty.
\end{align*}
So, $P(\widehat{\Sigma}=\phi)\to 0$ as $n\to \infty$ since the correct labeling is a member of the set $\widehat{\Sigma}$.
We will show that the probability that a labeling in $\widehat{\Sigma}$ labels $n(1-\alpha_n)$ vertices incorrectly goes to $0$ as $n\to \infty$. 
Define the following:
\begin{align*}
 \mathcal{E}=\{{\sigma'}^2\Big| ||\sigma^2-{\sigma'}^2||_0\geq n(1-\alpha_n)\},
\end{align*}
where $||\cdot||_0$ is the $L_0$-norm. The set $\mathcal{E}$ is the set of all labelings which match more than $n\alpha_n$ vertices incorrectly. \textcolor{black}{We have summarized the notation in Table \ref{tab:H}.}
We show the following:
\begin{align*}
 P(\mathcal{E}\cap \widehat{\Sigma}\neq \phi)\to 0, \qquad \text{as} \qquad n\to \infty.
 \end{align*}
Note that:
\begin{align*}
  &
  {P(\mathcal{E}\cap \widehat{\Sigma}\neq \phi)
  = P\left(\bigcup_{{\sigma'}^2: ||\sigma^2-{\sigma'}^2||_0\geq n(1-\alpha_n)}\{{\sigma'}^2\in  \widehat{\Sigma}\}\right)}
  \stackrel{(a)}{\leq} {\sum_{i=0}^{n\alpha_n}\sum_{{\sigma'}^2: ||\sigma^2-{\sigma'}^2||_0=n-i} P({\sigma'}^2\in  \widehat{\Sigma})}
  \\&\stackrel{(b)}{=} {\sum_{i=0}^{n\alpha_n}\sum_{{\sigma'}^2: ||\sigma^2-{\sigma'}^2||_i=n-i}
   P((U^{1}_{{\sigma}^1},\Pi_{\sigma^2,{\sigma'}^2}(U^{2}_{{\sigma}^2}))\in \mathcal{A}_{\epsilon}^{\frac{n(n-1)}{2}})}
\\&\stackrel{(c)}{\leq} {\sum_{i=0}^{n\alpha_n}\sum_{{\sigma'}^2: ||\sigma^2-{\sigma'}^2||_i=n-i}
  \textcolor{black}{\exp_2\Big(-\frac{n(n-1)}{2}\left(\max\big\{E_{\frac{i(i-1)}{n(n-1)}}, E'_{\frac{i(i-1)}{n(n-1)}}\big\}-\zeta''_{\frac{n(n-1)}{2}}- \delta_{\epsilon}\right)
  \Big)}}
   \\&\stackrel{(d)}{=}  {\sum_{i=0}^{n\alpha_n} {n \choose i}(!(n-i))
  \textcolor{black}{\exp_2\Big(-\frac{n(n-1)}{2}\left(\max\big\{E_{\frac{i(i-1)}{n(n-1)}}, E'_{\frac{i(i-1)}{n(n-1)}}\big\}-\zeta''_{\frac{n(n-1)}{2}}- \delta_{\epsilon}\right)\Big)}}
  \\&\leq  {\sum_{i=0}^{n\alpha_n} n^{n-i}
 \textcolor{black}{\exp_2\Big(-\frac{n(n-1)}{2}\left(\max\big\{E_{\frac{i(i-1)}{n(n-1)}}, E'_{\frac{i(i-1)}{n(n-1)}}\big\}-\zeta''_{\frac{n(n-1)}{2}}- \delta_{\epsilon}\right)\Big)}}
  \\&\leq
  {\sum_{i=0}^{n\alpha_n} 
 \textcolor{black}{\exp_2\Big((n-i)\log{n}-\frac{n(n-1)}{2}\left(\max\big\{E_{\frac{i(i-1)}{n(n-1)}}, E'_{\frac{i(i-1)}{n(n-1)}}\big\}-\zeta''_{\frac{n(n-1)}{2}}- \delta_{\epsilon}\right)\Big)}.}
\end{align*}
{where we have defined $\zeta''_n= max(\zeta_n,\zeta'_n)$,} and (a) follows from the union bound, (b) follows from the definition of $ \widehat{\Sigma}$, in (c) we have used Theorems \ref{th:1:improved} and  \ref{th:1} and the fact that $||\sigma^2-{\sigma'}^2||_0=n-i$ so that $\Pi_{\sigma^2,{\sigma'}^2}$ has $\frac{i(i-1)}{2}$ fixed points, and in (d) we have denoted the number of derangement of sequences of length $i$ by $!i$. Note that the right hand side in the last inequality approaches 0 as $n\to \infty$ as long as:
\begin{align*}
   &{ (n-i+3)\log{n}\leq \frac{n(n-1)}{2}\left(\max\big\{E_{\frac{i(i-1)}{n(n-1)}}, E'_{\frac{i(i-1)}{n(n-1)}}\big\}-\zeta''_{\frac{n(n-1)}{2}}- \delta_{\epsilon}\right), i\in [0,n\alpha_n]}
 \\&{  \iff
 (1-\alpha)\log{n}\leq \frac{n-1}{2}\left(\max\big\{E_{\alpha^2}, E'_{\alpha^2}\big\}-\zeta''_{\frac{n(n-1)}{2}}- \delta_{\epsilon}\right),}
\end{align*}
where we have defined $\alpha=\frac{i}{n}$. The last equation is satisfied by the theorem assumption for small enough $\epsilon$ and large enough $n$ {by noting that $\zeta''_{\frac{n(n-1)}{2}}= O(\frac{\log{n}}{n^2})$ and $\delta_{\epsilon}= \epsilon o(\log{n})= o(\frac{\log{n}}{n})$ since $\max_{(x_1,x_2): P^{(n)}_{X_1,X_2}(x_1,x_2)\neq 0}|\log{\frac{P_{X_1}(x_1)P_{X_2}(x_2)}{P_{X_1,X_2}(x_1,x_2)}}|^+= o(\log{n})$ by assumption and $n\epsilon$ can be taken to go to infinity arbitrarily slowly for the probability of the typical set to approach one asymptotically. }
\qedsymbol

\section{Proof of Theorem \ref{th:ER}}
\label{app:th:ER}
{
Note that $\max_{(x_1,x_2): P^{(n)}_{X_1,X_2}(x_1,x_2)\neq 0}|\log{\frac{P_{X_1}(x_1)P_{X_2}(x_2)}{P_{X_1,X_2}(x_1,x_2)}}|^+= o(\log{n})$. So, From Theorem \ref{th:21}, it suffices to show that for a sequence $p_n,n \in \mathbb{N}$ such that
$
    \lim_{n\to \infty}\frac{\frac{\ln{n}}{n}}{p_n} > \frac{s}{2}$, we have
$2(1-\alpha)\frac{\ln{n}}{n}< 
\max\{E_{\alpha^2},E'_{\alpha^2}\},  0\leq \alpha\leq \alpha_n, \text{ as } n\to\infty$
for a given $\alpha_n$ such that $\alpha_n\to 1$ as $n\to \infty$, where the KL divergence terms in $E_{\alpha^2}$ and $E'_{\alpha^2}$ are evaluated in nats instead of bits.   Equivalently, we need to show that:
\begin{align}
\frac{\frac{\ln{n}}{n}}{p_n}< \frac{
\max\{E_{\alpha^2},E'_{\alpha^2}\}}{2(1-\alpha)p_n},  0\leq \alpha\leq \alpha_n, \text{ as } n\to\infty
\label{eq:success:1}
\end{align}}
{Let $p_n= \theta \frac{\log{n}}{n}$. We need to show that there exists $\theta<\frac{2}{s}$ such that  \eqref{eq:success:1} holds. It suffices to have:
\begin{align*}
    \frac{s}{2}<
    \lim_{n\to\infty}min_{\alpha\in [0,\alpha_n]}\frac{
\max\{E_{\alpha^2},E'_{\alpha^2}\}}{2(1-\alpha)p_n},
\end{align*}}
{since if the above holds we can take $\frac{1}{\theta}= \frac{s}{4}+ \frac{1}{2}\lim_{n\to\infty}min_{\alpha\in [0,\alpha_n]}\frac{
\max\{E_{\alpha^2},E'_{\alpha^2}\}}{2(1-\alpha)p_n}$ in which case both $\eqref{eq:success:1}$ and $\lim_{n\to\infty} \frac{\frac{\ln{n}}{n}}{p_n}>\frac{s}{2}$ hold. }

{We will prove that for $\alpha_0=0.8$ and $s<\frac{1}{2}$:
\begin{align*}
    \frac{s}{2}< \lim_{n\to\infty}min_{\alpha\in [0,\alpha_0]}\frac{
E_{\alpha^2}}{2(1-\alpha)p_n},
\text{ and }\quad
    \frac{s}{2}< \lim_{n\to\infty}min_{\alpha\in [\alpha_0,\alpha_n]}\frac{
E'_{\alpha^2}}{2(1-\alpha)p_n}.
\end{align*}}

{We first investigate $\frac{s}{2}< \lim_{n\to\infty}min_{\alpha\in [0,\alpha_0]}\frac{
E_{\alpha^2}}{2(1-\alpha)p_n}$.
Recall that 
\[E_{\alpha^2}= \min_{\underline{t}'_X\in
\mathcal{P}}\frac{1}{2}\left((1-\alpha^2)D(\underline{t}'_X||P_X)+\alpha^2 D(\underline{t}''_{X}||P_X)+ D(P_{X,Y}|| (1-\alpha^2) P_XP_{Y''}+\alpha^2 P_{X,Y})\right),\]
where $\mathcal{P}\triangleq \{ \underline{t}_X\in \mathcal{P}_X|\forall x\in \mathcal{X}: \underline{t}_X(x)\in \frac{1}{1-\alpha^2}[P_X(x)-\alpha^2, P_X(x)]\}$, $\mathcal{P}_X$ is the probability simplex on the alphabet $\mathcal{X}$,  $\underline{t}''_X\triangleq \frac{1}{\alpha^2}(P_X-(1-\alpha^2) \underline{t}'_X)$, and $P_{Y''}(\cdot)\triangleq \sum_{x\in \mathcal{X}} \underline{t}'_X(x) P_{Y|X}(\cdot|x)$. Let $\underline{t}'_X(1)= \gamma p_n$, where $\gamma p_n\in \frac{1}{1-\alpha^2}[max(0,p_n-\alpha^2), p_n]$. Then, $\underline{t}''_X(1)= p_n(\frac{1-(1-\alpha^2)\gamma}{\alpha^2})\triangleq \overline{\gamma}p_n$ and $P_{Y''}(1)= \gamma p_n s$. Let $\Gamma=\{\gamma|\gamma p_n\in  \frac{1}{1-\alpha^2}[max(0,p_n-\alpha^2), p_n]\}$. Then,}
{\begin{align*}
    &E_{\alpha^2}= \min_{\gamma\in \Gamma} \frac{1}{2}\Bigg(
    (1-\alpha^2)\left((1-\gamma p_n)\log{\frac{1-\gamma p_n}{1-p_n}}+\gamma p_n\log{\gamma}\right)
    \\&+\alpha^2\left((1-\overline{\gamma} p_n)\log{\frac{1-\overline{\gamma} p_n}{1-p_n}}+\overline{\gamma} p_n\log{\overline{\gamma}}
    \right)
    \\&
    +
    (1-p_n)\log{\frac{1-p_n}{(1-\alpha^2)(1-\gamma p_ns)(1-p_n)+\alpha^2 (1-p_n)}}
    \\&
    +
    p_n(1-s)\log{\frac{p_n(1-s)}{(1-\alpha^2)(1-\gamma p_ns)p_n+\alpha^2 p_n(1-s)}}
    \\&+ 
    p_ns\log{\frac{p_ns}{(1-\alpha^2)\gamma p^2_ns+\alpha^2 p_ns}}
    \Bigg)
    \\&
    = \min_{\gamma\in \Gamma} \frac{1}{2}\Bigg(
    (1-\alpha^2)\left((1-\gamma p_n)\log{\frac{1-\gamma p_n}{1-p_n}}+\gamma p_n\log{\gamma}\right)
    \\&+\alpha^2\left((1-\overline{\gamma} p_n)\log{\frac{1-\overline{\gamma} p_n}{1-p_n}}+\overline{\gamma} p_n\log{\overline{\gamma}}
    \right)
    \\&
    +
    (1-p_n)\log{\frac{1}{(1-\alpha^2)(1-\gamma p_ns)+\alpha^2 }}
    \\&
    +
    p_n(1-s)\log{\frac{1-s}{(1-\alpha^2)(1-\gamma p_ns)+\alpha^2 (1-s)}}
    \\&+ 
    p_ns\log{\frac{1}{(1-\alpha^2)\gamma p_n+\alpha^2 }}
    \Bigg)
\end{align*}}
{
So, 
\begin{align*}
  &  \lim_{n\to\infty}\min_{\alpha\in [0,\alpha_0]}\frac{
E_{\alpha^2}}{2(1-\alpha)p_n}
= 
\\
&\min_{\alpha\in [0,\alpha_0]}\lim_{n\to\infty}\frac{1}{4(1-\alpha)}\min_{\gamma\in \Gamma} \Bigg(
    (1-\alpha^2)\left(\frac{(1-\gamma p_n)}{p_n}\log{\frac{1-\gamma p_n}{1-p_n}}+\gamma \log{\gamma}\right)
    \\&+\alpha^2\left(\frac{(1-\overline{\gamma} p_n)}{p_n}\log{\frac{1-\overline{\gamma} p_n}{1-p_n}}+\overline{\gamma} \log{\overline{\gamma}}
    \right)
    \\&
    +
    \frac{(1-p_n)}{p_n}\log{\frac{1}{(1-\alpha^2)(1-\gamma p_ns)+\alpha^2 }}
    \\&
    +
    (1-s)\log{\frac{1-s}{(1-\alpha^2)(1-\gamma p_ns)+\alpha^2 (1-s)}}
    + 
    s\log{\frac{1}{(1-\alpha^2)\gamma p_n+\alpha^2 }}
    \Bigg)
\end{align*}}
Note that $\frac{(1-\gamma p_n)}{p_n}\log{\frac{1-\gamma p_n}{1-p_n}}= 1-\gamma+o(p_n)$. Similarly, $\frac{(1-\overline{\gamma} p_n)}{p_n}\log{\frac{1-\overline{\gamma} p_n}{1-p_n}}= 1-\overline{\gamma}+o(p_n)$, and $ \frac{(1-p_n)}{p_n}\log{\frac{1}{(1-\alpha^2)(1-\gamma p_ns)+\alpha^2 }}= (1-\alpha^2)\gamma s+o(p_n)$. So,
\begin{align*}
  &  \min_{\alpha\in [0,\alpha_0]}\lim_{n\to\infty}\frac{
E_{\alpha^2}}{2(1-\alpha)p_n}
= 
\\
&\min_{\alpha\in [0,\alpha_0]}\lim_{n\to\infty}\frac{1}{4(1-\alpha)}\min_{\gamma\in \Gamma} \Bigg(
    (1-\alpha^2)\left(1-\gamma+o(p_n)+\gamma \log{\gamma}\right)
    \\&+\alpha^2\left(1-\overline{\gamma}+o(p_n)+\overline{\gamma} \log{\overline{\gamma}}
    \right)
    +
    (1-\alpha^2)\gamma s+o(p_n)
    \\&
    +
    (1-s)\log{\frac{1-s}{(1-\alpha^2)+\alpha^2 (1-s)}}
    + 
    s\log{\frac{1}{\alpha^2 }}
    \Bigg)
    \\&
    =\min_{\alpha\in [0,\alpha_0]}
   \frac{1}{4(1-\alpha)}\min_{\gamma\in \Gamma} \Bigg(
    (1-\alpha^2)\left(1-\gamma+\gamma \log{\gamma}\right)
    +\alpha^2\left(1-\overline{\gamma}+\overline{\gamma} \log{\overline{\gamma}}
    \right)
    +
    (1-\alpha^2)\gamma s
    \\&
    +
    (1-s)\log{\frac{1-s}{(1-\alpha^2)+\alpha^2 (1-s)}}
    + 
    s\log{\frac{1}{\alpha^2 }}
    \Bigg)
     \\&
    \textcolor{black}{
   \geq
    \min_{\alpha\in [0,\alpha_0]}
   \frac{1}{4(1-\alpha)}\min_{\gamma\in [0,\frac{1}{1-\alpha^2}]} \Bigg(
    (1-\alpha^2)\left(1-\gamma+\gamma \log{\gamma}\right)
    +\alpha^2\left(1-\overline{\gamma}+\overline{\gamma} \log{\overline{\gamma}}
    \right)
    +
    (1-\alpha^2)\gamma s}
    \\&
        \textcolor{black}{
    +
    (1-s)\log{\frac{1-s}{(1-\alpha^2)+\alpha^2 (1-s)}}
    + 
    s\log{\frac{1}{\alpha^2 }\Bigg)
    }},
\end{align*}
\textcolor{black}{where the last inequality follows from $\Gamma\subseteq [0,\frac{1}{1-\alpha^2}]$.
As shown in Figure \ref{Fig:2}, the result of this optimization is greater than $\frac{s}{2}$ for all $\alpha_0\leq 0.8$.  This is proved analytically in the following. }
\begin{figure}[!h]
 \begin{center}
\includegraphics[draft=false, width=0.8\textwidth]{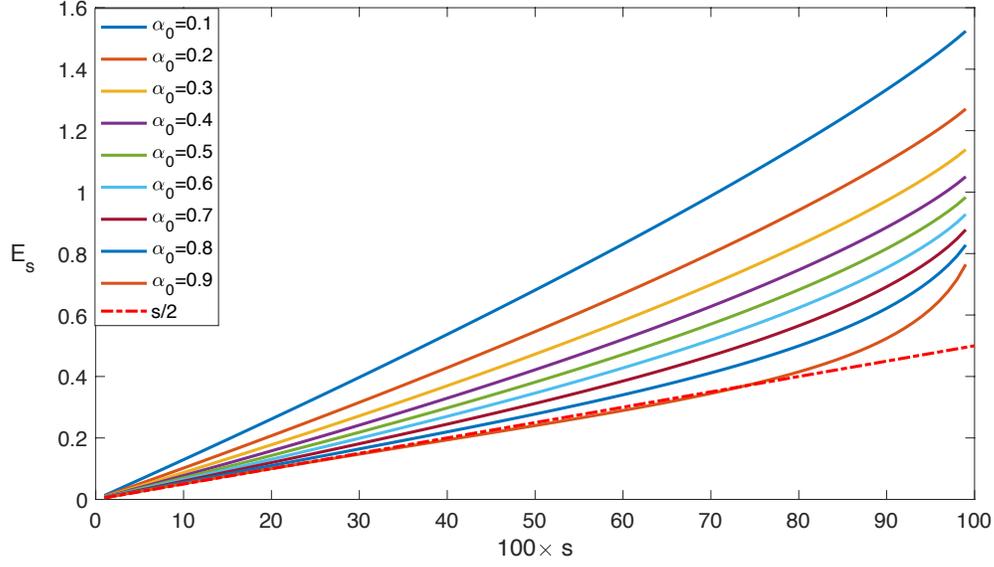}
\caption{The figure shows the value of $E_s=\lim_{n\to \infty}\min_{\alpha\in [0,\alpha_0]}\frac{E_{\alpha}}{2(1-\alpha)p_n}$ for different values of $\alpha_0$. The red-dashed line shows the $\frac{s}{2}$ values. It can be seen that if $\alpha_0\leq 0.8$, then $\frac{s}{2}<\min_{\alpha\in [0,\alpha_0]}E_{\alpha}$.}
\label{Fig:2}
\end{center}
\end{figure}

\textcolor{black}{
Note that $\overline{\gamma}=(\frac{1-(1-\alpha^2)\gamma}{\alpha^2})$, we have: 
\begin{align*}
  &  \min_{\alpha\in [0,\alpha_0]}\lim_{n\to\infty}\frac{
E_{\alpha^2}}{2(1-\alpha)p_n}
     \\&
    \textcolor{black}{
   \geq
    \min_{\alpha\in [0,\alpha_0]}
   \frac{1}{4(1-\alpha)}\min_{\gamma\in [0,\frac{1}{1-\alpha^2}]} \Bigg(
    (1-\alpha^2)\gamma \log{\gamma}
    +\alpha^2\overline{\gamma} \log{\overline{\gamma}}
    +
    (1-\alpha^2)\gamma s}
    \\&
        \textcolor{black}{
    +
    (1-s)\log{\frac{1-s}{(1-\alpha^2)+\alpha^2 (1-s)}}
    + 
    s\log{\frac{1}{\alpha^2 }\Bigg)
    }}
    \end{align*}}
\textcolor{black}{Taking the derivative of the terms which are functions of $\gamma$ and setting equal to 0, we get:
\begin{align*}
    & f(\gamma)\triangleq 
    (1-\alpha^2)\gamma \log{\gamma}
    +\alpha^2\overline{\gamma} \log{\overline{\gamma}}
    +
    (1-\alpha^2)\gamma s
    \\&
    \frac{\partial}{\partial \gamma}f(\gamma)|_{\gamma=\gamma^*}=0
    \\&
    \Rightarrow (1-\alpha^2)\log{\gamma^*}+(1-\alpha^2)-(1-\alpha^2)\log{(\frac{1-(1-\alpha^2)\gamma^*}{\alpha^2})}-(1-\alpha^2)+(1-\alpha^2)s=0
    \\&
    \Rightarrow 
    \log{\gamma^*}-\log{(\frac{1-(1-\alpha^2)\gamma^*}{\alpha^2})}+s=0
    \\&
        \Rightarrow 
    \log{(\frac{\alpha^2\gamma^*}{1-(1-\alpha^2)\gamma^*})}=-s
    \\&
    \Rightarrow 
    \gamma^*= \frac{1}{\alpha^2 e^s+1-\alpha^2}.
\end{align*}
Note that:
\begin{align*}
   \frac{\partial^2}{\partial \gamma^2}f(\gamma)=
   \frac{1-\alpha^2}{\gamma}+\frac{(1-\alpha^2)^2}{{1-(1-\alpha^2)\gamma}}>0.
\end{align*}}
\textcolor{black}{ So, $\gamma^*$ achieves the minimimum of $f(\gamma)$.
Note that $\overline{\gamma}^*= \frac{e^s}{\alpha^2 e^s+1-\alpha^2}$.
 Consequently,
\begin{align*}
  &  \min_{\alpha\in [0,\alpha_0]}\lim_{n\to\infty}\frac{
E_{\alpha^2}}{2(1-\alpha)p_n}
     \\&
  \geq
    \min_{\alpha\in [0,\alpha_0]}
  \frac{1}{4(1-\alpha)}\Bigg(
    \frac{1-\alpha^2}{\alpha^2 e^s+1-\alpha^2} \log{\frac{1}{\alpha^2 e^s+1-\alpha^2}}
    +{\frac{e^s\alpha^2}{\alpha^2 e^s+1-\alpha^2}} \log{{\frac{e^s}{\alpha^2 e^s+1-\alpha^2}}}
    +
    \\&
        \frac{1-\alpha^2}{\alpha^2 e^s+1-\alpha^2} s
    +
    (1-s)\log{\frac{1-s}{(1-\alpha^2)+\alpha^2 (1-s)}}
    + 
    s\log{\frac{1}{\alpha^2 }
    }\Bigg)
    \\&
    =    
    \min_{\alpha\in [0,\alpha_0]}
  \frac{1}{4(1-\alpha)}\Bigg(
     \log{\frac{1}{\alpha^2 e^s+1-\alpha^2}}
    +\frac{se^s\alpha^2}{\alpha^2 e^s+1-\alpha^2}
    +
    \\&
        \frac{(1-\alpha^2)s}{\alpha^2 e^s+1-\alpha^2} 
    +
    (1-s)\log{\frac{1-s}{(1-\alpha^2)+\alpha^2 (1-s)}}
    + 
    s\log{\frac{1}{\alpha^2 }
    }\Bigg)
    \\&
      =    
    \min_{\alpha\in [0,\alpha_0]}
  \frac{1}{4(1-\alpha)}\Bigg(
     \log{\frac{1}{\alpha^2 e^s+1-\alpha^2}}
    +s
    +
    (1-s)\log{\frac{1-s}{(1-\alpha^2)+\alpha^2 (1-s)}}
    + 
    s\log{\frac{1}{\alpha^2 }
    }\Bigg)
    \\&
          =    
    \min_{\alpha\in [0,\alpha_0]}
  \frac{1}{4(1-\alpha)}\Bigg(
     \log{\frac{1}{\alpha^2 e^s+1-\alpha^2}}
    +
    (1-s)\log{\frac{1-s}{(1-\alpha^2)+\alpha^2 (1-s)}}
    + 
    s\log{\frac{e}{\alpha^2 }
    }\Bigg)
    \end{align*}}
%     Note that the last expression inside the minimization is increasing in $\alpha$. It suffices to show the following:
%     \begin{align*}
%      &h(\alpha)\triangleq  
%      \log{\frac{1}{\alpha^2 e^s+1-\alpha^2}}
%     +
%     (1-s)\log{\frac{1-s}{(1-\alpha^2)+\alpha^2 (1-s)}}
%     + 
%     s\log{\frac{e}{\alpha^2 }
%     }
%     \\&
%     \frac{\partial}{\partial \alpha}h(\alpha)=
%     \frac{-2\alpha(e^s-1)}{1+\alpha^2(e^s-1)}
%     +\frac{2\alpha s(1-s)}{1-\alpha^2s}- \frac{2s}{\alpha}<0,
%     \end{align*}
%     where the last inequality follows from $\frac{2\alpha s(1-s)}{1-\alpha^2s}- \frac{2s}{\alpha}<0$. So, 
%     \begin{align*}
%         &  \min_{\alpha\in [0,\alpha_0]}\lim_{n\to\infty}\frac{
% E_{\alpha^2}}{2(1-\alpha)p_n}
% \geq 
% \frac{1}{4(1-\alpha_0)}\Bigg(
%      \log{\frac{1}{\alpha_0^2 e^s+1-\alpha_0^2}}
%     +
%     (1-s)\log{\frac{1-s}{(1-\alpha_0^2)+\alpha_0^2 (1-s)}}
%     + 
%     s\log{\frac{e}{\alpha_0^2 }
%     }\Bigg)
%     \end{align*}
\textcolor{black}{Define the following:
\begin{align*}
    &h(\alpha,s)\triangleq
  \frac{1}{4(1-\alpha)}\Bigg(
     \log{\frac{1}{\alpha^2 e^s+1-\alpha^2}}
    +
    (1-s)\log{\frac{1-s}{(1-\alpha^2)+\alpha^2 (1-s)}}
    + 
    s\log{\frac{e}{\alpha^2 }
    }\Bigg),\\
    &f_s(\alpha)\triangleq h(\alpha,s).
\end{align*}}
\textcolor{black}{Note that for any $s\in [\frac{1}{4},\frac{1}{2}]$, the function $f_s(\alpha)$ is continuous. So, $min_{\alpha\in [0,\alpha_0]} f_s(\alpha), \frac{1}{4}<s<\frac{1}{2}$ is achieved either at $\alpha=\alpha_0$, $\alpha=0$, or at a point where $\frac{\partial f_s(\alpha)}{\partial \alpha} =0$. Also, $f_s(\alpha)\to \infty$ as $\alpha\to 0$ for $ \frac{1}{4}<s<\frac{1}{2}$.
We first show that  $\frac{\partial f_\alpha(s)}{\partial \alpha} \neq 0$ for all $\alpha\in (0,1]$ and $\frac{1}{4}<s<\frac{1}{2}$. Assume otherwise:
\begin{align*}
    &\frac{\partial f_\alpha(s)}{\partial \alpha} = \frac{1}{4(1-\alpha)}\bigg(
    \frac{-2\alpha(e^s-1)}{1+\alpha^2(e^s-1)}+\frac{2(1-s)s\alpha}{1-\alpha^2s}-\frac{2s}{\alpha}\bigg)+
    \\&
    \frac{1}{4(1-\alpha)^2}\bigg( \log{\frac{1}{\alpha^2 e^s+1-\alpha^2}}
    +
    (1-s)\log{\frac{1-s}{(1-\alpha^2)+\alpha^2 (1-s)}}
    + 
    s\log{\frac{e}{\alpha^2 }
    }\Bigg)=0
    \\&
    \iff (1-\alpha)\bigg(
    \frac{-2\alpha(e^s-1)}{1+\alpha^2(e^s-1)}+\frac{2(1-s)s\alpha}{1-\alpha^2s}-\frac{2s}{\alpha}\bigg)+
    \\&\bigg( \log{\frac{1}{\alpha^2 e^s+1-\alpha^2}}
    +
    (1-s)\log{\frac{1-s}{(1-\alpha^2)+\alpha^2 (1-s)}}
    + 
    s\log{\frac{e}{\alpha^2 }
    }\Bigg)=0,
    \end{align*}
    We argue that the last equality cannot hold:
    \begin{align*}
        &\kappa_s(\alpha)\triangleq (1-\alpha)\bigg(
    \frac{-2\alpha(e^s-1)}{1+\alpha^2(e^s-1)}+\frac{2(1-s)s\alpha}{1-\alpha^2s}-\frac{2s}{\alpha}\bigg)+
    \\&\bigg( \log{\frac{1}{\alpha^2 e^s+1-\alpha^2}}
    +
    (1-s)\log{\frac{1-s}{(1-\alpha^2)+\alpha^2 (1-s)}}
    + 
    s\log{\frac{e}{\alpha^2 }
    }\Bigg)
    \end{align*}
    Note that $\lim_{\alpha\to 0} \kappa_s(\alpha)=-\infty$. Also, $\kappa_s(\alpha)$ is an increasing function in $\alpha$ as shown below:
    \begin{align*}
      &  \frac{\partial \kappa_s(\alpha)}{\partial \alpha}= (1-\alpha) \frac{\partial}{\partial \alpha}\bigg(
    \frac{-2\alpha(e^s-1)}{1+\alpha^2(e^s-1)}+\frac{2(1-s)s\alpha}{1-\alpha^2s}-\frac{2s}{\alpha}\bigg)\geq 0
    \\&
    \iff\frac{\partial}{\partial \alpha}\bigg(
    \frac{-2\alpha(e^s-1)}{1+\alpha^2(e^s-1)}+\frac{2(1-s)s\alpha}{1-\alpha^2s}-\frac{2s}{\alpha}\bigg)\geq 0
    \\&
    \iff 
    \frac{-2(e^s-1)(1+\alpha^2(e^s-1))+4\alpha^2(e^s-1)}{(1+\alpha^2(e^s-1))^2}+ 
    \frac{2(1-s)s(1-\alpha^2s)+4\alpha^2s^2(1-s)}{(1-\alpha^2s)^2}+\frac{2s}{\alpha^2}\geq 0
    \\&
    \iff 
  \frac{2(e^s-1)(-1+\alpha^2(e^s-1))}{(1+\alpha^2(e^s-1))^2}
  +
  \frac{2(1-s)s(1+\alpha^2s)}{(1-\alpha^2s)^2}
  +\frac{2s}{\alpha^2}\geq 0
    \end{align*}
      Note that the above holds if $\frac{2(e^s-1)(-1+\alpha^2(e^s-1))}{(1+\alpha^2(e^s-1))^2}\geq 0$. So, assume that
      $\frac{2(e^s-1)(-1+\alpha^2(e^s-1))}{(1+\alpha^2(e^s-1))^2}<0$. Then,
  \begin{align*}
&
 \frac{\partial \kappa_s(\alpha)}{\partial \alpha}\geq 0 
 \Leftarrow
 {2(e^s-1)(-1+\alpha^2(e^s-1))}
  +
  \frac{2(1-s)s(1+\alpha^2s)}{(1-\alpha^2s)^2}
  +\frac{2s}{\alpha^2}\geq 0,
  \\&
   \Leftarrow
   {2(e^s-1)(-1+\alpha^2(e^s-1))}
  +
  2s(1-s)
  +\frac{2s}{\alpha^2}\geq 0.
    \end{align*}
Note that $e^s-1>0$ for $\frac{1}{4}<s<\frac{1}{2}$. So, 
\begin{align*}
  & \frac{\partial \kappa_s(\alpha)}{\partial \alpha}\geq 0 
 \Leftarrow  
 -2(e^s-1)
  +
  2s(1-s)
  +2s\geq 0,
\end{align*}}
\textcolor{black}{The last inequality holds for $\frac{1}{4}<s<\frac{1}{2}$. To see this, let $t(s) \triangleq -2(e^s-1)
  +
  2s(1-s)
  +2s$. Note that $t(\frac{1}{4}), t(\frac{1}{2})>0$ and $\frac{\partial t(s)}{\partial s}=0 \iff -2e^s=4s-4$ which holds for $s= 0.3149...$ in which case $t(s)>0$. So, i) $t(s)$ is continuous, ii) $t(s)>0$ at $s=\frac{1}{4}, s=\frac{1}{2}$, and iii) $t(s)>0$ at all of the points where its derivative with respect to $s$ is 0, so we must have $t(s)>0, \frac{1}{4}<s<\frac{1}{2}$. Consequently, $\kappa_s(\alpha)$ is increasing in $\alpha$ and $\kappa_s(\alpha)<\kappa_s(1)= \log\frac{1}{1+e^s}+s<-\log{e^s}+s=0$. So, $\frac{\partial f_s{\alpha}}{\partial \alpha }\neq 0, \alpha\in [0,\alpha_0]$. So, $\min_{[0,\alpha_0]}f_s(\alpha)$ is achieved at $\alpha=\alpha_0$:
  \begin{align*}
       &  \min_{\alpha\in [0,\alpha_0]}\lim_{n\to\infty}\frac{
E_{\alpha^2}}{2(1-\alpha)p_n}
     \\&
  \geq
  \frac{1}{4(1-\alpha_0)}\Bigg(
     \log{\frac{1}{\alpha_0^2 e^s+1-\alpha_0^2}}
    +
    (1-s)\log{\frac{1-s}{(1-\alpha_0^2)+\alpha_0^2 (1-s)}}
    + 
    s\log{\frac{e}{\alpha_0^2 }
    }\Bigg)
  \end{align*}
  }
\textcolor{black}{Next we show that $h(\alpha,s)$ is increasing in $s$ for any fixed $\alpha$: 
\begin{align*}
    &g(s)\triangleq  \log{\frac{1}{\alpha^2 e^s+1-\alpha^2}}
    +
    (1-s)\log{\frac{1-s}{(1-\alpha^2)+\alpha^2 (1-s)}}
    + 
    s\log{\frac{e}{\alpha^2 }
    }\\
    &
    \Rightarrow
    \frac{\partial}{\partial s}g(s)= 
    \frac{-\alpha^2e^s}{\alpha^2 e^s+1-\alpha^2}-
    \log{\frac{1-s}{(1-\alpha^2)+\alpha^2 (1-s)}}
    -\frac{1-\alpha^2}{(1-\alpha^2)+\alpha^2 (1-s)}+\log{\frac{e}{\alpha^2 }}
    \\&
    = 
    \frac{-\alpha^2e^s}{\alpha^2 e^s+1-\alpha^2}-
    \log{\frac{1-s}{1-\alpha^2s}}
    -\frac{1-\alpha^2}{1-\alpha^2s}+\log{\frac{e}{\alpha^2 }}
    \\&
    =
    -1+\frac{1-\alpha^2}{\alpha^2 e^s+1-\alpha^2}-1+\frac{\alpha^2(1-s)}{1-\alpha^2s}-
    \log{\frac{1-s}{1-\alpha^2s}}
    +\log{\frac{e}{\alpha^2 }},
\end{align*}
}
\textcolor{black}{
Note that $\frac{\alpha^2(1-s)}{1-\alpha^2s}-
    \log{\frac{1-s}{1-\alpha^2s}}>s>\frac{1}{4}$ since $\log{x}>1-\frac{1}{x}$ and $s<\frac{1}{2}$. Also, $\frac{1-\alpha^2}{\alpha^2 e^s+1-\alpha^2}>\frac{1-0.64}{1-0.64(1-e^{\frac{1}{2}})}$ and $\log{\frac{e}{\alpha^2}}>\log{\frac{e}{0.64}}$ since $\alpha<0.8$ and $s<\frac{1}{2}$. So,
    \begin{align*}
        \frac{\partial}{\partial s}g(s)>
        -2+ \frac{1-0.64}{1-0.64(1-e^{\frac{1}{2}})}+\frac{3}{10}+\log{\frac{e}{0.64}}>0.
    \end{align*}
}
\textcolor{black}{Let $s_i= \frac{1}{4}(1+\frac{i}{1000}), i\in [1,1001]$. Then, it can be verified numerically that $h(\alpha_0, s_i)>\frac{s_{i+1}}{2}, i\in [1,1000]$. Since $h(\alpha_0,s)$ is increasing in $s$ as shown above, we conclude that for any $s\in [\frac{1}{4},\frac{1}{2}]$, we have $h(\alpha_0,s)>h(\alpha_0, s_i)> \frac{s_{i+1}}{2}>\frac{s}{2}$ for $i$ such that $s_i<s<s_{i+1}$. } 

{Next, we investigate $  \frac{s}{2}< \lim_{n\to\infty}min_{\alpha\in [\alpha_0,\alpha_n]}\frac{
E'_{\alpha^2}}{2(1-\alpha)p_n}$. Note that $\alpha_n$ is a sequence which converges to one and can be chosen arbitrarily. We consider $\alpha_n$ such that $1-\alpha^2_n\to 0$ slower than $p_n\to 0$. More precisely, we assume that $\frac{p_n}{1-\alpha^2_n}\to 0$ as $n\to \infty$.}

{Recall that \begin{align*}
   E'_{\alpha^2}=\min_{\underline{t}'_{X,Y}\in \mathcal{P}'} \left(\frac{1-\alpha^2}{3}\right)
   D(\underline{t}'_{X,Y}||P_{X}P_{Y})+
   \alpha^2 D(\underline{t}''_{X,Y}||P_{X,Y}),
\end{align*}}
{where $\mathcal{P}'\triangleq \{ \underline{t}_{X,Y}\in \mathcal{P}_{X,Y}|\forall (x,y)\in \mathcal{X}\times \mathcal{Y}: \underline{t}_{X,Y}(x,y)\in \frac{1}{1-\alpha^2}[P_{X,Y}(x,y)-\alpha^2, P_{X,Y}(x,y)]$, $\mathcal{P}_{X,Y}$ is the probability simplex on the alphabet $\mathcal{X}\times \mathcal{Y}$, and $\underline{t}''_{X,Y}\triangleq \frac{1}{\alpha^2}(P_{X,Y}-(1-\alpha^2) \underline{t}'_{X,Y})$. Note that $\underline{t}'_{X,Y}(0,1)=0$ for all $\underline{t}'_{X,Y}\in \mathcal{P}'$ since $P_{X,Y}(0,1)=0$. So, minimizing over $\underline{t}'_{X,Y}$ requires minimizing over two free parameters $\gamma_1= \underline{t}'_{X,Y}(0,0)$ and $\gamma_2=\underline{t}'_{X,Y}(1,0)$, satisfying the following constraints:}
{\begin{align*}
    &    \max\{0,\frac{P_{X,Y}(0,0)-\alpha^2}{1-\alpha^2}\}
    \leq \gamma_1<\min\{1,\frac{P_{X,Y}(0,0)}{1-\alpha^2}\}\\
    &\max\{0,\frac{P_{X,Y}(1,0)-\alpha^2}{1-\alpha^2}\}
    \leq \gamma_2<\min\{1,\frac{P_{X,Y}(1,0)}{1-\alpha^2}\},\\
    &
   \max\{0,\frac{P_{X,Y}(1,1)-\alpha^2}{1-\alpha^2}\} \leq 1- \gamma_1-\gamma_2\leq 
   \min\{1,\frac{P_{X,Y}(1,1)}{1-\alpha^2}\}
\end{align*}}
{where $P_{X,Y}(0,0)=1-p_n$, $P_{X,Y}(1,0)= (1-s)p_n$ and $P_{X,Y}(1,1)=sp_n$.
Using the assumption that $\frac{p_n}{1-\alpha^2_n}\to 0$ as $n\to \infty$ and noting that $\alpha\leq \alpha_n$, we have $\max\{0 ,\frac{P_{X,Y}(0,0)-\alpha^2}{1-\alpha^2}\}= \max\{0, \frac{1-p_n-\alpha^2}{1-\alpha^2}\}=1-\frac{p_n}{1-\alpha^2}$, $\min\{1,\frac{P_{X,Y}(0,0)}{1-\alpha^2}\}=\min\{1,\frac{1-p_n}{1-\alpha^2}\}=1$, $\max\{0,\frac{P_{X,Y}(1,0)-\alpha^2}{1-\alpha^2}\}= \max\{0, \frac{p_n(1-s)-\alpha^2}{1-\alpha^2}\}=0$, and $\min\{1,\frac{P_{X,Y}(1,0)}{1-\alpha^2}\}= \min\{1,\frac{p_n(1-s)}{1-\alpha^2}\}=\frac{p_n(1-s)}{1-\alpha^2}$, $\max\{0,\frac{P_{X,Y}(1,1)-\alpha^2}{1-\alpha^2}\}= \max\{0,\frac{sp_n-\alpha^2}{1-\alpha^2}\}=0$, and $\min\{1,\frac{P_{X,Y}(1,1)}{1-\alpha^2}\}$ $= \min\{1,\frac{sp_n}{1-\alpha^2}\}= \frac{sp_n}{1-\alpha^2}$. To summarize:}
{\begin{align}
    &   \label{eq:con:1} 1-\frac{p_n}{1-\alpha^2}
    \leq \gamma_1<1\\
    &\label{eq:con:2}
    0
    \leq \gamma_2<\frac{p_n(1-s)}{1-\alpha^2},\\
    &\label{eq:con:3}
     1 - \frac{sp_n}{1-\alpha^2}\leq \gamma_1+\gamma_2\leq 1 
  .
\end{align}
Note that $\left(\frac{1-\alpha^2}{3}\right)
   D(\underline{t}'_{X,Y}||P_{X}P_{Y})+
   \alpha^2 D(\underline{t}''_{X,Y}||P_{X,Y})$ is convex in $\underline{t}'_{X,Y}$ as shown in the following. Let $\underline{t}'_{X,Y}= \lambda \underline{t}^1+(1-\lambda)\underline{t}^2, \lambda\in [0,1]$. Then, $\underline{t}''_{X,Y}= \frac{1}{\alpha^2}(P_{XY}-(1-\alpha^2)( \lambda \underline{t}^1+(1-\lambda)\underline{t}^2))$. So,
   \begin{align*}
     &\left(\frac{1-\alpha^2}{3}\right)
   D(\underline{t}'_{X,Y}||P_{X}P_{Y})+
   \alpha^2 D(\underline{t}''_{X,Y}||P_{X,Y})  \\&
   =\left(\frac{1-\alpha^2}{3}\right)
   D(\lambda \underline{t}^1+(1-\lambda)\underline{t}^2||P_{X}P_{Y})+
   \alpha^2 D(\frac{1}{\alpha^2}(P_{XY}-(1-\alpha^2)( \lambda \underline{t}^1+(1-\lambda)\underline{t}^2))||P_{X,Y}) 
   \\& \leq 
   \left(\frac{1-\alpha^2}{3}\right)
   (\lambda D(\underline{t}^1||P_XP_Y) +(1-\lambda)D(\underline{t}^2||P_{X}P_{Y}))
   \\&+
   \alpha^2 \left(\lambda D\left(\frac{1}{\alpha^2}(P_{XY}-(1-\alpha^2)  \underline{t}^1)||P_{XY}\right)+(1-\lambda)D\left(\frac{1}{\alpha^2}(P_{XY}-(1-\alpha^2)\underline{t}^2)||P_{XY}\right)\right),
   \end{align*}
   where we have used the convexity of  KL divergence.
   So, the minimum over $\underline{t}'_{X,Y}$ is achieved when the gradient is zero or on the endpoints of the optimization region. Next, we show that the gradient is not zero for any $\underline{t}'_{X,Y}$ for which $\gamma_1$ and $\gamma_2$ satisfy  \eqref{eq:con:1}-\eqref{eq:con:3}}. {Let us assume otherwise. That is, let $Q^*_{X,Y}$ be such that the gradient is $0$ when $\underline{t}'_{X,Y} =Q^*_{X,Y}$. Let $Q_{t,X,Y}(x,y)\triangleq Q^*_{X,Y}(x,y)+t\epsilon_{x,y}, t\in [0,1]$ for $\epsilon_{x,y}\in [-\frac{Q^*_{X,Y}(x,y)}{2},\frac{Q^*_{X,Y}(x,y)}{2}]$ such that $\sum_{x,y}\epsilon_{x,y}=0$. Note that $Q_{t,X,Y}(x,y)$ is a valid probability distribution on $\mathcal{X}\times \mathcal{Y}$. Let \[f(t)\triangleq \left(\frac{1-\alpha^2}{3}\right)
   D(Q_{t,X,Y}(x,y)||P_{X}P_{Y})+
   \alpha^2 D\left(\frac{1}{\alpha^2}\left(P_{XY}- (1-\alpha^2)Q_{t,X,Y}(x,y)\right)||P_{X,Y}\right).\]
   By construction, we must have $\frac{\partial}{\partial t}f(t)|_{t=0}=0$ for any choice of $\epsilon_{x,y}$. So,}
{   \begin{align*}
     &\left(\frac{1-\alpha^2}{3}\right)\left(
     \sum_{x,y}\epsilon_{x,y}\log{\frac{Q^*_{X,Y}(x,y)}{P_{X}(x)P_Y(y)}}+\sum_{x,y}\epsilon_{x,y}
     \right)+
    \\& \alpha^2\left(\sum_{x,y}
    \frac{-(1-\alpha^2)}{\alpha^2}\epsilon_{x,y}\log{\frac{P_{X,Y}(x,y)-(1-\alpha^2)Q^*_{X,Y}(x,y)}{\alpha^2 P_{X,Y}(x,y)}}+\frac{-(1-\alpha^2)}{\alpha^2}\sum_{x,y}\epsilon_{x,y}
    \right)=0
    \\&\Rightarrow
  ( 1-\alpha^2)\left(\frac{1}{3}
     \sum_{x,y}\epsilon_{x,y}\log{\frac{Q^*_{X,Y}(x,y)}{P_{X}(x)P_Y(y)}}
    -\sum_{x,y}\epsilon_{x,y}\log{\frac{P_{X,Y}(x,y)-(1-\alpha^2)Q^*_{X,Y}(x,y)}{\alpha^2 P_{X,Y}(x,y)}}\right)=0
    \\&\Rightarrow
    \sum_{x,y}\epsilon_{x,y}
    \log{\frac{\alpha^2P_{X,Y}(x,y){Q^*}^{\frac{1}{3}}_{X,Y}(x,y)}{(P_{X,Y}(x,y)-(1-\alpha^2)Q^*_{X,Y}(x,y))P_X^{\frac{1}{3}}(x)P_Y^{\frac{1}{3}}(y)}}=0.
   \end{align*}}
   {The above should hold for all choices of $\epsilon_{x,y}$. Consequently, 
   \begin{align*}
       \alpha^2P_{X,Y}(x,y){Q^*}^{\frac{1}{3}}_{X,Y}(x,y)
       =(P_{X,Y}(x,y)-(1-\alpha^2)Q^*_{X,Y}(x,y))P_X^{\frac{1}{3}}(x)P_Y^{\frac{1}{3}}(y), \forall x,y\in \mathcal{X}\times \mathcal{Y}
   \end{align*}}
 { \noindent From Equations \eqref{eq:con:1}-\eqref{eq:con:3} we can take $Q^*_{X,Y}(0,0)= 1-\tau_1p_n, \tau_1\in [0, \frac{1}{1-\alpha^2}]$, $Q^*_{X,Y}(1,0)= \tau_2p_n, $ $\tau_2\in [0,\frac{1-s}{1-\alpha^2}]$, and $Q^*_{X,Y}(1,1)= \tau_3p_n, \tau_3\in [0, \frac{s}{1-\alpha^2}]$ such that $\tau_1=\tau_2+\tau_3$. So,
  \begin{align*}
  &(x,y)=(0,0):   \alpha^2(1-p_n)(1-\tau_1p_n)^{\frac{1}{3}}= (1-p_n-(1-\alpha^2)(1-\tau_1p_n))(1-p_n)^{\frac{1}{3}}(1-sp_n)^{\frac{1}{3}}
      \\&
    (x,y)=(1,0):   \alpha^2(1-s)p_n(\tau_2p_n)^{\frac{1}{3}}= ((1-s)p_n-(1-\alpha^2)\tau_2p_n)p_n^{\frac{1}{3}}(1-sp_n)^{\frac{1}{3}}
      \\&
    (x,y)=(1,1):   \alpha^2sp_n(\tau_3p_n)^{\frac{1}{3}}= (sp_n-(1-\alpha^2)\tau_3p_n)p_n^{\frac{1}{3}}(sp_n)^{\frac{1}{3}}
  \end{align*}
  So,
    \begin{align*}
  &\alpha^2(1-p_n)(1-\tau_1p_n)^{\frac{1}{3}}= (1-p_n-(1-\alpha^2)(1-\tau_1p_n))(1-p_n)^{\frac{1}{3}}(1-sp_n)^{\frac{1}{3}}
      \\&
  \alpha^2(1-s)p_n\tau_2^{\frac{1}{3}}= ((1-s)p_n-(1-\alpha^2)\tau_2p_n)(1-sp_n)^{\frac{1}{3}}
      \\&
  \alpha^2sp_n\tau_3^{\frac{1}{3}}= (sp_n-(1-\alpha^2)\tau_3p_n)(sp_n)^{\frac{1}{3}}
  \end{align*}}
 { Note that $(1-\tau_1p_n)^{\frac{1}{3}}= 1-\frac{1}{3}\tau_1 p_n+ o(p_n)$, $(1-sp_n)^{\frac{1}{3}}= 1-\frac{1}{3}s p_n+ o(p_n)$, $(1-p_n)^{\frac{1}{3}}=1-\frac{1}{3}p_n+o(p_n)$. Consequently,
  \begin{align*}
& \alpha^2(1-p_n)(1-\frac{1}{3}\tau_1p_n)= (1-p_n-(1-\alpha^2)(1-\tau_1p_n))(1-\frac{1}{3}p_n)(1-\frac{1}{3}sp_n)+o(p_n)
      \\&
 \alpha^2(1-s)p_n\tau_2^{\frac{1}{3}}= ((1-s)p_n-(1-\alpha^2)\tau_2p_n)(1-\frac{1}{3}sp_n)+o(p_n)
      \\&
    \alpha^2sp_n\tau_3^{\frac{1}{3}}= (sp_n-(1-\alpha^2)\tau_3p_n)(sp_n)^{\frac{1}{3}}
  \end{align*}
  So,
  \begin{align*}
      & \alpha^2\frac{(1-p_n)(1-\frac{1}{3}\tau_1p_n)}{(1-\frac{1}{3}p_n)(1-\frac{1}{3}sp_n)}= 1-p_n-(1-\alpha^2)(1-\tau_1p_n)+o(p_n)
      \\&
 \alpha^2\frac{(1-s)p_n\tau_2^{\frac{1}{3}}}{(1-\frac{1}{3}sp_n)}= (1-s)p_n-(1-\alpha^2)\tau_2p_n+o(p_n)
      \\&
    \alpha^2{s^{\frac{2}{3}}p_n^{\frac{2}{3}}\tau_3^{\frac{1}{3}}}= sp_n-(1-\alpha^2)\tau_3p_n
  \end{align*}}
{  Note that if $\tau_3\neq 0$, then from the last inequality, we get $p_n^{-\frac{1}{3}}=
    \frac{s-(1-\alpha^2)\tau_3}{\alpha^2s^{\frac{2}{3}}\tau_3^{\frac{1}{3}}}$ which is a contradiction since   the left hand side goes to infinity as $n\to \infty$. Since we must have $\tau_1=\tau_2+\tau_3$ and $\tau_3=0$, we get $\tau_1=\tau_2$. Also, note that $\frac{(1-p_n)(1-\frac{1}{3}\tau_1p_n)}{(1-\frac{1}{3}p_n)(1-\frac{1}{3}sp_n)}= 1-p_n-\frac{1}{3}\tau_1p_n+\frac{1}{3}p_n+\frac{1}{3}sp_n+o(p_n)= 1-\frac{1}{3}(2+\tau_1-s)p_n+o(p_n)$. Similarly, $\frac{(1-s)p_n\tau_1^{\frac{1}{3}}}{(1-\frac{1}{3}sp_n)}= (1-s)p_n\tau_1^{\frac{1}{3}}(1+\frac{1}{3}sp_n)+o(p_n)=(1-s)p_n\tau_1^{\frac{1}{3}}+o(p_n)$. So, }
{      \begin{align*}
      & \alpha^2(1-\frac{1}{3}(2+\tau_1-s)p_n)= 1-p_n-(1-\alpha^2)(1-\tau_1p_n)+o(p_n)
      \\&
 \alpha^2((1-s)p_n\tau_1^{\frac{1}{3}})= (1-s)p_n-(1-\alpha^2)\tau_1p_n+o(p_n)
  \end{align*}
We have:
 \begin{align}
 &\nonumber \alpha^2(1-\frac{1}{3}(2+\tau_1-s)p_n)= 1-p_n-(1-\alpha^2)(1-\tau_1p_n)+o(p_n)
 \Rightarrow
      \\&\nonumber
      -\frac{\alpha^2}{3}(2+\tau_1-s)p_n= -p_n+(1-\alpha^2)\tau_1p_n+o(p_n)
      \\&\nonumber
      \Rightarrow
     -\frac{\alpha^2}{3}(2+\tau_1-s)= -1+(1-\alpha^2)\tau_1+o(1)
     \\&\label{eq:tau1}
     \Rightarrow 
     \tau_1= \frac{1-\frac{\alpha^{2}}{3}(2-s)}{1-\frac{2\alpha^2}{3}}+o(1),
     \end{align}}
{     and
     \begin{align}
 &\nonumber \alpha^2((1-s)\tau_1^{\frac{1}{3}})= (1-s)-(1-\alpha^2)\tau_1+o(1)
  \\& 
  \nonumber\Rightarrow(\alpha^2\tau_1^{\frac{1}{3}}-1)(1-s)= -(1-\alpha^2)\tau_1+o(1)
  \\&\label{eq:tau2}
  \Rightarrow 
  (1-\alpha^2\tau_1^{\frac{1}{3}})(1-s)= (1-\alpha^2)\tau_1+o(1)
 \end{align}
 Note that for $0<s$, we have $2-s<2$. So, $1-\frac{\alpha^{2}}{3}(2-s)>1-\frac{2\alpha^2}{3}$, so from Equation \eqref{eq:tau1} we have $\tau_1>1$. So, $(1-\alpha^2\tau_1^{\frac{1}{3}})<(\alpha^2-1)$ and $(1-s)<1<\tau_1$ which contradicts Equation \eqref{eq:tau2}.}
  
   {So, since the gradient is always non-zero, the minimum should be achieved at the corner points of the optimization region characterized by Equations \eqref{eq:con:1}-\eqref{eq:con:3}. That is, the minimum is achieved at one of the following four points:}
   {
   \begin{align*}
       &\gamma^{(1)}_1=1, \gamma^{(1)}_2=0\\
       &\gamma^{(2)}_1=1-\frac{p_n(1-s)}{1-\alpha^2}, \gamma^{(2)}_2=\frac{p_n(1-s)}{1-\alpha^2}\\
        &\gamma^{(3)}_1=1-\frac{p_n}{1-\alpha^2}, \gamma^{(3)}_2=\frac{p_n(1-s)}{1-\alpha^2}\\
       &\gamma^{(4)}_1=1-\frac{sp_n}{1-\alpha^2}, \gamma^{(4)}_2=0\\
   \end{align*}
   We investigate each point separately:}
 {  \begin{align*}
       (\gamma^{(1)}_1,\gamma^{(1)}_2):&
       E'_{\alpha^2}=\left(\frac{1-\alpha^2}{3}\right)
   D(\underline{t}'_{X,Y}||P_{X}P_{Y})+
   \alpha D(\underline{t}''_{X,Y}||P_{X,Y})
   \geq
   \left(\frac{1-\alpha^2}{3}\right)
   D(\underline{t}'_{X,Y}||P_{X}P_{Y})
   \\&= \left(\frac{1-\alpha^2}{3}\right)
   \log\frac{1}{P_X(0)P_Y(0)}
   =\left(\frac{1-\alpha^2}{3}\right)\log{\frac{1}{(1-p_n)(1-sp_n)}} 
\\&   \Rightarrow \lim_{n\to \infty}\min_{\alpha\in [\alpha_0,\alpha_n]} \frac{E_{\alpha^2}}{2(1-\alpha)p_n}\geq \lim_{n\to \infty}\min_{\alpha\in [\alpha_0,\alpha_n]}\frac{1+\alpha}{6}(1+s)+ o(1)
\\&\geq \frac{(1+\alpha_0)(1+s)}{6}>\frac{s}{2}\text{ for } s<\frac{1}{2}, \alpha_0=0.8.
   \end{align*}}
 {  \begin{align*}
       (\gamma^{(2)}_1,\gamma^{(2)}_2):& \underline{t}'_{X,Y}(0,0)= \gamma^{(2)}_1 \Rightarrow \underline{t}''_{X,Y}(0,0)=\frac{1}{\alpha^2}(P_{X,Y}(0,0)-(1-\alpha^2)\gamma^{(2)}_1)
       \\& = \frac{1}{\alpha^2}(1-p_n-(1-\alpha^2)+p_n(1-s))= \frac{1}{\alpha^2}(\alpha^2-p_ns)= 1-\frac{p_ns}{\alpha^2},
       \\&
       \underline{t}''_{X,Y}(1,0)=\frac{1}{\alpha^2}(P_{X,Y}(1,0)-(1-\alpha^2)\gamma^{(2)}_2)
       \\& = \frac{1}{\alpha^2}((1-s)p_n-p_n(1-s))= 0,
       \\& \Rightarrow  \underline{t}''_{X,Y}(1,1)=1- (1-\frac{p_ns}{\alpha^2})= \frac{p_ns}{\alpha^2}.
       \\&E'_{\alpha^2}\geq \alpha^2D(\underline{t}''_{X,Y}||P_{X,Y})
       \\&= \alpha^2\left(\left(1-\frac{p_ns}{\alpha^2}\right)\log{\frac{1-\frac{p_ns}{\alpha^2}}{1-p_n}}+ \frac{p_ns}{\alpha^2}\log{\frac{p_ns}{\alpha^2 p_ns}}\right) 
       \\&\Rightarrow
       \lim_{n\to \infty}\min_{\alpha\in [\alpha_0,\alpha_n]} \frac{E_{\alpha^2}}{2(1-\alpha)p_n}\\&\geq
       \min_{\alpha\in [\alpha_0,\alpha_n]}\frac{\alpha^2-s}{2(1-\alpha)}+ \frac{s}{2(1-\alpha)}\log{\frac{1}{\alpha^2}}+o(1)\geq \frac{s}{2} \text{ for } \alpha_0=0.8, s<\frac{1}{2}.
   \end{align*}}
{\begin{align*}
     (\gamma^{(3)}_1,\gamma^{(3)}_2):& \underline{t}'_{X,Y}(0,0)= \gamma^{(3)}_1 \Rightarrow \underline{t}''_{X,Y}(0,0)=\frac{1}{\alpha^2}(P_{X,Y}(0,0)-(1-\alpha^2)\gamma^{(3)}_1)=1
     \\& 
     E'_{\alpha^2}\geq \alpha^2D(\underline{t}''_{X,Y}||P_{X,Y})= \alpha^2\log{\frac{1}{1-p_n}}
     \\&\Rightarrow 
      \lim_{n\to \infty}\min_{\alpha\in [\alpha_0,\alpha_n]} \frac{E_{\alpha^2}}{2(1-\alpha)p_n}\\&\geq \lim_{n\to \infty}\min_{\alpha\in [\alpha_0,\alpha_n]}\frac{\alpha^2}{2(1-\alpha)}+o(1)>\frac{s}{2} \text{ for } \alpha_0>0.8, s<\frac{1}{2}.
     \end{align*}}
     {
       \begin{align*}
       (\gamma^{(4)}_1,\gamma^{(4)}_2):& \underline{t}'_{X,Y}(0,0)= \gamma^{(4)}_1 \Rightarrow \underline{t}''_{X,Y}(0,0)=\frac{1}{\alpha^2}(P_{X,Y}(0,0)-(1-\alpha^2)\gamma^{(4)}_1)
       \\& = \frac{1}{\alpha^2}(1-p_n-(1-\alpha^2)+p_ns)= \frac{1}{\alpha^2}(\alpha^2-p_n(1-s))= 1-\frac{p_n(1-s)}{\alpha^2},
       \\&
       \underline{t}''_{X,Y}(1,0)=\frac{1}{\alpha^2}(P_{X,Y}(1,0)-(1-\alpha^2)\gamma^{(4)}_2)
       \\& = \frac{p_n(1-s)}{\alpha^2},
       \\&\underline{t}''_{X,Y}(1,0)=0,
       \\&\Rightarrow 
       E'_{\alpha^2}\geq \alpha^2D(\underline{t}''_{X,Y}||P_{X,Y})
       = \alpha^2\left((1-\frac{p_n(1-s)}{\alpha^2})\log{\frac{1-\frac{p_n(1-s)}{\alpha^2}}{1-p_n}}+ \frac{p_n(1-s)}{\alpha^2}\log{\frac{1}{\alpha^2}}\right)
       \\&
       \Rightarrow 
        \lim_{n\to \infty}\min_{\alpha\in [\alpha_0,\alpha_n]} \frac{E_{\alpha^2}}{2(1-\alpha)p_n}\\&\geq
        \frac{\alpha^2-(1-s)}{2(1-\alpha)}+
        \frac{1-s}{2(1-\alpha)}\log{\frac{1}{\alpha^2}}>\frac{s}{2}, \text{ for } 0<s<\frac{1}{2}.
\end{align*}
So, if $0<s<\frac{1}{2}$ and $\alpha_0=0.8$, then
\begin{align*}
    \frac{s}{2}< \lim_{n\to\infty}min_{\alpha\in [0,\alpha_0]}\frac{
E_{\alpha^2}}{2(1-\alpha)p_n},
\text{ and }\quad
    \frac{s}{2}< \lim_{n\to\infty}min_{\alpha\in [\alpha_0,\alpha_n]}\frac{
E'_{\alpha^2}}{2(1-\alpha)p_n}.
\end{align*}
This completes the proof.} 

\section{Proof of Theorem \ref{th:3}}
\label{app:th3}

 \begin{table}[]
  \centering
  \setlength\extrarowheight{5pt}
\begin{tabular}{|ll|ll|}
\hline
 $U_{\sigma,\mathcal{C}_i,\mathcal{C}_j}$: & UT/adj. matrix between $\mathcal{C}_i$ and $\mathcal{C}_j$ in $g^1$ & $\widehat{\Sigma}_{\mathcal{C},\mathcal{C}'}$:  & Possible Output Labelings of TM \\ 
 \hline
  $U'_{\sigma,\mathcal{C}_i,\mathcal{C}_j}$: & UT/adj. matrix between $\mathcal{C}_i$ and $\mathcal{C}_j$ in $g^2$   &$\mathcal{E}$: & Lebelings with incorrect labels $\geq n\lambda_n$   \\ \hline
 $\sigma^1,\sigma^2,\sigma^;$: & labelings &$n_i$: &  size of $\mathcal{C}_i$\\ \hline
\end{tabular}
\vspace{0.1in}
\caption{Notation Table: Appendix \ref{app:th3}}
\label{tab:J}
\vspace{-0.2in}
\end{table}

Let $\epsilon_n= O(\frac{\log{n}}{n})$ be a sequence of positive numbers. Fix $n\in \mathbb{N}$ and let $\epsilon=\epsilon_n$. For a given labeling $\sigma''$, define the event $\mathcal{B}_{\sigma''}$ as the event that the sub-matrices corresponding to each community pair are jointly typical:
\begin{align*}
&\mathcal{B}_{\sigma''}:
 (U_{\sigma,\mathcal{C}_i,\mathcal{C}_i}, U'_{\sigma'',\mathcal{C}'_i,\mathcal{C}'_i})\in \mathcal{A}_{\epsilon}^{\frac{n_i(n_i-1)}{2}}(P_{X,X'|{\mathcal{C}}_i, {\mathcal{C}}_i,{\mathcal{C}}'_i, {\mathcal{C}}'_i}),\forall i\in [c],
 \\&
 (G_{\sigma,\mathcal{C}_i,\mathcal{C}_j}, \widetilde{G'}_{\sigma'',\mathcal{C}'_i,\mathcal{C}'_j})\in \mathcal{A}_{\epsilon}^{n_i\cdot n_j}(P_{X,X'|{\mathcal{C}}_i, {\mathcal{C}}_j,{\mathcal{C}}'_i, {\mathcal{C}}'_j}),\forall i,j\in [c], i\neq j
 \},
\end{align*}
Particularly, $\beta_{\sigma'}$ is the event that the sub-matrices are jointly typical under the canonical labeling for the second graph. From standard typicality arguments it follows that:
\begin{align*}
P(\mathcal{B}_{\sigma'})\to 1 \quad \text{as}\quad n\to \infty.
\end{align*}
So, $P(\widehat{\Sigma}_{\mathcal{C}.\mathcal{C}'}=\phi)\to 0$ as $n\to \infty$ since the correct labeling is a member of the set $\widehat{\Sigma}_{\mathcal{C}.\mathcal{C}'}$.
Let $(\lambda_n)_{n\in \mathbb{N}}$ be an arbitrary sequence of numbers such that $\lambda_n= \Theta(n)$. We will show that the probability that a labeling in $\widehat{\Sigma}_{\mathcal{C}.\mathcal{C}'}$ labels $\lambda_n$ vertices incorrectly goes to $0$ as $n\to \infty$.
Define the following:
\begin{align*}
 \mathcal{E}=\{{\sigma'}^2\Big| ||\sigma^2-{\sigma'}^2||_1\geq \lambda_n\},
\end{align*}
where $||\cdot||_1$ is the $L_1$-norm. The set $\mathcal{E}$ is the set of all labelings which match more than $\lambda_n$ vertices incorrectly. \textcolor{black}{We have summarized the notation in Table \ref{tab:J}.}

We show the following:
\begin{align*}
 P(\mathcal{E}\cap \widehat{\Sigma}_{\mathcal{C}.\mathcal{C}'}\neq \phi)\to 0, \qquad \text{as} \qquad n\to \infty.
 \end{align*}
We use the union bound on the set of all permutations along with Corollary \ref{th:cor:1} as follows:
\begin{align*}
  &P(\mathcal{E}\cap \widehat{\Sigma}_{\mathcal{C}.\mathcal{C}'}\neq \phi)
  = P(\bigcup_{{\sigma''}: ||\sigma'-{\sigma''}||_1\geq \lambda_n}\{{\sigma''}\in  \widehat{\Sigma}_{\mathcal{C}.\mathcal{C}'}\})
  \stackrel{(a)}{\leq} \sum_{k=\lambda_n}^{n}\sum_{{\sigma''}: ||\sigma'-{\sigma''}||_1=k} P(\sigma''\in  \widehat{\Sigma}_{\mathcal{C}.\mathcal{C}'})
  \\&\stackrel{(b)}{=} \sum_{k=\lambda_n}^{n}\sum_{{\sigma''}: ||\sigma'-{\sigma'}''||_1=k}
   P(\beta_{\sigma''})
\stackrel{(c)}{\leq} \sum_{k=\lambda_n}^{n}\sum_{{\sigma'}^2: ||\sigma^2-{\sigma'}^2||_0=k} \textcolor{black}{\exp_2\Big(O(nlog{n}))\Big)} \times
\\&   \prod_{i,j\in [c], i< j}
   \textcolor{black}{\exp_2\Big(-\frac{n_i\cdot n_j}{3}(D(P_{X,X'|\mathcal{C}_i,\mathcal{C}_j,\mathcal{C}'_i,\mathcal{C}'_j}
 ||(1-\beta_{i,j})P_{X|\mathcal{C}_i,\mathcal{C}_j}P_{X'|\mathcal{C}'_i,\mathcal{C}'_j}+ \beta_{i,j} P_{X,X'|\mathcal{C}_i,\mathcal{C}_j,\mathcal{C}'_i,\mathcal{C}'_j}))\Big)}
 \\&\times 
  \prod_{i\in [c]}
 \textcolor{black}{\exp_2\Big(-\frac{n_i(n_i-1)}{6}(D(P_{X,X'|\mathcal{C}_i,\mathcal{C}_i,\mathcal{C}'_i,\mathcal{C}'_i}
 ||(1-\beta_i)P_{X|\mathcal{C}_i,\mathcal{C}_i}P_{X'|\mathcal{C}'_i,\mathcal{C}'_i}+ \beta_i P_{X,X'|\mathcal{C}_i,\mathcal{C}_i,\mathcal{C}'_i,\mathcal{C}'_i})\Big)} 
   \\&\stackrel{(d)}{\leq}  \sum_{k=\lambda_n}^{n} {n \choose k}(!k)
 \max_{[\alpha_i]_{i\in [c]}\in \mathcal{A}}\left( \textcolor{black}{\exp_2\Big(-\frac{n^2}{3}(\Phi([\alpha_i]_{i\in [c]})
 +O(\frac{\log{n}}{n}))
 \Big)}\right)
  \\&\leq 
  \max_{\alpha\in [0,1-\frac{\lambda_n}{n}]} 
  \max_{[\alpha_i]_{i\in [c]}}\left( \textcolor{black}{\exp_2\Big(-\frac{n^2}{3}(-(1-\alpha)\frac{\log{n}}{n}+\Phi([\alpha_i]_{i\in [c]})
 +O(\frac{\log{n}}{n}))
 \Big)}\right),
 \end{align*}
where $\mathcal{A}= \{([\alpha_i]_{i\in [c]}): \alpha_i\leq \frac{n_i}{n}, \sum_{i\in [c]}\alpha_i=\frac{n-\lambda_n}{n}\}$ and
\begin{align*}
 &   \Phi([\alpha_i]_{i\in [c]})=
 \sum_{i,j\in [c], i< j}n_in_j \cdot 
 D(P_{X,X'|\mathcal{C}_i,\mathcal{C}_j,\mathcal{C}'_i,\mathcal{C}'_j}
  || (1-\beta_{i,j}) P_{X|\mathcal{C}_i,\mathcal{C}_j} P_{X'|\mathcal{C}'_i,\mathcal{C}'_j}
 + \beta_{i,j} P_{X,X'|\mathcal{C}_i,\mathcal{C}_j,\mathcal{C}'_i,\mathcal{C}'_j})
  \\&+\sum_{i\in [c]}
\frac{n_i(n_i-1)}{2}
 D(P_{X,X'|\mathcal{C}_i,\mathcal{C}_i,\mathcal{C}'_i,\mathcal{C}'_i}
 ||(1-\beta_i)P_{X|\mathcal{C}_i,\mathcal{C}_i}P_{X'|\mathcal{C}'_i,\mathcal{C}'_i}+ \beta_i P_{X,X'|\mathcal{C}_i,\mathcal{C}_i,\mathcal{C}'_i,\mathcal{C}'_i}),
\end{align*}
 and $\beta_{i,j}= \frac{n^2}{n_in_j}\alpha_i\alpha_j$ and $\beta_i= \frac{n\alpha_i(n\alpha_i-1)}{n_i(n_i-1)}$. Here, $\alpha_i$ is the number of fixed points in the $i^{th}$ community divided by $n$, and $\beta_i$ is the number of fixed points in $U'_{\sigma'',\mathcal{C}'_i,\mathcal{C}'_i}$ divided by $\frac{n_i(n_i-1)}{2}$, and $\beta_{i,j}$ is the number of fixed points in $U'_{\sigma'',\mathcal{C}'_i,\mathcal{C}'_j}$ divided by $n_i n_j$. Inequality (a) follows from the union bound, (b) follows from the definition of $ \widehat{\Sigma}_{\mathcal{C}.\mathcal{C}'}$, in (c) we have used Corollary \ref{th:cor:1}, in (d) we have denoted the number of derangement of sequences of length $i$ by $!i$. Note that the right hand side in the (d) goes to 0 as $n\to \infty$ as long as \eqref{eq:th31} holds.
\qedsymbol

'\section{Proof of Theorem \ref{th:7}}
\label{app:th7}
The proof build upon the arguments used in the proof of Theorem \ref{th:2}.
First, note that for the correct labeling the UTs are jointly typical with probability approaching one as $n\to \infty$. So, $P(\widehat{\Sigma}=\phi)\to 0$ as $n\to \infty$ since the correct labeling is a member of the set $\widehat{\Sigma}$.
On the other hand, the probability that a labeling in $\widehat{\Sigma}$ labels $n(1-\alpha_n)$ vertices incorrectly goes to $0$ as $n\to \infty$. 
Define the following:
\begin{align*}
 \mathcal{E}=\{({\sigma'}_{i})_{i\in [m]}\Big| ||({\sigma'}_{i})_{i\in [m]}-({\sigma}_{i})_{i\in [m]}||_0\geq n(1-\alpha_n)\},
\end{align*}
where $||\cdot||_0$ is the $L_0$-norm. Without loss of generality, we assume that the labeling for the denonymized graph is the trivial labeling, i.e. $\sigma_{1}=\sigma'_{1}=id(\cdot)$, where $id(\cdot)$ is the identity function. 
The set $\mathcal{E}$ is the set of all labelings which match more than $n\alpha_n$ vertices incorrectly. We show the following:
\begin{align*}
 P(\mathcal{E}\cap \widehat{\Sigma}\neq \phi)\to 0, \qquad \text{as} \qquad n\to \infty.
 \end{align*}
We partition the set $\mathcal{E}$ into subsets of Bell permutation vectors with the same parameters. We define the following:
\begin{align*}
\mathcal{E}_{i_1,i_2,\cdots, i_{b_m}}= \{({\sigma'}_{i})_{i\in [m]}\Big| ({\sigma'}_{i})_{i\in [m]} \text{ is a $(i_1,i_2,\cdots, i_{b_m})$-Bell permutation vector}\},
\end{align*}
where $i_1,i_2,\cdots, i_{b_m}\in [0,n]$ and $\sum_{k\in [b_m]}i_{\mathcal{P}_{k}}=n$. Then the family of sets $\{\mathcal{E}_{i_1,i_2,\cdots, i_{b_m}}: \sum_{k\in [b_m]}i_{\mathcal{P}_{k}}=n\}$ partitions the set $\mathcal{E}$. 

Note that similar to the proof of Theorem \ref{th:2}, we have:
\begin{align*}
  &P(\mathcal{E}\cap \widehat{\Sigma}\neq \phi)
  = P\left(\bigcup_{{(\sigma'_i)_{i\in [m]}}: ||\sigma^2-{\sigma'}^2||_0\geq n(1-\alpha_n)}\{{\sigma'}^2\in  \widehat{\Sigma}\}\right)
  \\&P(\mathcal{E}\cap \widehat{\Sigma}\neq \phi)
  = P\left(\bigcup_{\substack{(i_1,i_2,\cdots, i_{b_m}):\\ \sum_{k\in [b_m]}i_k=n}}
  \qquad 
  \bigcup_{\substack{
  {(\sigma'_i)_{i\in [m]}}\in \mathcal{E}_{i_1,i_2,\cdots, i_{b_m}}:\\ ||\sigma^2-{\sigma'}^2||_0\geq n(1-\alpha_n)}}
  \{{\sigma'}^2\in  \widehat{\Sigma}\}\right)
  \\&\leq 
  \sum_{\substack{(i_1,i_2,\cdots, i_{b_m}):\\ \sum_{k\in [b_m]}i_k=n}}
  \qquad 
  \sum_{\substack{
  {(\sigma'_i)_{i\in [m]}}\in \mathcal{E}_{i_1,i_2,\cdots, i_{b_m}}:\\ ||\sigma^2-{\sigma'}^2||_0\geq n(1-\alpha_n)}}
 P\left(\{{\sigma'}^2\in  \widehat{\Sigma}\}\right)
 \\&\leq  \sum_{\substack{(i_1,i_2,\cdots, i_{b_m}):\\ \sum_{k\in [b_m]}i_k=n}}
  \qquad 
  \sum_{\substack{
  {(\sigma'_i)_{i\in [m]}}\in \mathcal{E}_{(i_1,i_2,\cdots, i_{b_m})}:\\ ||\sigma^2-{\sigma'}^2||_0\geq n(1-\alpha_n)}}
  \textcolor{black}{\exp_2\Big(-\frac{\frac{n(n-1)}{2}}{(m(m-1)+1)(b_m-1)}(D(P_{X^m}
 ||\sum_{k\in [b_m]}\frac{i'_k}{\frac{n(n-1)}{2}}
 P_{X_{\mathcal{P}_k}})-}
 \\&
 \textcolor{black}{\epsilon\prod_{j\in [m]}|\mathcal{X}_j|+O(\frac{\log{n}}{n^2}))\Big)}
 \\&= \sum_{\substack{(i_1,i_2,\cdots, i_{b_m}):\\ \sum_{k\in [b_m]}i_k=n, i_{b_m}\geq n(1-\alpha_n)}}
 N_{i_1,i_2,\cdots,i_{b_m}}
 \textcolor{black}{\exp_2\Big(-\frac{\frac{n(n-1)}{2}}{(m(m-1)+1)(b_m-1)}(D(P_{X^m}
 ||\sum_{k\in [b_m]}\frac{i'_k}{\frac{n(n-1)}{2}}
 P_{X_{\mathcal{P}_k}})-}
 \\&
 \textcolor{black}{\epsilon\prod_{j\in [m]}|\mathcal{X}_j|+O(\frac{\log{n}}{n^2}))\Big)}
 \\&\leq
 \sum_{\substack{(i_1,i_2,\cdots, i_{b_m}):\\\sum_{k\in [b_m]}i_k=n, i_{b_m}\geq n(1-\alpha_n)}}
 \textcolor{black}{\exp_2\Big(n\log{n}(\sum_{k\in [b_m]}|\mathcal{P}_k|\frac{i_k}{n}-1)+O(n\log{n})\Big)}\times
 \\&\textcolor{black}{\exp_2\Big(-\frac{\frac{n(n-1)}{2}}{(m(m-1)+1)(b_m-1)}(D(P_{X^m}
 ||\sum_{k\in [b_m]}\frac{i'_k}{\frac{n(n-1)}{2}}
 P_{X_{\mathcal{P}_k}})-\epsilon\prod_{j\in [m]}|\mathcal{X}_j|+O(\frac{\log{n}}{n^2}))\Big)}
  \\& 
  \leq \sum_{\substack{(i_1,i_2,\cdots, i_{b_m}):\\\sum_{k\in [b_m]}i_k=n, i_{b_m}\geq n(1-\alpha_n)}}
 \textcolor{black}{\exp_2\Big(n\log{n}(\sum_{k\in [b_m]}|\mathcal{P}_k|\frac{i_k}{n}-1)+O( n\log{n})
 -}
 \\&
 \textcolor{black}{\frac{\frac{n(n-1)}{2}}{(m(m-1)+1)(b_m-1)}(D(P_{X^m}
 ||\sum_{k\in [b_m]}\frac{i'_k}{\frac{n(n-1)}{2}}
 P_{X_{\mathcal{P}_k}})-\epsilon\prod_{j\in [m]}|\mathcal{X}_j|+O(\frac{\log{n}}{n^2}))\Big)},
\end{align*}
where $i'_k= \frac{i_k(i_k-1)}{2}+ \sum_{k',k'':\mathcal{P}_{k',k''}=\mathcal{P}_l}i_{k'}i_{k''}$, and $\mathcal{P}_{k',k''}= \{\mathcal{A}'\cap \mathcal{A}'': \mathcal{A}'\in \mathcal{P}_{k'}, \mathcal{A}''\in \mathcal{P}_{k''}\}, k',k''\in [b_m]$.
 Note that the right hand side in the last inequality approaches 0 as $n\to \infty$ as long as:
\begin{align*}
    &n\log{n}(\sum_{k\in [b_m]}|\mathcal{P}_k|\alpha_k-1)\leq \frac{n(n-1)}{2(b_m-1)(m(m-1)+1)}(D(P_{X^m}
 ||\sum_{k\in [b_m]}\alpha'_k
 P_{X_{\mathcal{P}_k}})-\epsilon\prod_{j\in [m]}|\mathcal{X}_j|+O(\frac{\log{n}}{n^2}))
 \\& \!\leftrightarrow 
 \log{n}(\sum_{k\in [b_m]}|\mathcal{P}_k|\alpha_k-1)\leq \frac{(n-1)}{2(b_m-1)(m(m-1)+1)}(D(P_{X^m}
 ||\!\sum_{k\in [b_m]}\alpha'_k
 P_{X_{\mathcal{P}_k}})-\!\epsilon\!\prod_{j\in [m]}|\mathcal{X}_j|+O(\frac{\log{n}}{n^2})), 
\end{align*}
for all $\alpha_1,\alpha_2,\cdots, \alpha_{b_m}:\sum_{k\in [b_m]}\alpha_k=n, \alpha_{b_m} \in [1,1-\alpha_n]$,
where we have defined $\alpha'_k=\frac{i'_k}{\frac{n(n-1)}{2}}$. The last equation is satisfied by the theorem assumption for small enough $\epsilon$. 
\qedsymbol

\section{Proof of Theorem \ref{th:converse}}
\label{app:conv}

Let $n\in \mathcal{N}$, and $g$ and $g'$ be the adjacency matrices of the two graphs under a pre-defined labeling. Let $\hat{\sigma}$ be the output of the matching algorithm. Let $\mathbbm{1}_{C}$ be the indicator of the event that the matching algorithm mislabels at most $\epsilon_n$  fraction of the vertices with probability at least $P_e$, where $\epsilon_n, P_e\to 0$ as $n\to \infty$.  Note that $\hat{\sigma}$ is a function of $\sigma',g,g'$. So:

\begin{align*}
    &0=H(\hat{\sigma}|\sigma,g,g')
    \stackrel{(a)}{=} 
    H(\sigma',\hat{\sigma},\mathbbm{1}_C|\sigma,g,g')- H(\sigma', \mathbbm{1}_C| \hat{\sigma},\sigma,g,g')
    =  H(\sigma',\hat{\sigma},\mathbbm{1}_C|\sigma,g,g')-
    \\&
    H(\sigma' |\mathbbm{1}_C, \hat{\sigma},\sigma,g,g')
    - H( \mathbbm{1}_C| \hat{\sigma},\sigma,g,g')
    \stackrel{(b)}{\geq}
      H(\sigma',\hat{\sigma},\mathbbm{1}_C|\sigma,g,g')-
    H(\sigma' |\mathbbm{1}_C, \hat{\sigma},\sigma,g,g')-1=
    \\&
     H(\sigma',\hat{\sigma},\mathbbm{1}_C|\sigma,g,g')-
    P(\mathbbm{1}_C=1) H(\sigma' |\mathbbm{1}_C=1, \hat{\sigma},\sigma,g,g')-
    P(\mathbbm{1}_C=0) H(\sigma' |\mathbbm{1}_C=0, \hat{\sigma},\sigma,g,g')\!-\!1
    \\&\stackrel{(c)}{\geq} 
    H(\sigma',\hat{\sigma},\mathbbm{1}_C|\sigma,g,g')-
    \epsilon_n n\log{n}-
    P_e n\log{n}-1
    \stackrel{(d)}{\geq}  H(\sigma'|\sigma,g,g')- (\epsilon_n+P_e)n\log{n}-1,
\end{align*}
where in (a) we have used the chain rule of entropy, in (b) we have used the fact that $\mathbbm{1}_C$ is binary, in (c) we define the probability of mismatching more than $\epsilon_n$ fraction of the vertices by $P_e$, and (d) follows from the fact that entropy is non-negative.
 As a result,
 \begin{align*}
     H(\sigma'|\sigma,g,g')\leq 
     (\epsilon_n+P_e)n\log{n}+1= o(n\log{n)},
 \end{align*}
{where $(\epsilon_n+P_e)n\log{n}$ is $o(n\log{n})$ since $\epsilon_n, P_e$ go to $0$ as $n\to\infty$.}
 Consequently,
% \begin{align*}
%     \frac{\log{n}}{n}=H(\sigma)\geq I(\sigma;\sigma')= H(\sigma)-H(\sigma|\sigma')\geq n\log{n}- \epsilon n\log{n}.
% \end{align*}
\begin{align*}
    n\log{n}\stackrel{(a)}{=} \log{n!}+n+O(\log{n})= H(\mathbf{\sigma'})+O(n)
    \stackrel{(b)}{=} I(\mathbf{\sigma}';\mathbf{\sigma}, g, g')+o(n\log{n}),
\end{align*}
where in (a) we have used Stirling's approximation, and in (b) we have used the fact that $\epsilon, P_e\to 0$ as $n\to\infty$. We have:
\begin{align*}
&n\log{n} =
I(\mathbf{\sigma}';\mathbf{\sigma}, g, g')+o(n\log{n})\\
& = I(\mathbf{\sigma}'; g')+
I(\mathbf{\sigma}';\mathbf{\sigma}, g |g')+o(n\log{n})
\stackrel{(a)}{=}
I(\mathbf{\sigma}';\mathbf{\sigma}, g |g')+o(n\log{n})
\\&= I(\mathbf{\sigma}'; g|g')
+I(\mathbf{\sigma}'; g |g',\mathbf{\sigma})+o(n\log{n})
\stackrel{(b)}{=}
I(\mathbf{\sigma}'; g |g',\mathbf{\sigma})+o(n\log{n})
\\& \stackrel{(c)}{\leq}
I(\mathbf{\sigma}',g'; g |\mathbf{\sigma})+o(n\log{n})
 \stackrel{(d)}{=}
I(g'; g |\mathbf{\sigma}, \mathbf{\sigma}')+o(n\log{n})
\\&\stackrel{(e)}{=} \sum_{i,j \in [c], i<j}n_in_j I(X,X'|\mathcal{C}_i,\mathcal{C}_j, \mathcal{C}'_i\mathcal{C}'_j)
+  \sum_{i \in [c]}\frac{n_i(n_i-1)}{2} I(X,X'|\mathcal{C}_i,\mathcal{C}_i, \mathcal{C}'_i,\mathcal{C}'_i)+o(n\log{n}),
\end{align*}
where (a) follows from $\sigma'\indep g'$, (b) follows from the fact that $\sigma'\indep g,g'$, (c) is true due to the non-negativity of the mutual inforamtion, (d) follows from $\sigma,\sigma'\indep G$, and (e) follows from the fact that the edges whose vertices have different labels are independent of each other given the labels.

\qedsymbol

\section{Proof of Lemma \ref{lem:card}}
\label{app:card}
The ambiguity set $\mathcal{L}$ is defined as:
\begin{align*}
 \mathcal{L}=\{v_{j}\big| \nexists! i: (\underline{F}_{i}, \underline{F}'_{j})\in \mathcal{A}_{\epsilon}^n(X,X')\}.
\end{align*}
From the Chebychev inequality, we have:
 \begin{align}
 P\left(|\mathcal{L}|>2\mathbb{E}(|\mathcal{L}|)\right)=  P\left(\big||\mathcal{L}|-\mathbb{E}(|\mathcal{L}|)\big|>\mathbb{E}(|\mathcal{L}|)\right)
 \leq \frac{Var(|\mathcal{L}|)}{\mathbb{E}^2(|\mathcal{L}|)}.
 \label{eq:cheb}
\end{align}

 Let $\mathcal{B}_{j}$ be the event that $v_{j}\in \mathcal{L}$, then
\begin{align*}
&\mathbb{E}(|\mathcal{L}|)=
\mathbb{E}\left(\sum_{j=1}^n\mathbbm{1}(v_{j}\in \mathcal{L})\right)=
\sum_{j=1}^nP(v{j}\in \mathcal{L})=\sum_{j=1}^n P(\mathcal{B}_j)\\
 &Var(|\mathcal{L}|)= \sum_{j=1}^n P(\mathcal{B}_j)+ \sum_{i\neq j} P(\mathcal{B}_i,\mathcal{B}_j)- \left(\sum_{j=1}^n P(\mathcal{B}_j)\right)^2\\
 &= \sum_{j=1}^n P(\mathcal{B}_j)- \sum_{j=1}^n P^2(\mathcal{B}_j)\leq \mathbb{E}(|\mathcal{L}|).
\end{align*}
So, from \eqref{eq:cheb}, we have:
\begin{align*}
 P\left(|\mathcal{L}|>2\mathbb{E}(|\mathcal{L}|)\right)\leq \frac{1}{\mathbb{E}(|\mathcal{L}|)},
\end{align*}
which goes to 0 as $n\to \infty$ provided that $\mathbb{E}(|\mathcal{L}|) \to \infty$ (otherwise the claim is proved since $\mathbb{E}(|\mathcal{L}|)$ is finite.).
It remains to find an upper bound on $\mathbb{E}(|\mathcal{L}|)$. Let $\mathcal{C}_j$ be the event that the fingerprint $\underline{F}'_{j}$ is not typical with respect to $P_{X'}$ and let $\mathcal{D}_{i,j}$ be the event that there exists $i\in [1,n]$ such that $\underline{F}_{i}$ and $\underline{F}'_{j}$ are jointly typical with respect to $P_{X,X'}$. Then,
\begin{align*}
 P(B_j)&\leq P\left(\mathcal{C}_j\bigcup \left( \bigcup_{i\neq j} \mathcal{D}_{i,j}\right)\right)\leq P(\mathcal{C}_j)+\sum_{i\neq j} P(\mathcal{D}_{i,j}|\mathcal{C}_j^c)
 \stackrel{(a)}\leq \frac{1}{\Lambda\epsilon^2}+ n2^{-\Lambda(I(X;X')-\epsilon)},
\end{align*}
where (a) follows from the standard information theoretic arguments (e.g proof of Theorem 3 in \cite{Allerton}). So,
\begin{align*}
 \mathbb{E}(|\mathcal{L}|)\leq \frac{n}{\Lambda\epsilon^2}+ n^22^{-\Lambda(I(X;X')-\epsilon)}.
\end{align*}
 From $\Lambda>\frac{2\log{n}}{I(X;X')}$, we conclude that the second term approaches 0 as $n\to \infty$. This completes the proof.
\qedsymbol

\section{Proof of Theorem \ref{th:seeded}}
\label{app:thseeded}
Let $H_1$ be the event the algorithm fails and $H_2$ the event that $|\mathcal{L}|> \frac{2n}{\Lambda\epsilon^2}$. Then:\[
 P(H_1)\leq P(H_2)+P(H_1|H_2^c).\] From Claim 1, we know that $P(H_2)\to 0$ as $n\to \infty$. For the second term, let $\mathcal{L}'$ be the set of vertices which are not matched in the second iteration. The algorithm fails if  $\mathcal{L}'\neq \phi$. However, by a similar argument as in the proof of claim 1, we have:
 \begin{align*}
 P(|\mathcal{L}'|>\frac{1}{2} \Big| |\mathcal{L}|<\frac{2n}{\Lambda\epsilon^2})\to 0 \text{ as } n\to\infty.
\end{align*}
So,  $P(|\mathcal{L}'|=0) \to 1$ as $n\to \infty$. This completes the proof.

\qedsymbol

\end{appendices}

%{This is due to the fact that such encoding functions are ineffective in preserving correlation. The loss of correlation undermines the encoders' ability to collaborate to take advantage of the multi-terminal nature of the problem.  In PtP communication problems, where there is only one transmitter, the necessity for cooperation does not manifest itself. For this reason, although encoders with asymptotically large block-lengths are optimal in PtP communications, they are sub-optimal in the network communication case. When transmitting data over networks, there is a trade-off between the sender's need to transmit in a PtP optimal manner, and the networks' requirements for coordination among the transmitters. This results in the so-called `sweet-spot' for the length of an encoder. In this section, we show that to achieve a fixed correlation between the outputs of the encoders,  the  {effective-length} of the encoding functions is bounded from above. Alternatively, we derive a bound on the correlation between the outputs of two arbitrary encoding functions given their effective-lengths.
%}


%\startbibliography

 \end{document}